\documentclass[12pt]{article}
\usepackage{amsmath}
\usepackage{amsthm}

\usepackage{graphicx}
\usepackage[round]{natbib}
\usepackage{url} 
\usepackage[T1]{fontenc}
\usepackage[latin1]{inputenc} 

\makeatletter
\g@addto@macro\th@plain{\thm@preskip=2pt \thm@postskip=2pt}     
\g@addto@macro\th@definition{\thm@preskip=2pt \thm@postskip=2pt}  
\g@addto@macro\th@remark{\thm@preskip=2pt \thm@postskip=2pt}      
\makeatother

\newtheorem{theorem}{Theorem}

\newtheorem{lemma}{Lemma}

\newtheorem{remark}{Remark}

\newtheorem{definition}{Definition}

\usepackage{amssymb}
\usepackage{bm}
\usepackage{algorithmic}
\usepackage{booktabs}
\usepackage{amsmath}
\usepackage{algorithm}
\usepackage{algorithmic}
    
\usepackage{color}
\usepackage{algorithm}
\usepackage{verbatim}
\usepackage{algorithmic}
\usepackage{makecell}
\usepackage{enumerate}
\usepackage{graphicx}
\usepackage{enumitem}
\usepackage{threeparttable}
\usepackage{multirow}
\usepackage{mathrsfs}
\usepackage{hyperref}
\usepackage{subfigure}
\usepackage{setspace}
\usepackage{fontawesome}
\usepackage[normalem]{ulem}
\usepackage{xr-hyper}
\usepackage{titlesec}
\makeatletter
\newcommand*{\addFileDependency}[1]{
  \typeout{(#1)}
  \@addtofilelist{#1}
  \IfFileExists{#1}{}{\typeout{No file #1.}}
}
\makeatother

\newcommand*{\myexternaldocument}[1]{
    \externaldocument{#1}
    \addFileDependency{#1.tex}
    \addFileDependency{#1.aux}
}

\myexternaldocument{TransNet_SM_revision}

 \hypersetup{
     colorlinks=true,
     linkcolor=blue,
     filecolor=blue,
     citecolor = darkgreen,
     urlcolor=cyan,
     }

\newcommand{\blind}{0}
\definecolor{darkgreen}{rgb}{0,0.6,0.2}

\titlespacing*{\section}    {0pt}{*0.8}{*0.4} 
\titlespacing*{\subsection} {0pt}{*0.6}{*0.3}
\titlespacing*{\subsubsection}{0pt}{*0.5}{*0.2}
\begin{document}

\def\spacingset#1{\renewcommand{\baselinestretch}%
{#1}\small\normalsize} \spacingset{1}


\if0\blind
{
  \title{\LARGE \bf Multilayer-Dynamic Network Clustering with Application to World Trade Data}
  \author{ Mengze Huang$^1$, Wenqing Su$^2$, Shisong Wu$^3$, Xiao Guo$^1$, Hai Zhang$^1$   \hspace{3cm}  \\
    $1$ School of Mathematics, Northwest University, 710127,  Shaanxi, China   \hspace{3.5cm} \\
    $2$ School of Mathematics and Statistics, Shaanxi Normal University, \hspace{3cm} \\ 710119, Shaanxi, China  \hspace{3cm} \\
    $3$ China Southern Power Grid Artificial Intelligence Technology Co., Ltd., \hspace{3cm}  \\ Guangzhou, China  \hspace{3cm} }
  \maketitle
} \fi

\if1\blind
{
 \bigskip
 \bigskip
 \bigskip
 \begin{center}
   {\LARGE\bf Multilayer-Dynamic Network Clustering with Application to World Trade Data}
\end{center}
 \medskip
} \fi

\bigskip
\begin{abstract}
International trade data, such as the FAO dataset, can be naturally
represented as \emph{multilayer-dynamic networks}, where countries are nodes,
trade relationships are edges, different products correspond to different
layers, and networks evolve over time. An important problem is how to identify
evolving community structures in such multilayer-dynamic trade networks. 
Motivated by this problem, we study community detection in multilayer-dynamic networks, allowing the community structure to vary across both layers and time. We propose a novel method, \emph{MuDySC} (Multilayer-Dynamic Spectral Clustering), which smooths the eigenspace projection matrices across adjacent time points and across layers at the same time point. We develop an efficient alternating iterative algorithm and establish both
global and local convergence results, with the latter allowing
weaker conditions on the tuning parameters when the relevant eigenspaces are
sufficiently close and the algorithm is suitably initialized.
We apply MuDySC to the FAO data. The analysis reveals clear asymmetry between export and import community structures and highlights both persistent and shifting trade positions of major countries. As an extension to accommodate substantial heterogeneity across network layers,
we develop TLC-MuDySC, a tensor-based layer-clustering method that first
identifies structurally related layers and then applies MuDySC within the
estimated layer groups.
\end{abstract}

\vfill

\newpage
\spacingset{1.7} 
\setlength{\abovedisplayskip}{1.8pt}  
\setlength{\belowdisplayskip}{1.8pt}

\section{Introduction}\label{sec1}
With the continuous advancement of global economic integration, international trade has become an important driving force for economic growth across countries. It not only improves the allocation of global resources, but also promotes industrial division of labor and technological diffusion. International trade exhibits high levels of dynamism and interdependence, with trade flows between countries driven by multiple factors such as market supply and demand, and macroeconomic policies. Therefore, studying the evolutionary patterns of international trade can provide valuable insights into the changing global economic landscape.

FAOSTAT, provided by the Food and Agriculture Organization of the United Nations (FAO), is a valuable source of international agricultural trade data. 
In particular, the bilateral trade flows of different food and agricultural products between countries over time are recorded, covering 573 products across 245 countries or regions from 1986 to 2023 ~\footnote{\url{https://www.fao.org/faostat/en/\#data/TM}}. For convenience, we hereafter refer to this dataset as the FAO dataset. The FAO dataset can be naturally represented as a \emph{multilayer-dynamic network} \citep{mucha2010community,boccaletti2014structure}, where countries are treated as nodes, trade flows between countries are represented as edges, different products correspond to different layers, and different years characterize the temporal evolution of the network.

Community detection or clustering is a fundamental problem in network analysis that aims to identify communities of nodes that are more densely connected or more similar to one another than to the rest of the network. Over the past decades, community detection methods for single networks have been extensively studied and developed, including modularity maximization, spectral clustering, likelihood-based methods, and semidefinite programming; see \citet{abbe2018community} for a survey. However, single networks fail to capture the temporal and multilayer features of real networks. Building on clustering methods for single networks, increasing attention has been devoted to clustering in multilayer or dynamic networks. On the one hand, community detection in static multilayer networks, namely, networks with multiple layers sharing a common set of nodes, has been widely studied by  \citet{han2015consistent,paul2016consistent,paul2020spectral,arroyo2021inference,macdonald2022latent,lei2023bias,huang2023spectral,zhang2024consistent,agterberg2025joint,wu2025privacy}, among others. These works achieve improved clustering accuracy over methods based on single networks by assuming a common community structure across layers. Relatedly, ~\citet{gao2026krafty} proposed KRAFTY for recovering
joint clustering structures from multiple network views.
We also note that, beyond the common community assumption, some works allow more heterogeneous structures across layers; for example, \citet{pensky2024clustering} considered diverse multiplex networks in which different groups of layers may have different common subspaces. On the other hand, community detection in dynamic networks, namely, networks that evolve over time on a common set of nodes, has been studied by \citet{aynaud2013communities,xu2014dynamic,liu2018global,zhang2017finding,pensky2019spectral,zhang2024fast,lin2026dynamic}; see also references therein. 
In these works, the community structure is allowed to be heterogeneous and to vary over time.

The aforementioned works focus on either static multilayer networks or single-layer dynamic networks, and thus cannot simultaneously capture the temporal and multilayer structure of complex networks.  Recently, only a few works have paid attention to the analysis of multilayer-dynamic networks. \citet{loyal2023eigenmodel} proposed a latent space model for multilayer-dynamic networks, which captures the common time-varying structure shared across layers while accommodating layer-specific variation and degree heterogeneity. \citet{zheng2024dynamic} considered the clustering of layers in the  multiple dynamic networks, where layers are grouped according to their common evolutionary patterns over time. 
\citet{casarin2025dynamic} proposed a Bayesian dynamic stochastic block model for multidimensional networks , providing a probabilistic framework for modeling evolving community structures.
\citet{wang2026changepoint} investigated change-point localization and inference in multilayer-dynamic networks. However, the focus of these existing works is different from ours, and the problem of node-level clustering in multilayer-dynamic networks remains largely underexplored.

Motivated by the FAO dataset, we study community detection in multilayer-dynamic networks, allowing the community structure to vary across layers and over time. By tracking the evolution of these communities, we can gain insight into the underlying organization and development of the international trade system. Our main methodological and theoretical contributions are as follows.
\begin{itemize}
\item First, we propose a new spectral clustering method, referred to as
\textbf{Mu}ltilayer-\textbf{Dy}namic \textbf{S}pectral
\textbf{C}lustering (MuDySC), for community detection in
multilayer-dynamic networks. MuDySC jointly smooths eigenspace projection
matrices across adjacent time points and across network layers, while
allowing the community structure to vary across both layers and time.
It therefore simultaneously incorporates temporal and layer-wise
information. The dynamic spectral clustering method of
\cite{liu2018global} is recovered as a special case when $M=1$ and
$\beta=0$. We further provide a network cross-validation procedure for
selecting the two tuning parameters that control the strengths of temporal
and layer-wise smoothing.

\item  Second, we develop an efficient alternating iterative algorithm for solving
the resulting optimization problem. We establish a global convergence result
showing that, under a sufficient condition on the tuning parameters, the
algorithm converges to the global minimum. We further establish a local
convergence result showing that, when the eigenspaces associated with
neighboring time points and different network layers are sufficiently close,
the algorithm can converge under substantially weaker conditions on the
tuning parameters.

\item Third, to further accommodate substantial heterogeneity across network
layers, we develop a tensor-based layer-clustering extension, referred to as
TLC-MuDySC. The method first identifies groups of structurally related layers
through a tensor-based spectral procedure and then applies MuDySC separately
within the estimated layer groups, so that information is primarily borrowed
among related networks. The details of this extension are provided in
Appendix~\ref{sec:tensor-extension}.

\end{itemize}

To demonstrate the practical utility of MuDySC, we apply it to the FAO dataset. Specifically, we study the import and export relationships of 15 vegetable oil products across 131 countries or regions over the five-year period from 2019 to 2023, yielding a multilayer-dynamic trade network with 15 layers, 5 time points, and 131 nodes. In particular, we analyze the evolution of communities in the olive-oil trade network and obtain the following main findings. First, our analysis reveals a pronounced asymmetry between the community structures of the export and import olive-oil trade networks. Specifically, the export network exhibits a more pronounced community structure, with denser within-community connections, whereas the import network shows a weaker community structure. This finding is consistent with the structure of the global olive-oil market, in which supply is concentrated in a small number of producing countries, while import demand is distributed across a much broader set of countries.
Second, our analysis identifies several stable community patterns over the study period. In both the export and import olive-oil trade networks, China remains in the same community as Japan and Australia throughout 2019-2023, which may be related to their common position within the Asia-Pacific economic region. By contrast, China, a major olive-oil importer, is never grouped into the same community as Greece, a typical olive-oil producing country, reflecting their persistently different roles in the trade network. Third, our analysis also identifies structural changes in the community memberships of individual countries. For example, Russia's community membership in the export trade network changes markedly after 2021. From 2019 to 2021, Russia belongs to the same community as its European neighboring countries. In 2022, however, it separates from that community and is assigned to a community that also includes China. By 2023, its community membership changes again, with Russia grouped together with only a small number of neighboring countries, thereby forming a relatively small and distinct community. This pattern may be associated with the reorganization of trade relationships following the Russia-Ukraine conflict in 2022.

The remainder of this paper is organized as follows. Section \ref{sec2} provides the optimization framework of MuDySC method, as well as the algorithm and convergence analysis. Section \ref{sec4} applies MuDySC to the FAO dataset. Section \ref{sec3} presents the numerical experiments. Section \ref{sec5} concludes the paper. Appendices \ref{classification_table} and \ref{sec7} cover the FAO data analysis and its supplementary results. Appendix \ref{parameter_sensitivity} presents the parameter sensitivity analysis. Appendices \ref{sec8} and \ref{sec9} contain the proofs of the theorems. Appendix \ref{sec:tensor-extension} introduces TLC-MuDySC, an extension of MuDySC.

\section{Spectral clustering for multilayer-dynamic networks}\label{sec2}

In this section, we provide the optimization framework of MuDySC method for clustering multilayer-dynamic networks, as well as the algorithm, tuning parameter selection, and convergence analysis.

\subsection{Method}\label{sec2.1}

Suppose that a multilayer-dynamic network consists of \(M\) layers, \(T\) time points, and \(n\) aligned nodes, with adjacency matrices \(\{A_{m,t}\}_{m=1,t=1}^{M,T}\). Specifically, $A_{m,t}\in \{0,1\}^{n\times n}$ denotes the adjacency matrix of the network at layer $m$ and time stamp $t$. The Laplacian matrix of $A_{m,t}$ is denoted by $L_{m,t}$ as follows
\begin{equation*}
L_{m,t} = D^{-1/2} \, A_{m,t} \, D^{-1/2}.
\end{equation*}
Let $V_{m,t}\in\mathbb R^{n\times K}$ be the matrix of the top-$K$ eigenvectors of $L_{m,t}$ and define the projection matrix $U_{m,t}: = V_{m,t} V_{m,t}^{T}.$

To jointly incorporate the smoothness across layers and time points, we aim to obtain the smoothed projection matrices \(\{\overline{U}_{m,t}\}_{m=1,t=1}^{M,T}\) by minimizing

\begin{equation}
\begin{aligned}
&\min_{\overline{U}_{m,t}} \ 
\sum_{m = 1}^{M}\sum_{t = 1}^{T}\|U_{m,t}-\overline{U}_{m,t}\|_F^2 + \alpha\sum_{m = 1}^{M}\sum_{t = 1}^{T - 1}\|\overline{U}_{m,t}-\overline{U}_{m,t + 1}\|_F^2 
+\beta\cdot\frac{1}{M-1}\sum_{t = 1}^{T}\sum_{i<j}\|\overline{U}_{i,t}-\overline{U}_{j,t}\|_F^2\\
 &{\rm subject\; to}\quad \overline{U}_{m,t} \in \left\{ VV^{T} : V \in \mathbb{R}^{n \times K},\; V^{T}V = I_K \right\}\quad {\rm for}\quad m=1,...,M; \; t=1,...,T,
\end{aligned}
\label{eq6}
\end{equation}
where $\alpha>0$ and $\beta>0$ are the tuning parameters, \(\|\cdot\|_{F}\) denotes the Frobenius norm, and $I_K$ denotes the identity matrix of dimension $K$. The three terms in \eqref{eq6} can be explained as follows. The first term enforces that the smoothed projection matrices are close to the original ones. The second term ensures that for each given layer $m$, the smoothed projection matrices are close between adjacent time points $t$ and $t+1$. The third term encourages that for each given time point $t$, the smoothed projection matrices are close across layers, where the multiplicative factor $1/(M-1)$ is introduced only for rescaling, so that the third term is generally balanced with the second term. 

After obtaining the smoothed projection matrices $\overline{U}_{m,t}$'s, we can  extract their top-$K$ eigenvectors and apply $K$-means to obtain the $K$ clusters for each static network at layer \(m\in \{1,...,M\}\) and time \(t\in \{1,...,T\}\). We name this procedure MuDySC, short for \textbf{Mu}ltilayer-\textbf{Dy}namic \textbf{S}pectral \textbf{C}lustering.

The optimization problem in \eqref{eq6} can be solved by an alternating iterative algorithm. The basic idea is to update one smoothed projection matrix \(\overline{U}_{m,t}\) at a time while keeping all the others fixed at their current values. Under this blockwise updating scheme, the objective function involving \(\overline{U}_{m,t}\) depends only on the original projection matrix \(U_{m,t}\), the smoothed projection matrices at adjacent time points in the same layer, and the smoothed projection matrices from other layers at the same time point. Therefore, each update step amounts to solving a low-dimensional subproblem for a single pair \((m,t)\), which can be carried out efficiently. Repeating this procedure over all \(m=1,\ldots,M\) and \(t=1,\ldots,T\) yields an alternating iterative algorithm, which is continued until convergence. Specifically, denote the current estimates by \(\{\overline{U}_{m,t}^{\,l}\}\). Then for different $t$, we can update \(\{\overline{U}_{m,t}^{\,l}\}\) to \(\{\overline{U}_{m,t}^{\,l+1}\}\) according to the following rules:
\begin{equation}
\begin{aligned}
\overline{U}_{m,1}^{\,l+1}
&=
\Pi_{K}\Big(
U_{m,1}
+\alpha \overline{U}_{m,2}^{\,l}
+\frac{\beta}{M-1}\sum_{i\neq m}\overline{U}_{i,1}^{\,l}
\Big), \quad t=1, \\
\overline{U}_{m,t}^{\,l+1}
&=
\Pi_{K}\Big(
U_{m,t}
+\alpha \overline{U}_{m,t-1}^{\,l}
+\alpha \overline{U}_{m,t+1}^{\,l}
+\frac{\beta}{M-1}\sum_{i\neq m}\overline{U}_{i,t}^{\,l}
\Big), \quad t=2,\ldots,T-1, \\
\overline{U}_{m,T}^{\,l+1}
&=
\Pi_{K}\Big(
U_{m,T}
+\alpha \overline{U}_{m,T-1}^{\,l}
+\frac{\beta}{M-1}\sum_{i\neq m}\overline{U}_{i,T}^{\,l}
\Big),\quad t=T,
\end{aligned}
\label{eq7}
\end{equation}
where \(\Pi_{K}(H):=V_K V_K^\top\) provided that $V_K$ consists of the top-$K$ eigenvectors of $H$. 

\begin{remark}
The current MuDySC formulation assumes a common number of communities \(K\) across layers and time points. Allowing the number of communities to vary is not a direct extension, since projection matrices with different ranks have an intrinsic nonzero discrepancy under the Frobenius distance used in \eqref{eq6}. One possible direction is to decompose each eigenspace into a lower-dimensional component that is comparable across networks and additional layer- or time-specific directions, and smooth only the comparable components. We leave the development of such an extension to future work.
\end{remark}

\begin{remark}
The degree of temporal and layer-wise information sharing in MuDySC can be
adjusted through the tuning parameters $\alpha$ and $\beta$, whose data-driven
selection is discussed in Section~\ref{sec2.2}. Thus, appropriate choices of
$\alpha$ and $\beta$ can partially mitigate the adverse effects of
heterogeneity. When substantial heterogeneity exists across network layers,
however, a single global layer-wise smoothing parameter may be insufficient.
To address this issue, we develop a tensor-based layer-clustering extension,
TLC-MuDySC, which first identifies groups of structurally related layers and
then applies MuDySC separately within each group; see
Appendix~\ref{sec:tensor-extension}.
\end{remark}

\subsection{Tuning parameter selection}\label{sec2.2}

The optimization problem~\eqref{eq6} involves two tuning parameters
$\alpha$ and $\beta$, which control the strengths of temporal and layer-wise
smoothing, respectively. Generally, when the community structures are more
heterogeneous across time points or network layers, smaller values of
$\alpha$ or $\beta$ are preferred to avoid excessive smoothing. Therefore,
by appropriately selecting these tuning parameters, MuDySC can adapt to the
unknown degree of temporal and layer-wise heterogeneity in a data-driven
manner. We select their values via a network cross-validation procedure,
following the idea of \citet{chen2018network}, which was initially proposed
to select the number of communities.

Following \citet{chen2018network}, we adopt a stochastic block model (SBM)-type working model \citep{holland1983stochastic} to evaluate the predictive performance of the estimated community structures. Specifically, for each layer $m$ and time point $t$, let
$B^{(m,t)}\in[0,1]^{K\times K}$ denote the community connectivity matrix,
and let $z_{m,t}\in\{1,\ldots,K\}^n$ denote the community membership vector,
where $[z_{m,t}]_i=k$ indicates that node $i$ belongs to community $k$.
Conditional on $z_{m,t}$, we assume
\[
[A_{m,t}]_{i,j}
\sim
\mathrm{Bernoulli}
\left(
[B^{(m,t)}]_{[z_{m,t}]_i,[z_{m,t}]_j}
\right),
\qquad
[A_{m,t}]_{j,i}=[A_{m,t}]_{i,j},
\qquad j>i.
\]

Let $\mathcal{A}$ and $\mathcal{B}$ be the sets of candidate values for $\alpha$ and $\beta$, and let the number of folds be $V=5$. The procedure is presented as follows.

\begin{enumerate}
\item Randomly split the $n$ nodes into $V$ equal-sized subsets
$\{\tilde{\mathcal N}_v: 1\le v\le V\}$.
For each layer $m=1,\ldots,M$ and time point $t=1,\ldots,T$, split the adjacency matrix $A_{m,t}$ correspondingly into $V\times V$ blocks
\[
A_{m,t}=\big(A^{(uv)}_{m,t}: 1\le u,v\le V\big),
\]
where $A^{(uv)}_{m,t}$ is the submatrix of $A_{m,t}$ with rows in $\tilde{\mathcal N}_u$ and columns in $\tilde{\mathcal N}_v$.

\item For each $1\le v\le V$, and each candidate pair $(\alpha,\beta)\in\mathcal{A}\times\mathcal{B}$

\begin{enumerate}

\item \emph{Estimation of memberships:}
Denote $\bar{N}_1:=\cup_{u\neq v}\tilde{\mathcal N}_u$ and $\bar{N}_2:=\tilde{\mathcal N}_v$ as the training and test node sets, respectively. For each $m=1,\ldots,M$ and $t=1,\ldots,T$, let $A_{m,t}^{(-v)}$ be the $|\bar{N}_1|\times n$ submatrix of $A_{m,t}$ consisting of the rows indexed by $\bar{N}_1$.
Since MuDySC requires square symmetric input matrices, we construct
\[
\tilde{A}_{m,t}^{(v)} = \big(A_{m,t}^{(-v)}\big)^{\!\top} A_{m,t}^{(-v)} \in\mathbb{R}^{n\times n}.
\]
Run MuDySC on $\{\tilde{A}_{m,t}^{(v)}: m=1,\ldots,M;\; t=1,\ldots,T\}$ with parameters $(\alpha,\beta)$ to obtain the community labels $\hat z^{(v)}_{m,t}\in\{1,\ldots,K\}^n$.

\item \emph{Estimation of the block connectivity matrix:}
For each $(m,t)$, let $\mathcal N_{j,k}^{(m,t)}$ be the set of nodes in $\bar{N}_j$ assigned to community $k$ by $\hat z^{(v)}_{m,t}$, and $n_{j,k}^{(m,t)} = |\mathcal N_{j,k}^{(m,t)}|$ for $j=1,2$ and $1\le k\le K$. The block connectivity matrix $\hat B^{(m,t,v)}\in\mathbb{R}^{K\times K}$ is estimated as
\begin{equation*}
\hat B^{(m,t,v)}_{kk'} =
\begin{cases}
\displaystyle
\frac{\sum_{i\in \mathcal N_{1,k}^{(m,t)},\, j\in \mathcal N_{1,k'}^{(m,t)}\cup \mathcal N_{2,k'}^{(m,t)}} [{A}_{m,t}]_{i,j}}
      {n_{1,k}^{(m,t)}\big(n_{1,k'}^{(m,t)}+n_{2,k'}^{(m,t)}\big)},
& k\neq k',
\\[3ex]
\displaystyle
\frac{\sum_{i,j\in \mathcal N_{1,k}^{(m,t)},\, i<j} [{A}_{m,t}]_{i,j}
      + \sum_{i\in \mathcal N_{1,k}^{(m,t)},\, j\in \mathcal N_{2,k}^{(m,t)}} [{A}_{m,t}]_{i,j}}
     {(n_{1,k}^{(m,t)}-1)n_{1,k}^{(m,t)}/2 + n_{1,k}^{(m,t)}n_{2,k}^{(m,t)}},
& k=k'.
\end{cases}
\end{equation*}

\item \emph{Validation:}
Let $\hat s_{m,t}^{(v)}\in\mathbb{R}^{n\times K}$ be the community membership indicator matrix derived from $\hat z^{(v)}_{m,t}$, where $[\hat s_{m,t}^{(v)}]_{ik}=1$ if node $i$ belongs to community $k$ and $0$ otherwise. Form the predicted edge probability matrix $\hat P_{m,t}^{(v)} = \hat s_{m,t}^{(v)}\,\hat B^{(m,t,v)}\,(\hat s_{m,t}^{(v)})^{\!\top}$. Define the negative log-likelihood loss $\ell(x,p) = -x\log p - (1-x)\log(1-p)$. The prediction performance is evaluated on the validation nodes $\bar{N}_2$:
\[
\hat L^{(v)}(\alpha,\beta)
= \frac{1}{MT} \sum_{m=1}^{M}\sum_{t=1}^{T}
\sum_{i,j\in \bar{N}_2,\, i\neq j}
\ell\bigl([A_{m,t}]_{i,j},\,[\hat P^{(v)}_{m,t}]_{i,j}\bigr).
\]
\end{enumerate}

\item Let $\hat L(\alpha,\beta)=\sum_{v=1}^{V}\hat L^{(v)}(\alpha,\beta)$ and select
\[
(\hat\alpha,\hat\beta)=\arg\min_{\alpha\in\mathcal{A},\,\beta\in\mathcal{B}}\; \hat L(\alpha,\beta).
\]
\end{enumerate}

\subsection{Convergence analysis}\label{sec2.3}

The following result shows that the alternating iterative algorithm converges to the global minimum of~\eqref{eq6}.

\begin{theorem}[Global convergence]
\label{thm1}
When $2\alpha + \beta < \frac{1}{1+2\sqrt{2}}$, the iterative algorithm defined through \eqref{eq7} converges to the global minimum of \eqref{eq6}.
\end{theorem}

\emph{Proof sketch of Theorem \ref{thm1}}: 
We will first write the alternating iterative rules as an operator $G:\mathbb R^{M\times T\times n\times n}\rightarrow \mathbb R^{M\times T\times n \times n}$. Then we show that the global minimum of \eqref{eq6} must be a fixed point of $G$. After that, we show that $G$ is a contraction mapping and has a unique fixed point. Finally, it follows that the unique fixed point of $G$ is the
global minimum of \eqref{eq6}. The detailed proof can be found in Appendix ~\ref{sec8}.

It is worth noting that Theorem~\ref{thm1} establishes the global convergence of the MuDySC
algorithm under the condition
\(
2\alpha+\beta<\frac{1}{1+2\sqrt{2}}.
\)
As shown in Lemma~\ref{lemma2}, this condition guarantees that the
iterative mapping is contractive over the entire parameter space and
is therefore a relatively conservative sufficient condition.
In fact, if we only consider convergence in a neighborhood of a local
minimizer, the condition on $\alpha$ and $\beta$ can be relaxed.
We provide the corresponding local convergence result below.

We provide some notations first. Let
\(
\overline U^\dagger
=
(\overline U_{m,t}^\dagger)_{M\times T}
\)
be a local minimizer of \eqref{eq6}. For each $t$, let
$\mathcal N_t=\{t-1,t+1\}\cap\{1,\ldots,T\}$ and
$q_t=|\mathcal N_t|$, and define
\[
\Sigma_{m,t}^\dagger
=
U_{m,t}
+\alpha\sum_{s\in\mathcal N_t}\overline U_{m,s}^\dagger
+\frac{\beta}{M-1}\sum_{i\neq m}\overline U_{i,t}^\dagger,
\qquad
\Delta^\dagger
=
\min_{m,t}
\left\{
\lambda_K(\Sigma_{m,t}^\dagger)
-\lambda_{K+1}(\Sigma_{m,t}^\dagger)
\right\}.
\]
For $\overline U=(\overline U_{m,t})_{M\times T}$ and
$\overline U'=(\overline U'_{m,t})_{M\times T}$, let
\[
d(\overline U,\overline U')
=
\sum_{m=1}^{M}\sum_{t=1}^{T}
\|\overline U_{m,t}-\overline U'_{m,t}\|_F.
\]

\begin{theorem}[Local convergence]
\label{thm:local}

\textnormal{(a)}
If
\begin{equation}
\label{eq:local-gap-condition}
\Delta^\dagger>2\sqrt{2}(2\alpha+\beta),
\end{equation}
then $\overline U^\dagger$ is a fixed point of the iterative mapping
defined by \eqref{eq7}. Moreover, for any
\[
0<r<
\frac{\Delta^\dagger-2\sqrt{2}(2\alpha+\beta)}
     {2(2\alpha+\beta)},
\]
the iterative mapping defined by \eqref{eq7} is contractive on
\[
\mathcal B_r(\overline U^\dagger)
:=
\{\overline U:d(\overline U,\overline U^\dagger)\leq r\},
\]
with contraction coefficient
\[
\rho_r
=
\frac{2\sqrt{2}(2\alpha+\beta)}
{\Delta^\dagger-2(2\alpha+\beta)r}
<1.
\]
Consequently, if the initialization belongs to
$\mathcal B_r(\overline U^\dagger)$, the sequence generated by
\eqref{eq7} converges to $\overline U^\dagger$.

\textnormal{(b)}
Suppose that there exists $\eta\geq0$ such that, for all $m,t$,
\begin{equation}
\label{eq:subspace-alignment}
\max\left\{
\|U_{m,t}-\overline U_{m,t}^\dagger\|_2,\,
\max_{s\in\mathcal N_t}
\|\overline U_{m,s}^\dagger-\overline U_{m,t}^\dagger\|_2,\,
\max_{i\neq m}
\|\overline U_{i,t}^\dagger-\overline U_{m,t}^\dagger\|_2
\right\}
\leq\eta.
\end{equation}
Then
\begin{equation}
\label{eq:local-gap-bound}
\Delta^\dagger
\geq
\min_t\left\{
(1+q_t\alpha+\beta)(1-2\eta)
\right\}.
\end{equation}
Hence, if
\begin{equation}
\label{eq:local-sufficient-condition}
\min_t\left\{
(1+q_t\alpha+\beta)(1-2\eta)
\right\}
>
2\sqrt{2}(2\alpha+\beta),
\end{equation}
the local contractivity and convergence results in
part~\textnormal{(a)} hold.
\end{theorem}

\begin{remark}
Part~\textnormal{(a)} provides a general local convergence condition
in terms of the eigengap at a local minimizer. Specifically,
condition~\eqref{eq:local-gap-condition} guarantees that the iterative
mapping is contractive in a neighborhood of the local minimizer, so
that a sufficiently good initialization leads to convergence.

Part~\textnormal{(b)} further provides an explicit characterization
of when condition~\eqref{eq:local-gap-condition} can hold. Under the
subspace alignment condition~\eqref{eq:subspace-alignment},
the local eigengap admits the lower bound in
\eqref{eq:local-gap-bound}. Consequently,
condition~\eqref{eq:local-sufficient-condition} provides a sufficient
condition for the local contractivity and convergence result in
part~\textnormal{(a)}.
\end{remark}

\begin{remark}
Compared with the worst-case eigengap bound used in the global
analysis, the lower bound in \eqref{eq:local-gap-bound} can be
substantially larger when the eigenspaces associated with neighboring
time points and different network layers are sufficiently close.
For example, when $\eta=0$, condition~\eqref{eq:local-sufficient-condition}
reduces to
\(
1+\alpha+\beta>2\sqrt{2}(2\alpha+\beta),
\)
which is substantially weaker than the global condition
\(
2\alpha+\beta<\frac{1}{1+2\sqrt{2}}
\)
in Theorem~\ref{thm1}.
Therefore, when the subspaces are sufficiently close, corresponding to a sufficiently small value of $\eta$, the local convergence condition may hold for larger values of $\alpha$ and $\beta$ than those allowed by Theorem~\ref{thm1}.
This also provides a theoretical explanation for why the algorithm may
converge under a good initialization even when the global sufficient
condition in Theorem~\ref{thm1} is not satisfied.
\end{remark}

\section{Application to FAO dataset}\label{sec4}
In this section, we apply the proposed MuDySC method to the FAO dataset. In particular, we focus on the import and export trade of 15 vegetable oil products across 131 countries or regions over the five-year period from 2019 to 2023.

\subsection{Data description}\label{subsec1}

The FAOSTAT ~\footnote{\url{https://www.fao.org/faostat/}}, provided by the Food and Agriculture Organization of the United Nations (FAO), is one of the most authoritative international platforms for agricultural and food trade statistics, covering over 20 thematic areas and more than 4,000 indicators on production, trade, land use, and population nutrition.
In particular, the trade sub-database Detailed Trade Matrix (DTM) ~\footnote{\url{https://www.fao.org/faostat/en/\#data/TM}} of FAOSTAT provides bilateral import and export volumes and values for 573 agricultural and food commodities (including cereals, dairy, and meat products) from 1986 to 2023 across 245 countries or regions, encompassing over 113 million trade records. 
Following the Central Product Classification (CPC) standard established by the United Nations Statistical Commission, these 573 products are categorized into multiple product groups. The categories and the number of products in each category are summarized in Appendix~\ref{classification_table}.

\subsection{Data preprocessing}\label{subsec2}
From the DTM of FAOSTAT, we extract the international trade data for 15 vegetable oil products from the ``Meat Products and Fats'' category over the period 2019-2023. The raw data are preprocessed as follows.

First, the countries or regions are regarded as the nodes of the potential multilayer-dynamic network. For each given year and each given product, if the export value of this product from country $i$ to country $j$ exceeds USD 100{,}000, a directed edge from $i$ (exporter) to $j$ (importer) is established. We regard countries that participate in trade in fewer than 4 product categories as low-activity nodes, and remove countries or regions with degree smaller than the resulting threshold of 8, yielding a final set of $n=131$ nodes. The directed relationships are thus represented by \emph{bipartite} multilayer-dynamic network with $n=131,M=15,T=5$. In particular, we define $X_{m,t} \in \{0,1\}^{n \times n}$ for $m\in \{1,...,15\}$ and $t\in\{1,...,5\}$, where the $(i,j)$th entry  
$[X_{m,t}]_{i,j}=1$ indicates that there exists a direct edge from country $i$ to $j$ for product $m$ at time $t$. Finally, to analyze the import and export community patterns separately, we follow the convention to define the adjacency matrices \(A_{m,t}^{im}=\operatorname{sign}(X_{m,t}^{\top}X_{m,t})\) and \(A_{m,t}^{ex}=\operatorname{sign}(X_{m,t}X_{m,t}^{\top})\) for the import and export multilayer-dynamic trade networks, respectively, where the indicator function sign is applied entrywise.

\subsection{Sparsity Analysis}\label{subsec3}
We first provide a descriptive analysis of the sparsity of the import and export trade networks. We define 
\begin{equation}
\begin{aligned}
S^{im}
&=1-\frac{1}{MT} \sum_{m=1}^{M}\sum_{t=1}^{T}\rho\left(A_{m,t}^{im}\right), \\
S^{ex}
&=1-\frac{1}{MT} \sum_{m=1}^{M}\sum_{t=1}^{T}\rho\left(A_{m,t}^{ex}\right).
\end{aligned}
\label{eq:sparsity}
\end{equation}
as the average sparsities of the import and export trade networks, respectively, where $\rho(A)$ denotes the density of $A$, namely, the proportion of observed edges relative to the maximum possible number of edges. The results show that $S^{ex}=0.9481$ and $S^{im}=0.7377$ .

We further exclude inactive countries that do not participate in trade for the corresponding product-year and recalculate the network sparsities. The results show that the average sparsities of the export and import networks are 0.6222 and 0.3920, respectively, with the export network remaining substantially sparser. This suggests that the sparsity difference cannot be explained solely by the number of participating countries. Among active countries, exporters are less likely to share common export destinations, whereas importers are more likely to share common import origins. The latter pattern may be related to the concentration of vegetable-oil production in a relatively small number of producing countries.

\subsection{Determination of the number of communities}\label{subsec4}
We use the scree plot to determine the number of communities within the import
and export trade networks, respectively. Specifically, we compute the
eigenvalues of the aggregated matrices
$\sum_{m,t} A_{m,t}^{im}$ and $\sum_{m,t} A_{m,t}^{ex}$.
The top 20 eigenvalues, ordered from largest to smallest, are shown in
Figure~\ref{fig3}. Both scree plots exhibit an elbow around $K=5$.

To further assess the choice of $K$, we use the Davies--Bouldin (DB) score,
following \cite{su2024spectral}. For each candidate $K=3,\ldots,10$, we form
a $K$-dimensional spectral embedding from the aggregated network, apply
K-means with $K$ clusters, and compute the corresponding DB score. As shown
in Fig.~\ref{fig:db_select_k}, the DB score is minimized at $K=5$ for both
the import and export networks. Based on these results, we set $K=5$ for the
subsequent analysis.

\begin{figure*}[!h]{}
\centering
\subfigure[Export trade network]{\includegraphics[height=5cm,width=6.5cm,angle=0]{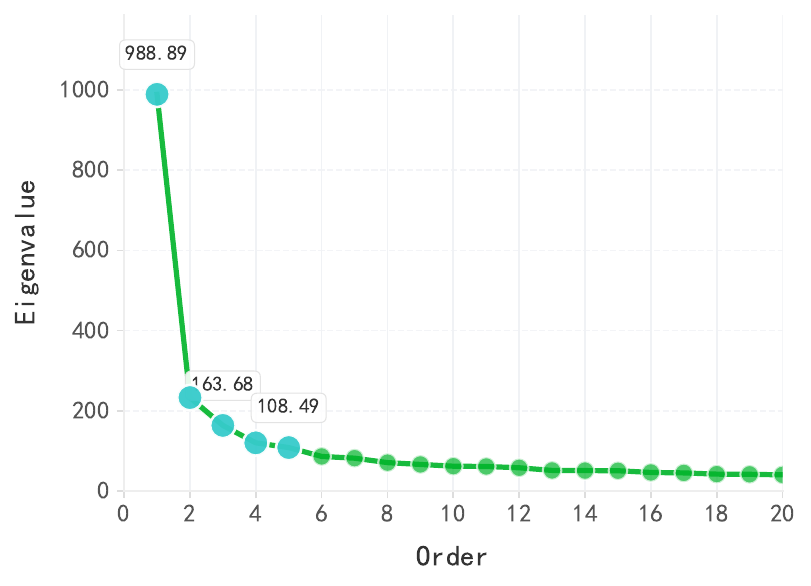}}
\subfigure[Import trade network]{\includegraphics[height=5cm,width=6.5cm,angle=0]{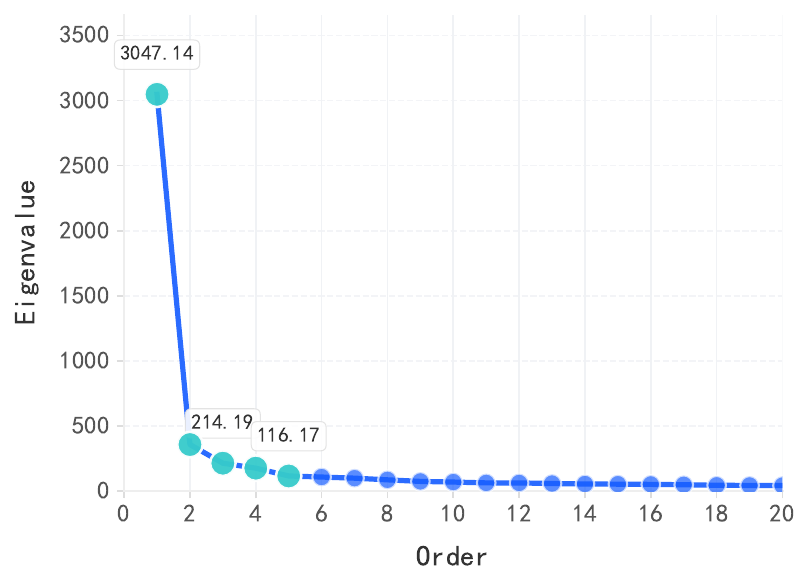}}
\caption{Scree plots of the leading eigenvalues for the (a) export trade network and (b) import trade network.}
\label{fig3}
\end{figure*}

\begin{figure*}[!h]{}
\centering
\includegraphics[width=11cm]{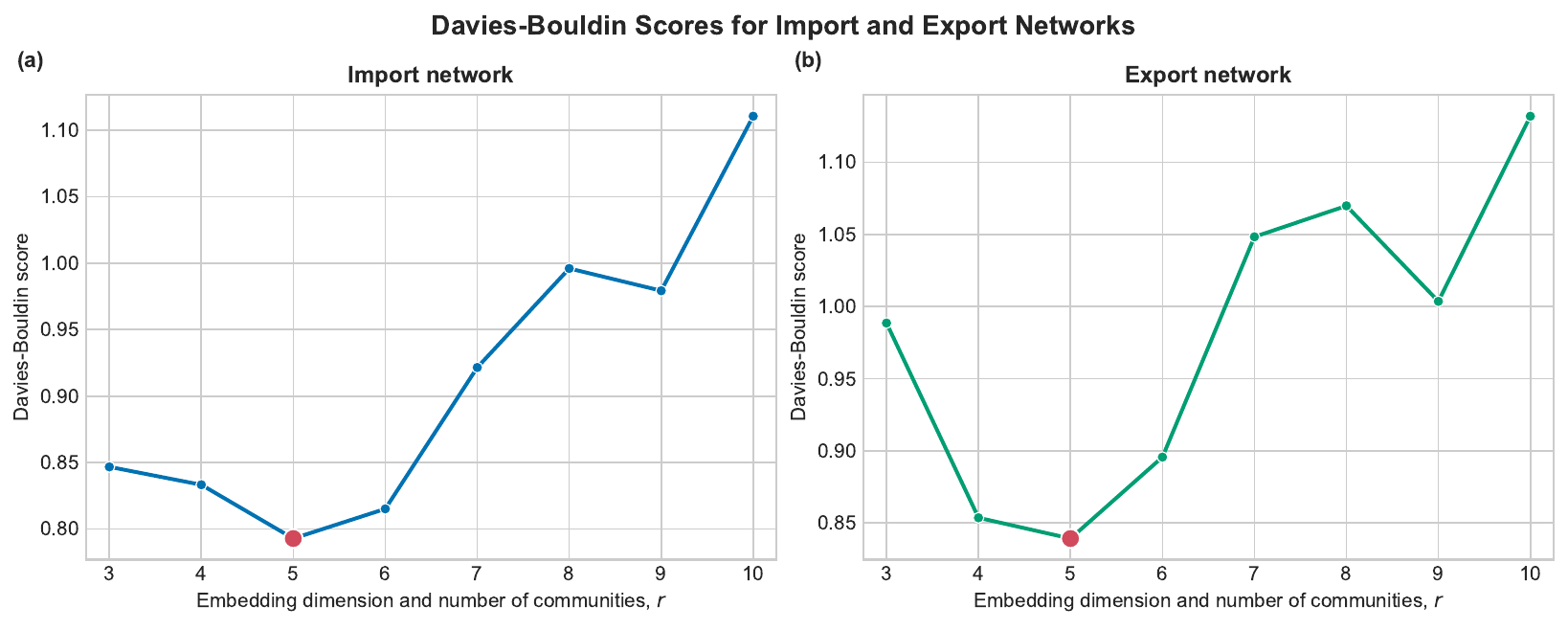}
\caption{Davies--Bouldin(DB) scores for the import (a) and export (b) networks. For each candidate $r=3,\ldots,10$, an $r$-dimensional spectral embedding is clustered into $r$ communities. Red points indicate the lowest DB scores within the displayed candidate range, attained at $r=5$ for both networks.}
\label{fig:db_select_k}
\end{figure*}

\subsection{Community detection}\label{subsec5}
We apply the proposed MuDySC method to the multilayer-dynamic import and export trade networks, each of which consists of 131 nodes, 15 layers, and 5 time points. 
The tuning parameters $\alpha$ and $\beta$ in \eqref{eq6} are selected using the cross-validation procedure described in Section \ref{sec2.2}.

\begin{figure*}[!h]{}
\centering
\subfigure[2019]{\includegraphics[height=3.5cm,width=4cm,angle=0]{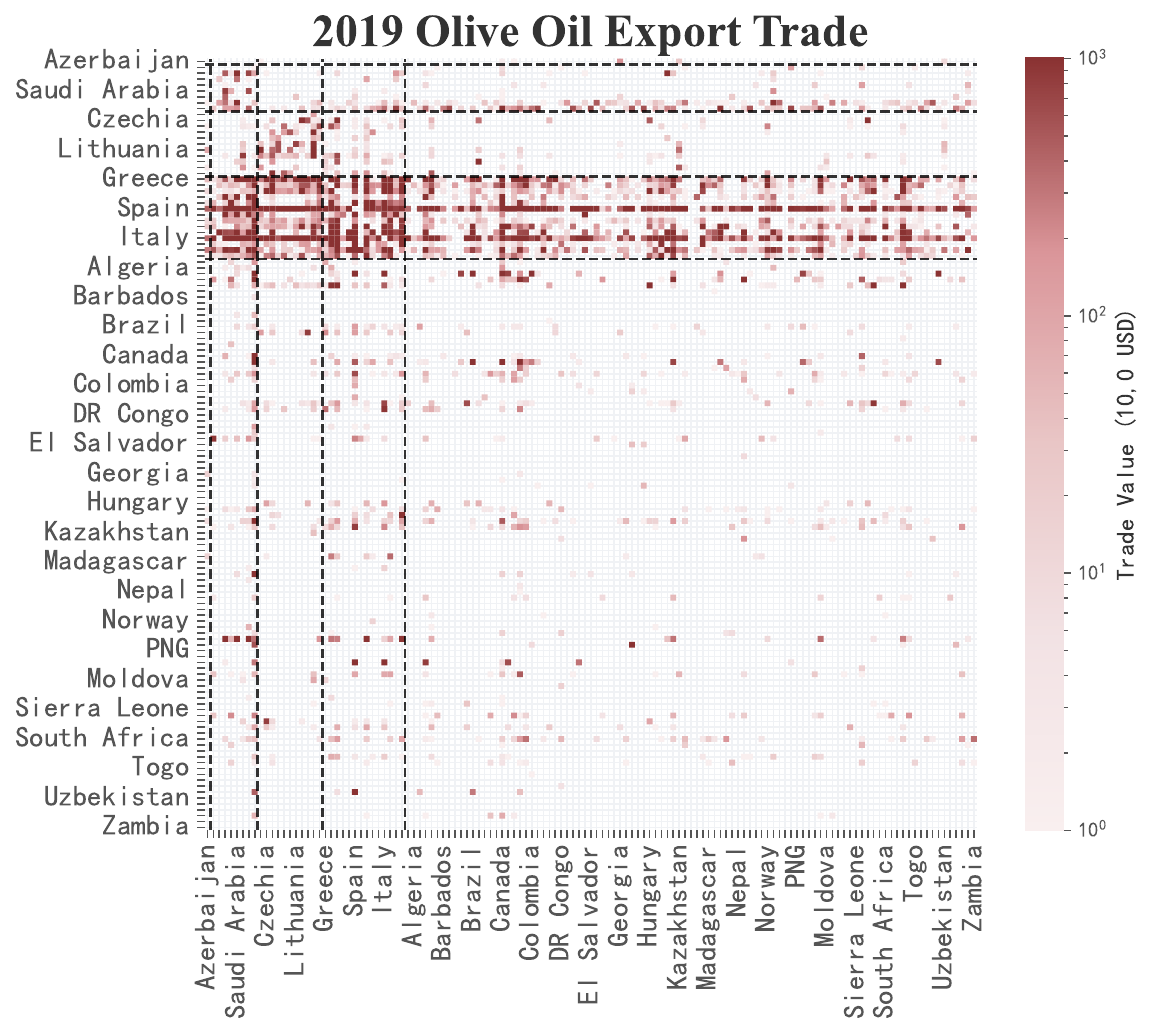}}
\subfigure[2020]{\includegraphics[height=3.5cm,width=4cm,angle=0]{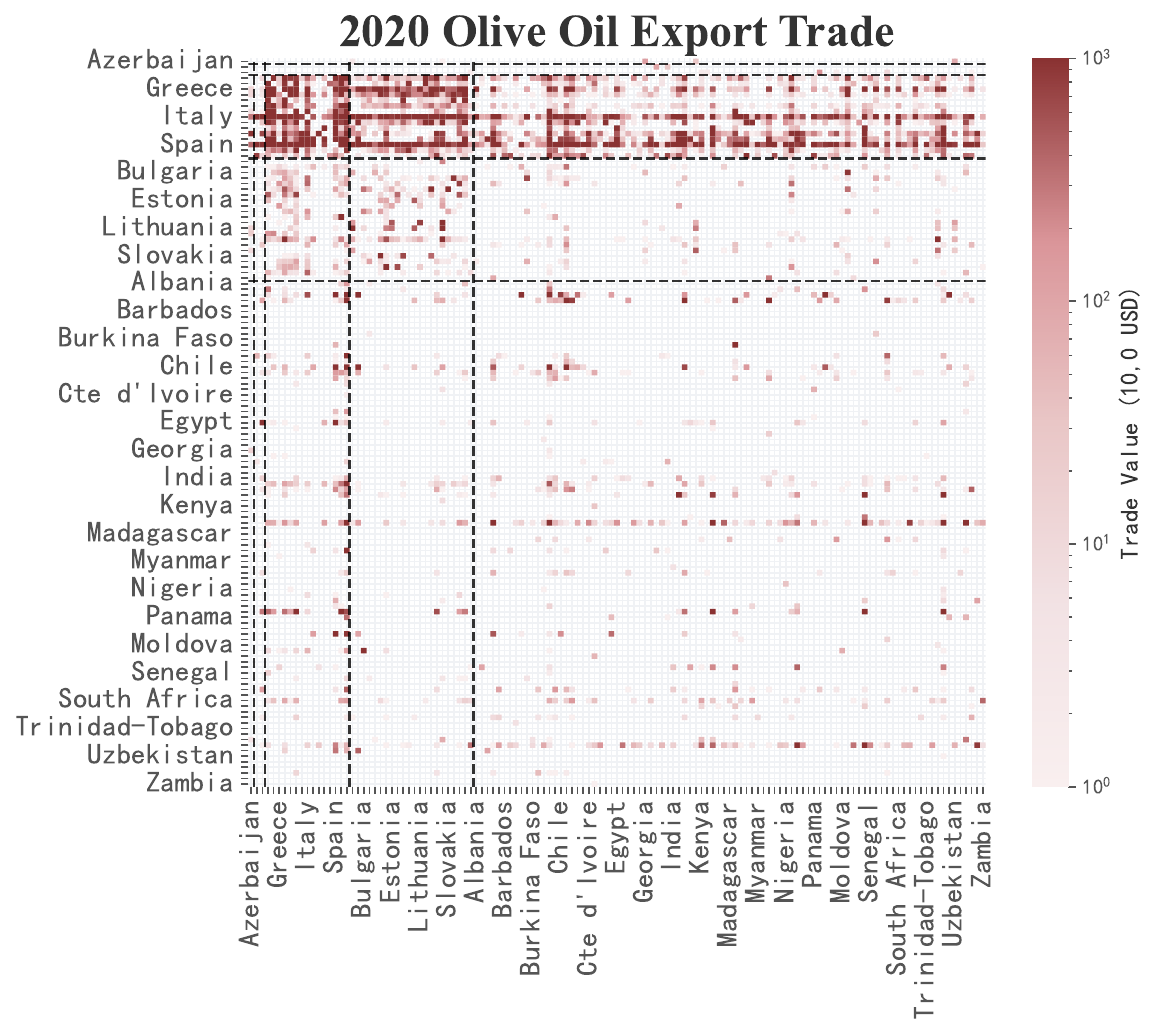}}
\subfigure[2021]{\includegraphics[height=3.5cm,width=4cm,angle=0]{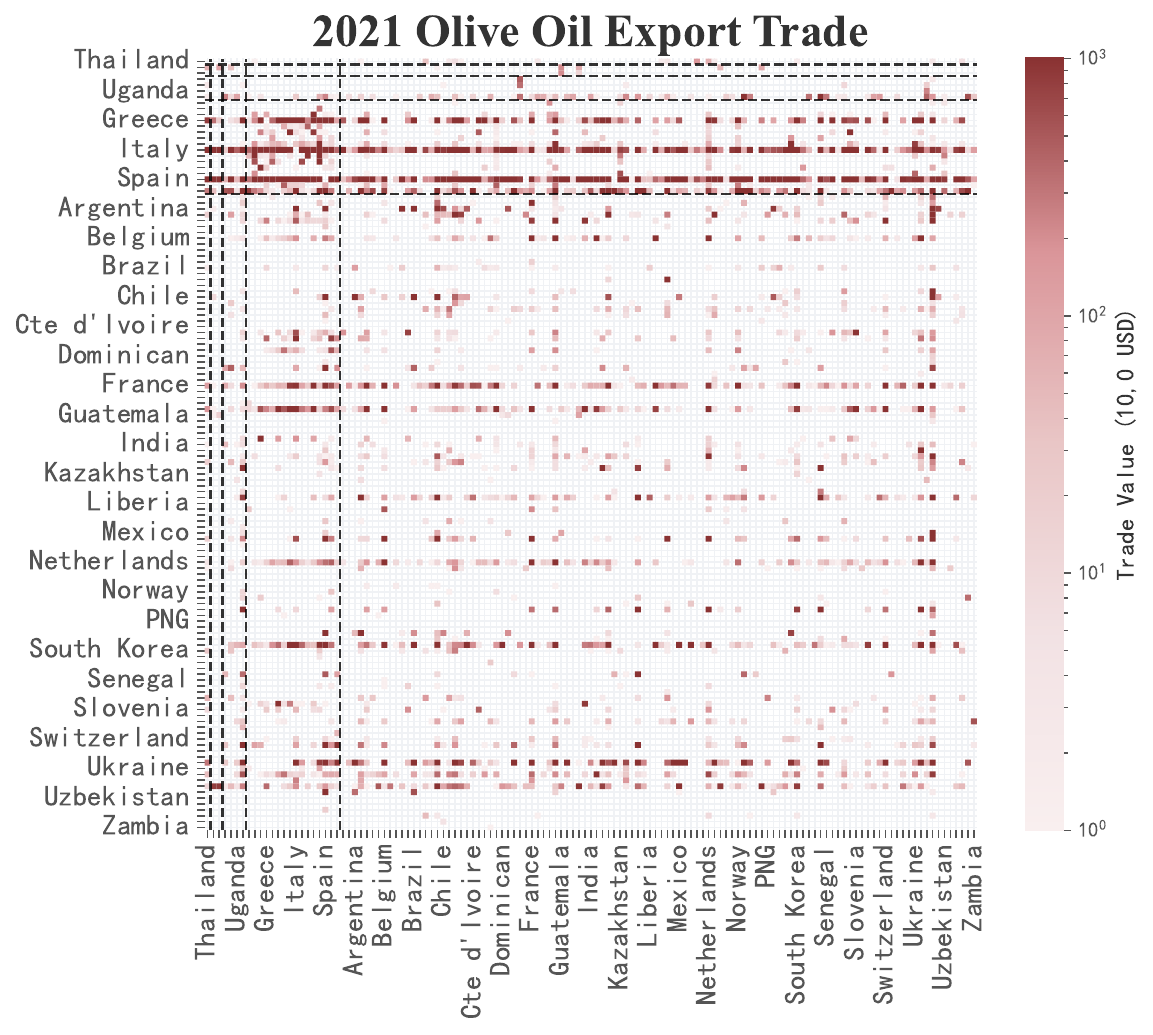}}
\subfigure[2022]{\includegraphics[height=3.5cm,width=4cm,angle=0]{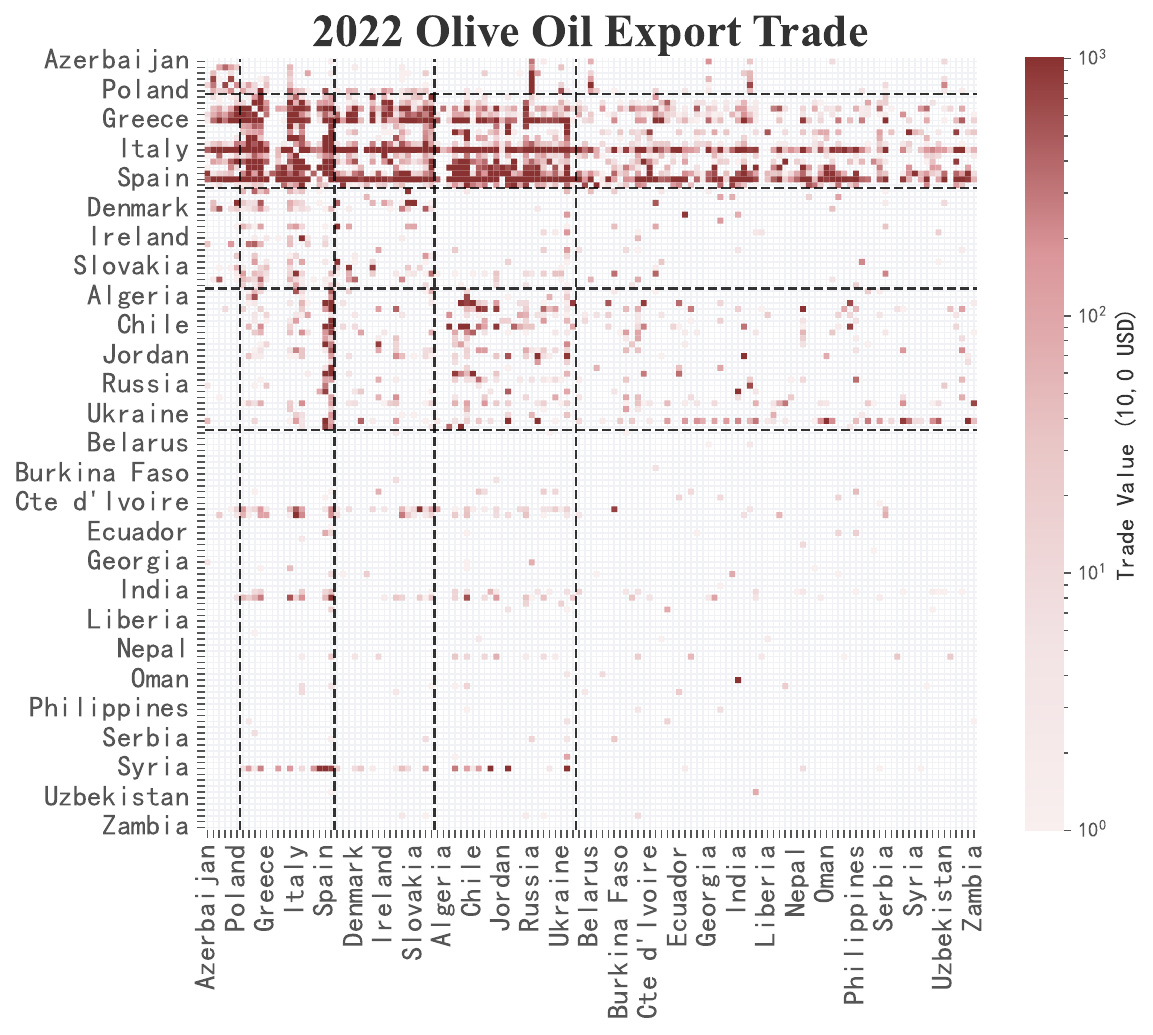}}
\subfigure[2023]{\includegraphics[height=3.5cm,width=4cm,angle=0]{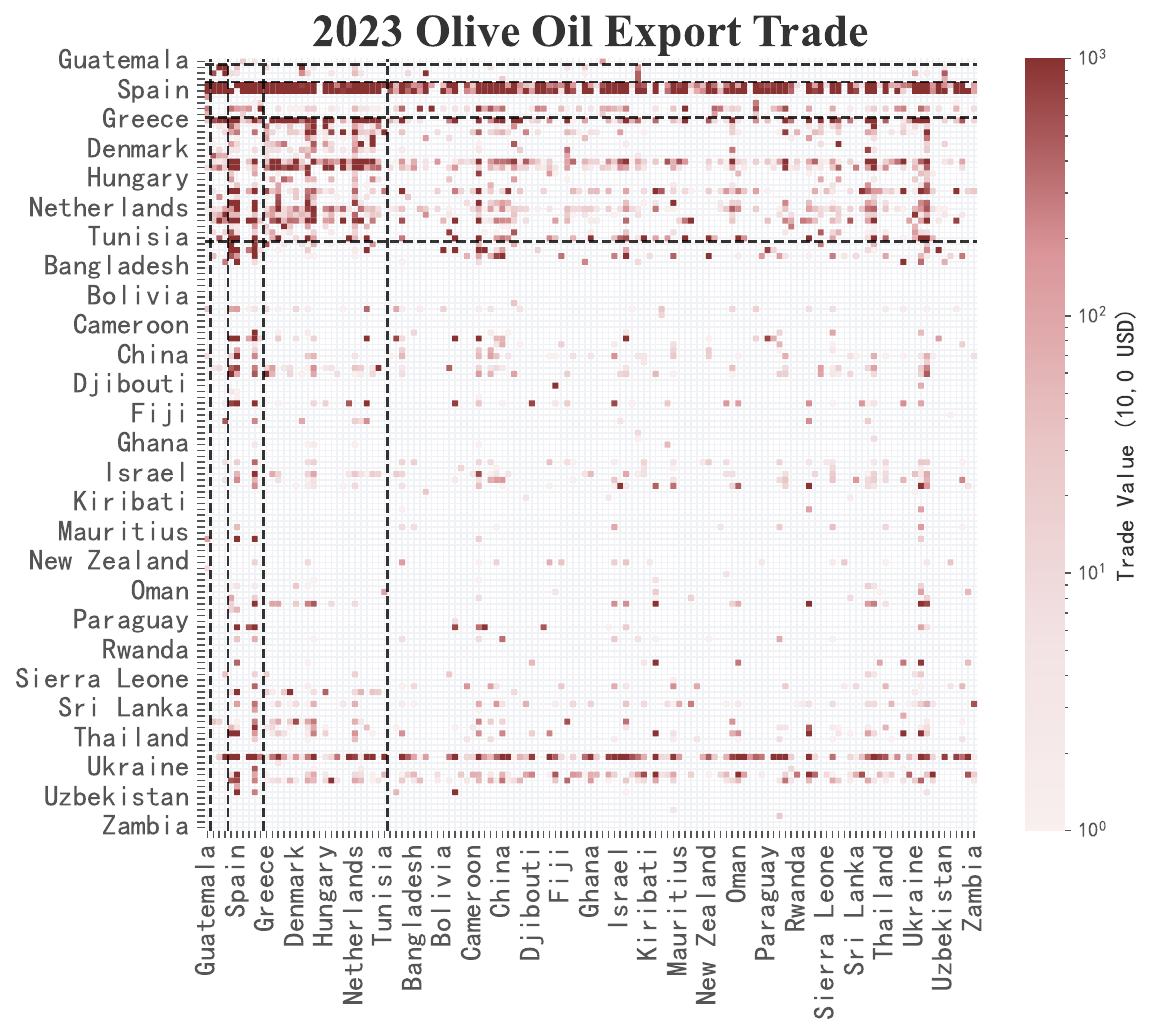}}
\caption{Heatmaps of olive-oil export values from 2019 to 2023. Each cell represents the export trade value from the country in the row to the country in the column. Darker colors indicate higher export values, and the dashed lines separate the five detected export communities.}
\label{fig7}
\end{figure*}

Next, we take the olive-oil as an example to analyze its community structure and evolution, which is one of the vegetable oil products that are traded globally at a large scale and involve a wide range of countries.
We first examine the overall community patterns in the import and export multilayer-dynamic olive-oil trade networks, and then analyze the evolution of communities in both the export and import networks in more detail.

\textbf{Overall community patterns.} 
To illustrate the overall community patterns in the olive-oil trade networks, we present heatmaps of the \emph{export and import matrices} from 2019 to 2023 in Figures~\ref{fig7} and \ref{fig9}, respectively. The rows and columns correspond to countries, which are reordered according to the detected communities so that countries in the same community are placed together. The dashed lines separate the five communities. Each cell in Figure~\ref{fig7} (resp. Figure~\ref{fig9}) represents the export (resp. import) trade value between two countries. In Figure~\ref{fig7}, rows indicate exporters and columns indicate importers, while in Figure~\ref{fig9}, rows indicate importers and columns indicate exporters. Darker colors indicate larger trade values.

\begin{figure*}[!h]{}
\centering
\subfigure[2019]{\includegraphics[height=3.5cm,width=4cm,angle=0]{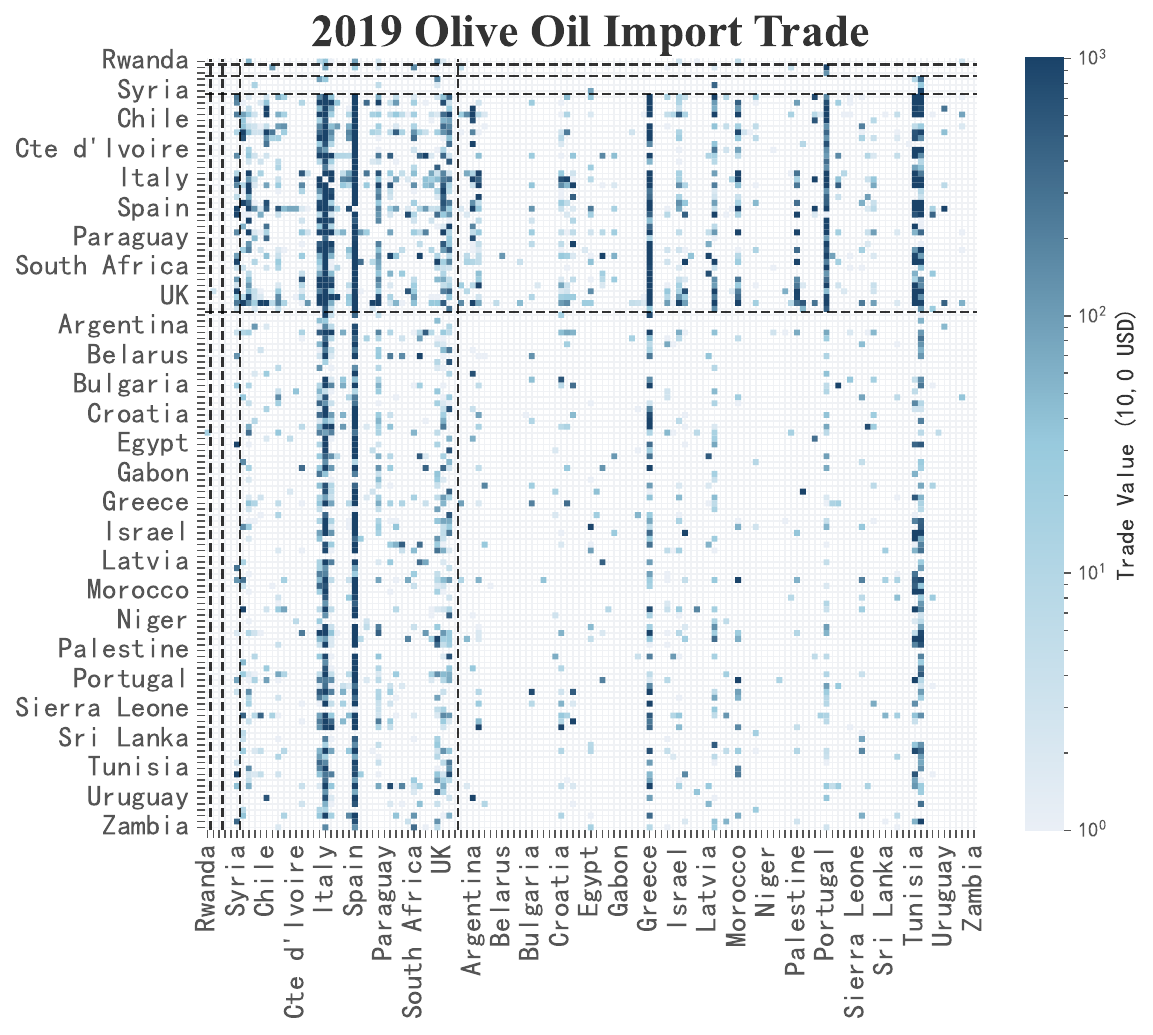}}
\subfigure[2020]{\includegraphics[height=3.5cm,width=4cm,angle=0]{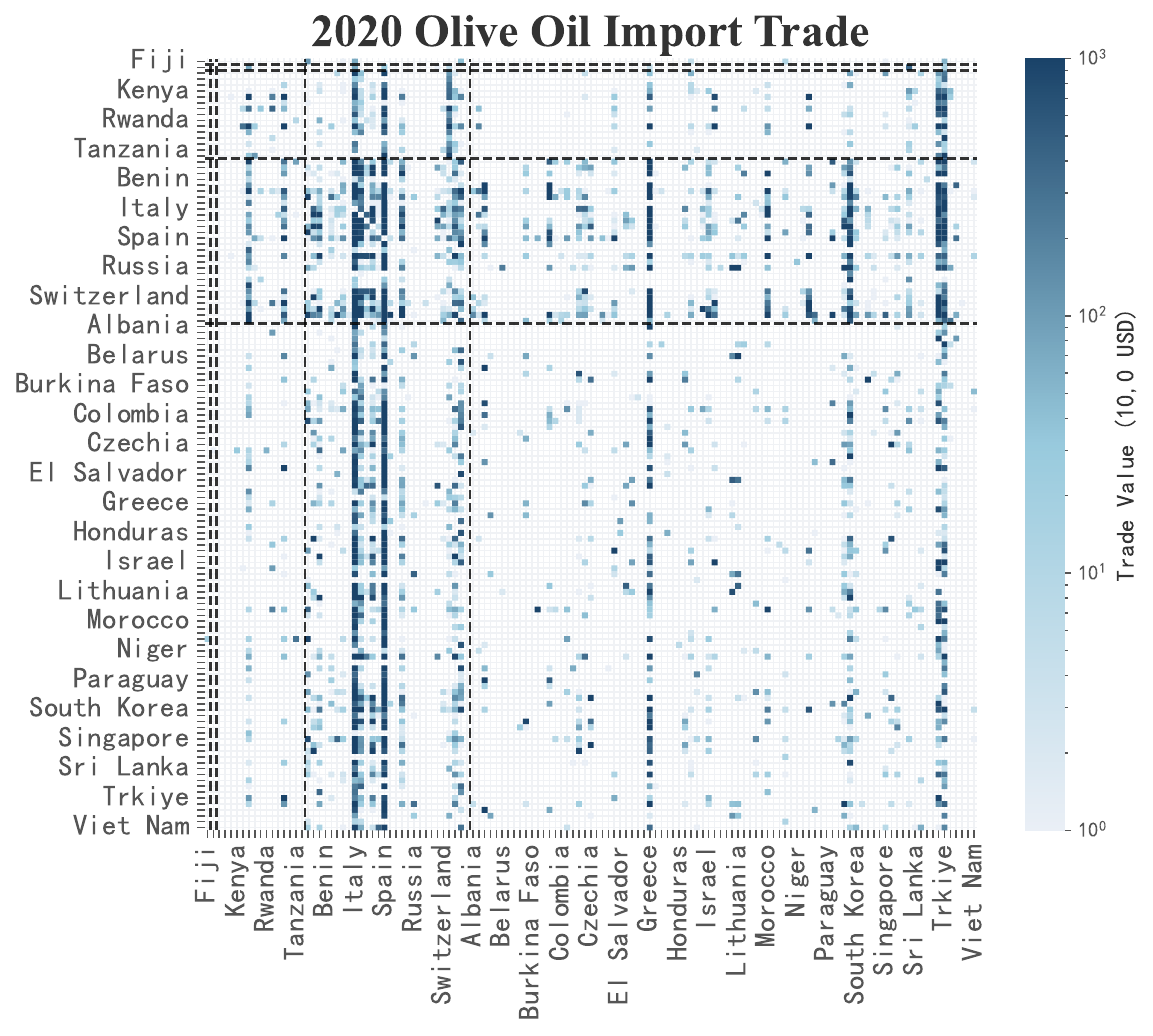}}
\subfigure[2021]{\includegraphics[height=3.5cm,width=4cm,angle=0]{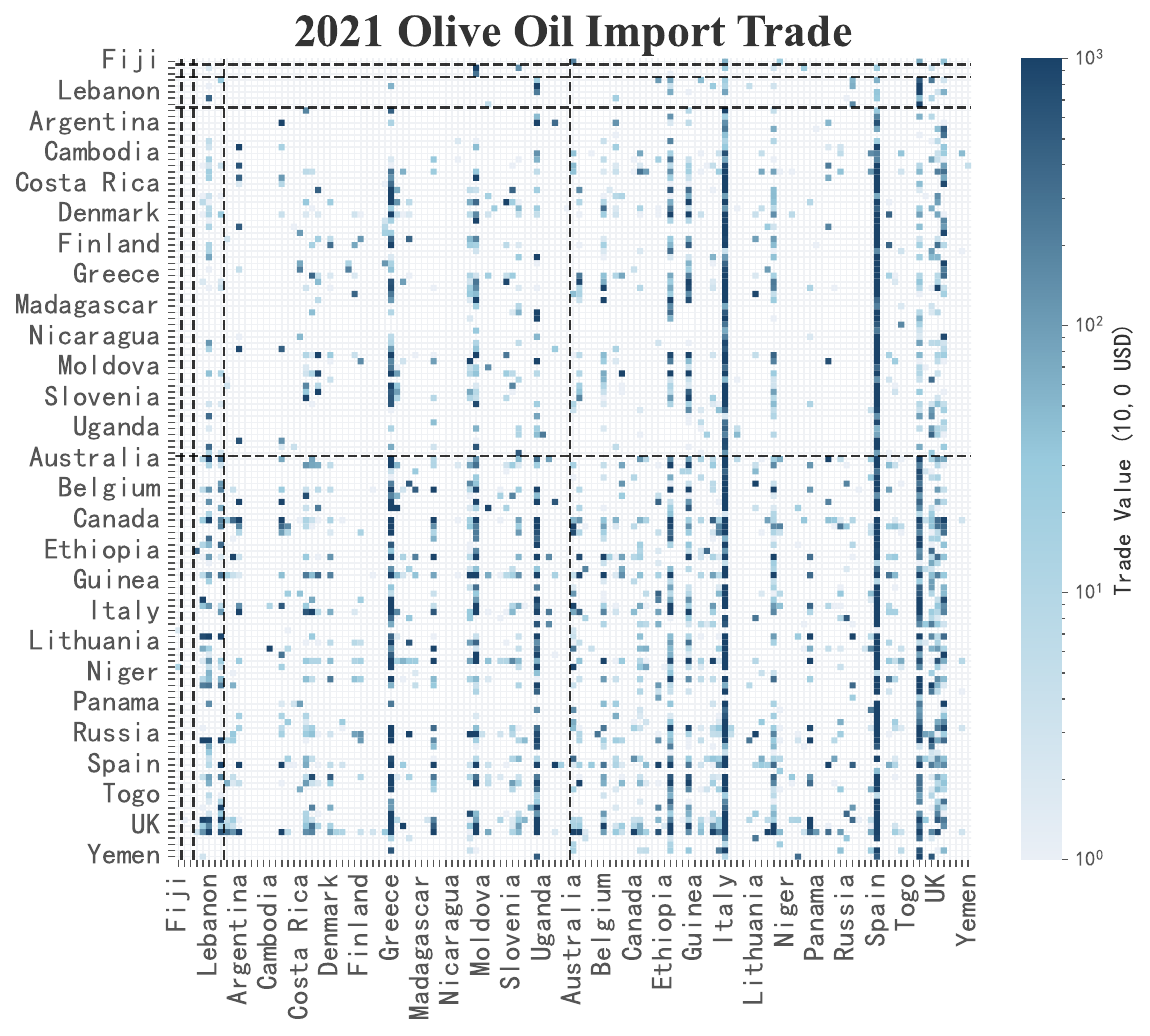}}
\subfigure[2022]{\includegraphics[height=3.5cm,width=4cm,angle=0]{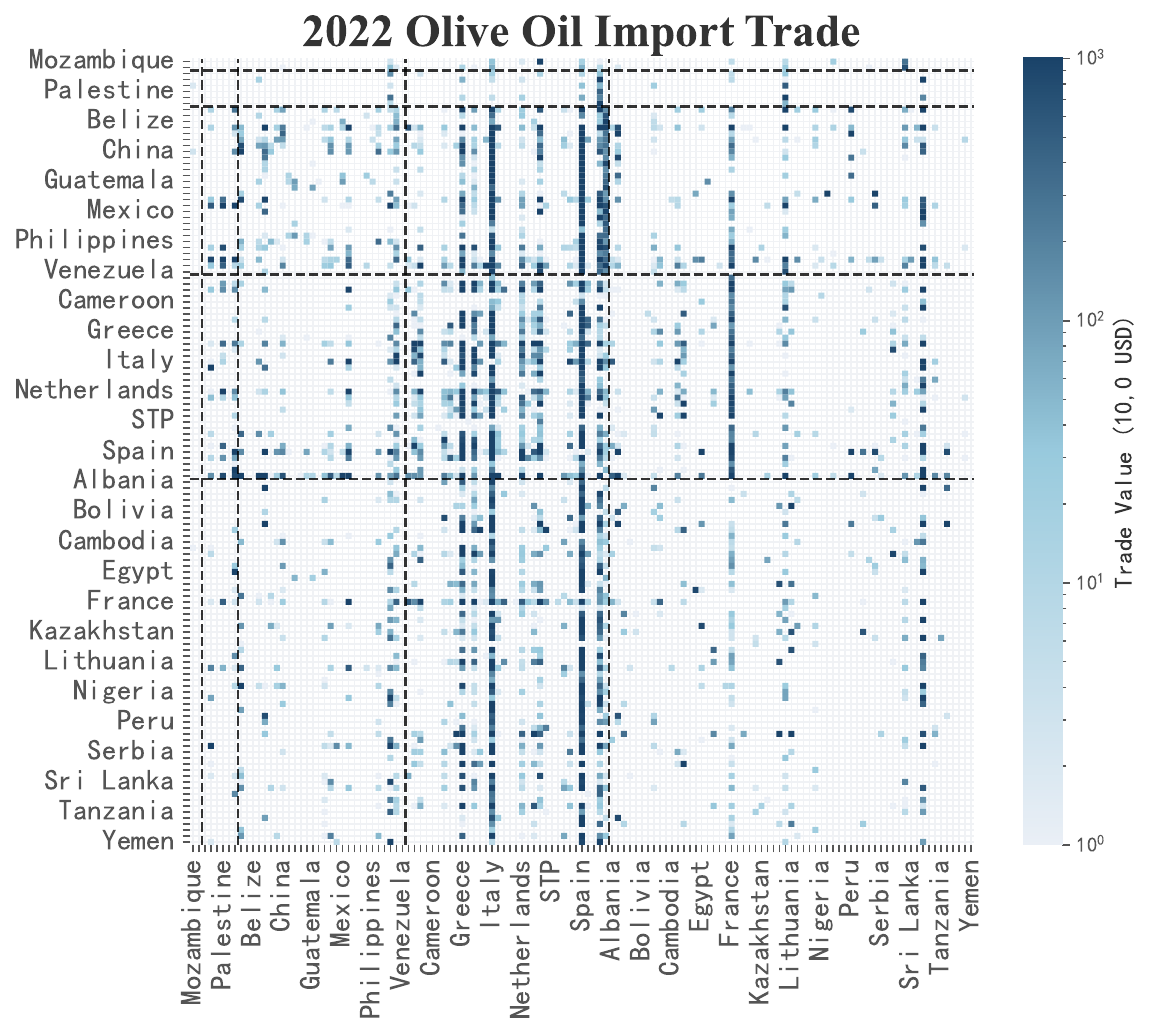}}
\subfigure[2023]{\includegraphics[height=3.5cm,width=4cm,angle=0]{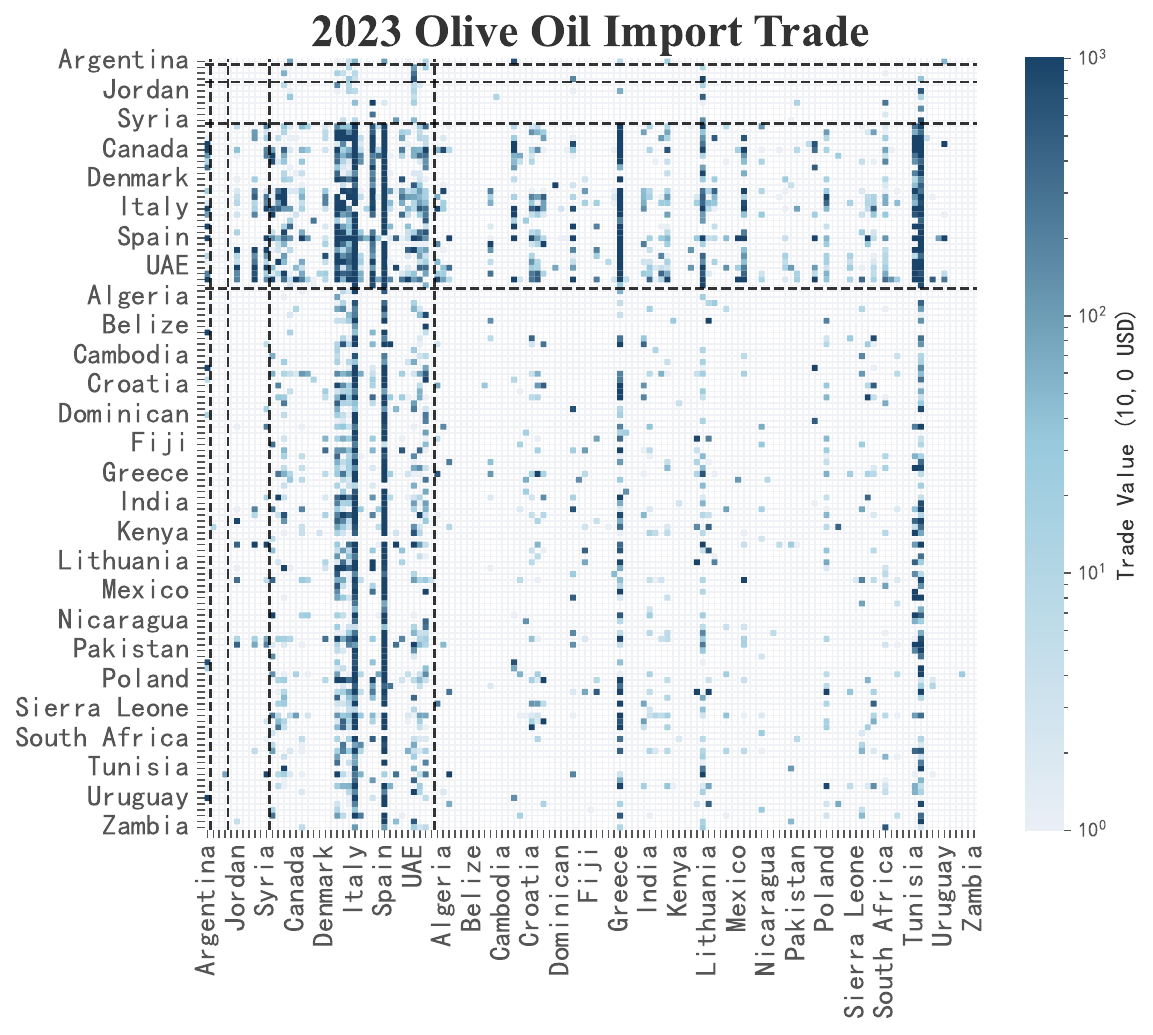}}
\caption{Heatmaps of olive-oil import values from 2019 to 2023. 
Each cell represents the import trade value by the country in the row from the country in the column.}
Darker colors indicate higher import values, and the dashed lines separate the five detected import communities.
\label{fig9}
\end{figure*}

We have the following observations. First, on the export side, countries within the same community exhibit more similar export patterns (i.e., rows), whereas countries from different communities are much less similar in their export patterns. In particular, the community containing major olive-oil producers such as Italy and Spain displays the darkest color, indicating that countries in this community account for relatively large export values to other countries. Second, on the import side, the differences across communities are less pronounced than those on the export side, suggesting that import patterns (i.e., rows) are more similar across countries. A main reason is that the sources of olive-oil supply are highly concentrated, with most countries relying on a small number of core producers, such as Spain, Italy, and Greece.

In addition to the findings above, the results indicate that, in both the import and export trade networks from 2019 to 2023, China is consistently grouped into the same community as Japan and Australia, indicating strong similarity in the olive-oil trade patterns of these three countries. This may be related to the fact that all three countries belong to the Asia-Pacific economic sphere. By contrast, China is never assigned to the same community as Greece during this five-year period, reflecting a substantial difference in their roles in the olive-oil trade network. Greece is a major olive-oil exporter, whereas China is primarily an importer, which may explain why the two countries remain in different communities over time.

\begin{figure*}[!h]{}
\centering
\subfigure[2019]{\includegraphics[height=2.6cm,width=4.5cm,angle=0]{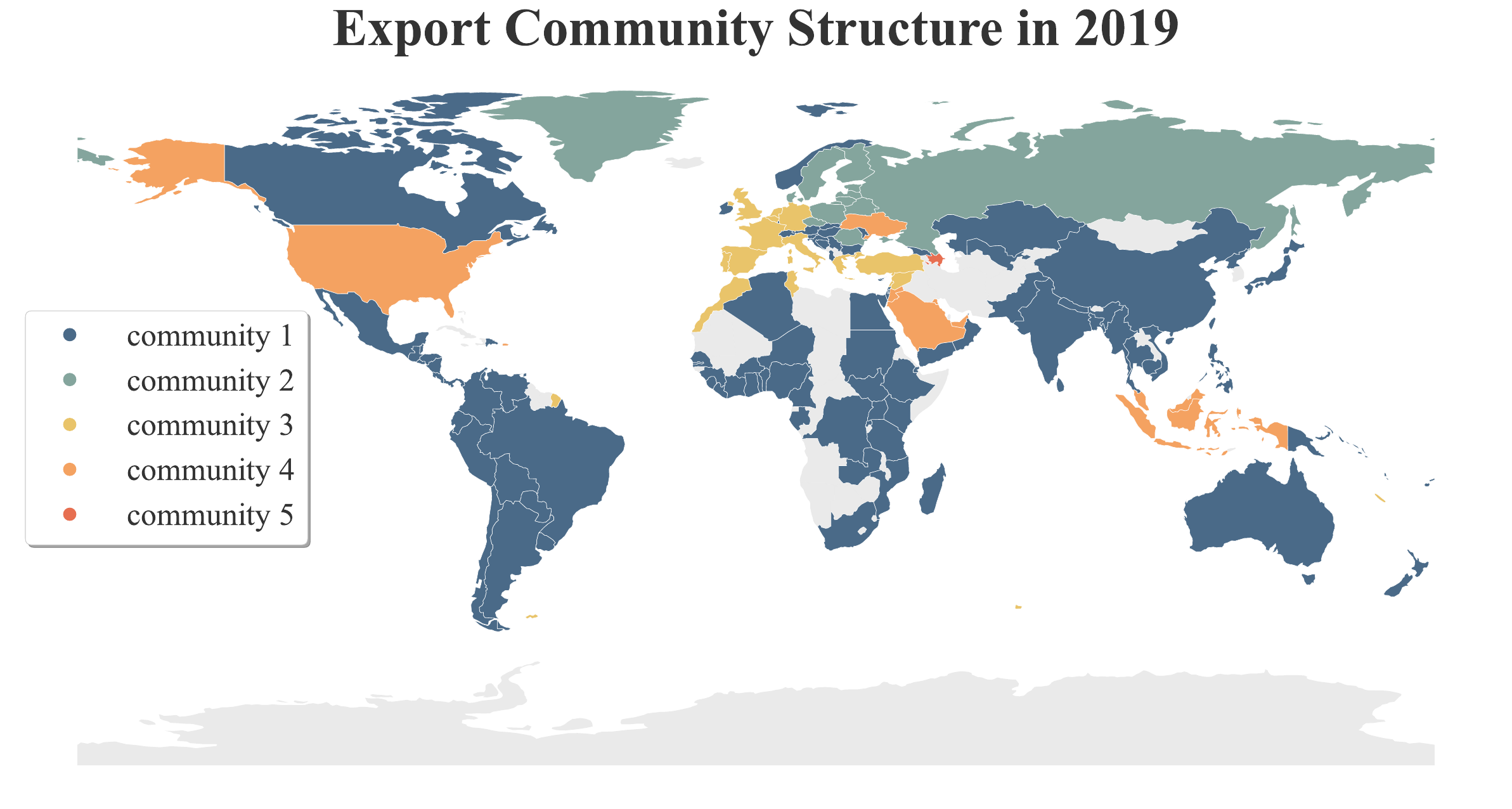}}
\subfigure[2020]{\includegraphics[height=2.6cm,width=4.5cm,angle=0]{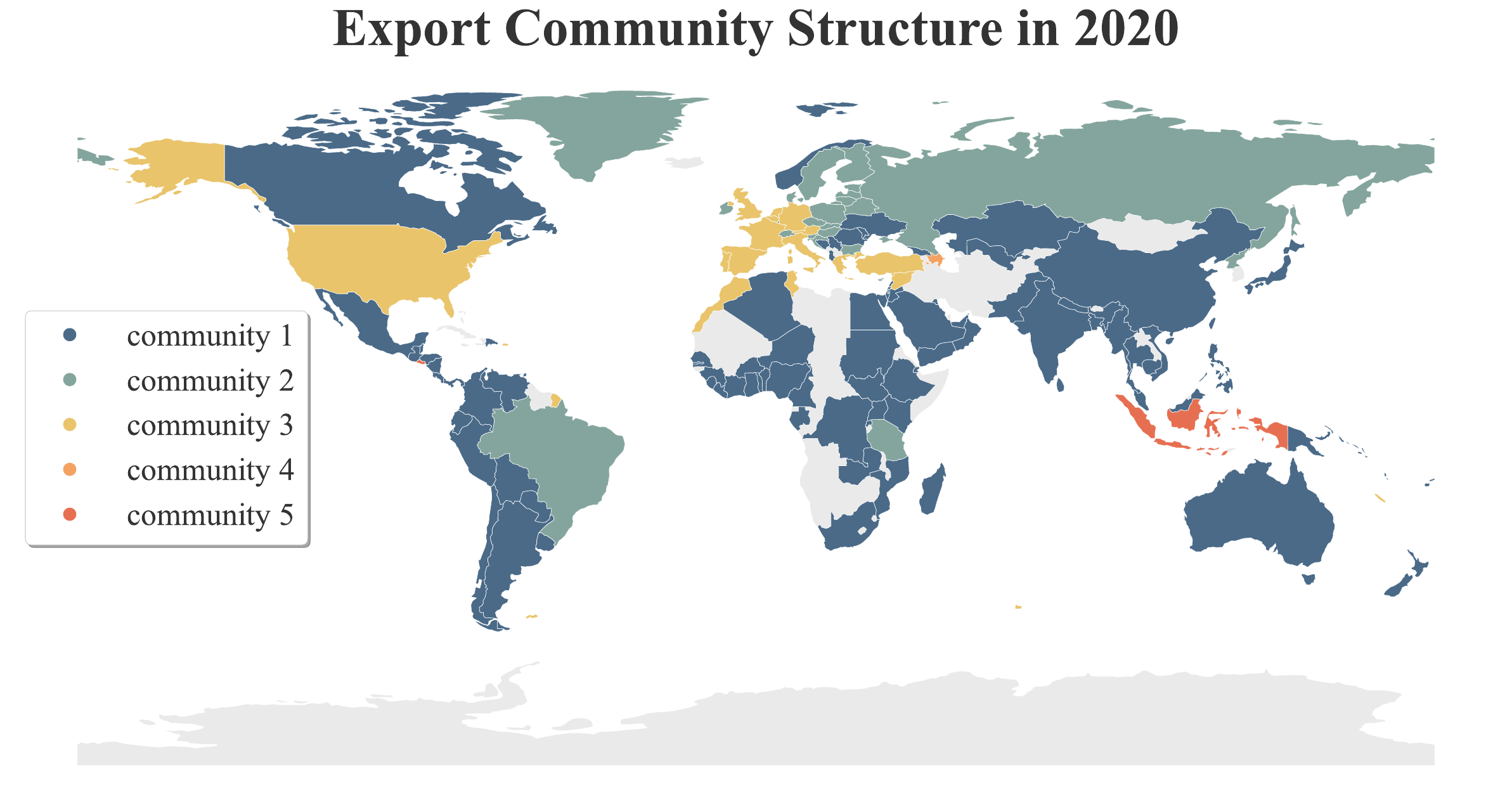}}
\subfigure[2021]{\includegraphics[height=2.6cm,width=4.5cm,angle=0]{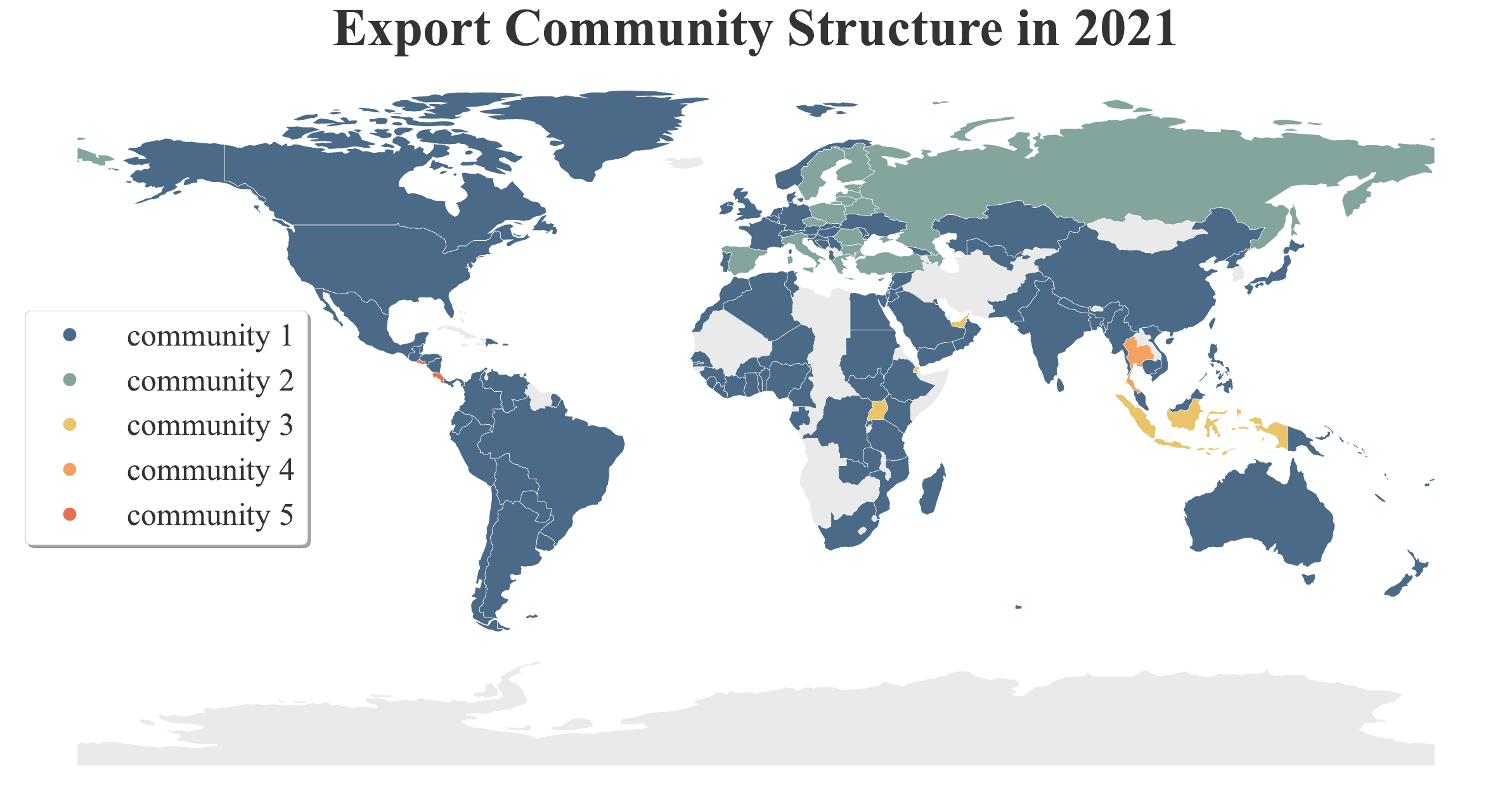}}
\subfigure[2022]{\includegraphics[height=2.6cm,width=4.5cm,angle=0]{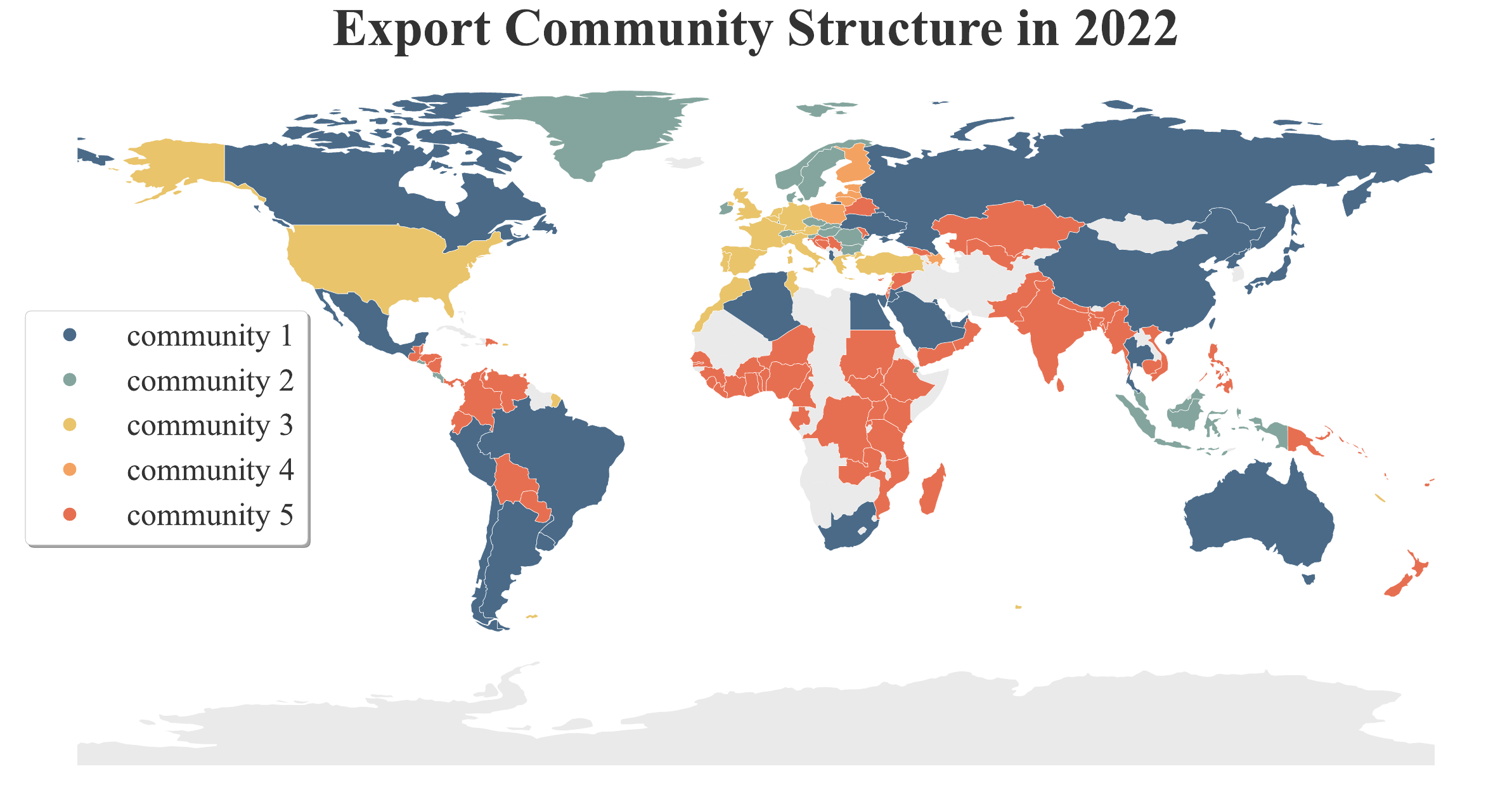}}
\subfigure[2023]{\includegraphics[height=2.6cm,width=4.5cm,angle=0]{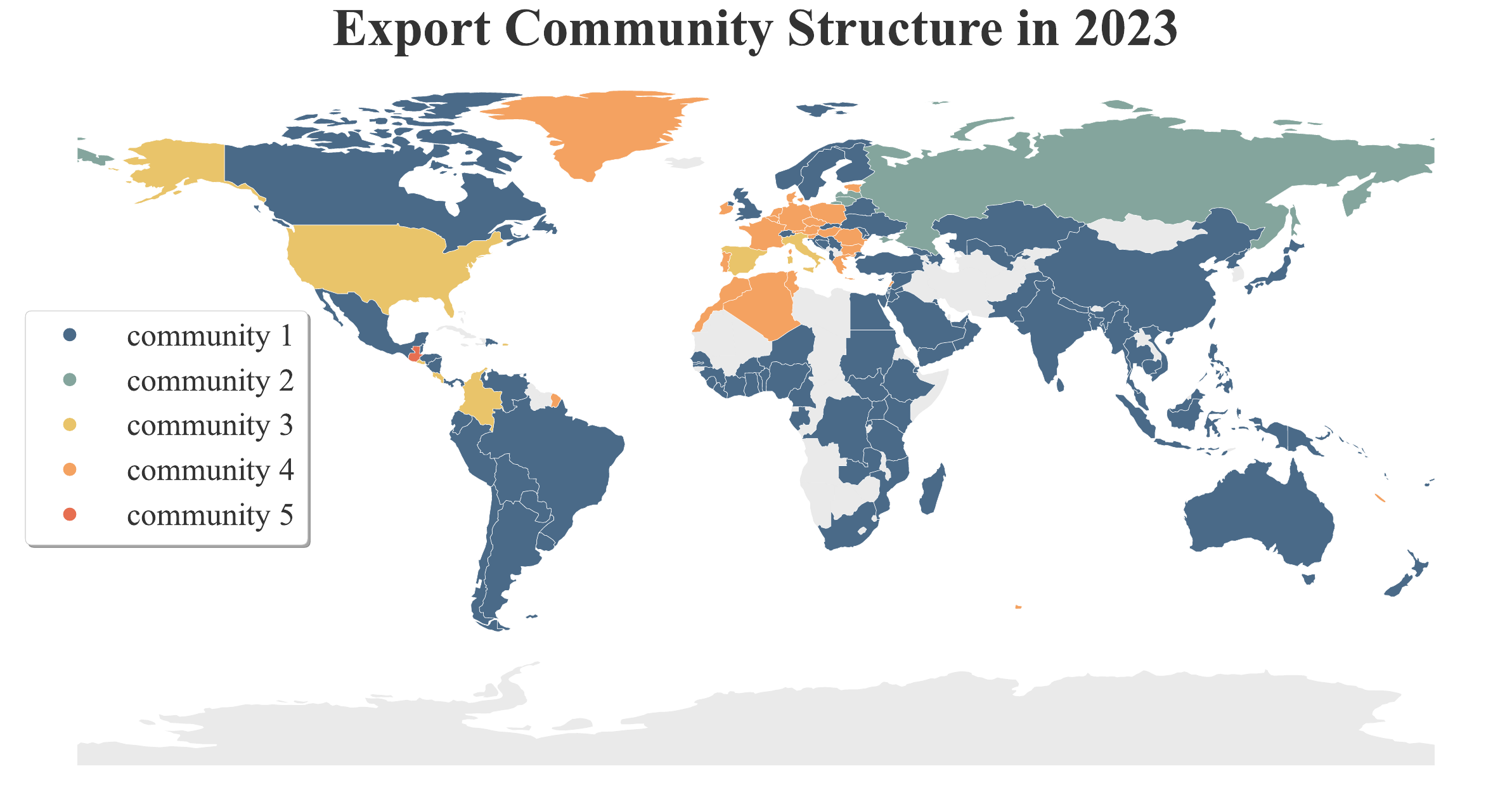}}
\caption{Export community structures of the olive-oil trade network from 2019 to 2023. The countries in the same export communities are marked with the same color.}
\label{fig4}
\end{figure*}

\textbf{Evolution of export communities.} The community structure of the olive-oil export trade network is shown on the map, where countries in the same community are marked with the same color; see Figure~\ref{fig4}. It is worth noting that the major olive-oil-producing and exporting countries are highly concentrated in the Mediterranean region. Based on Figure~\ref{fig4}, we make the following observations.

First, over the five-year period, major olive-oil-producing countries in the Mediterranean region, such as Spain and Italy, exhibit highly similar export patterns and are generally grouped into the same community. Taking the 2019 community structure as an example, Spain, Italy, and France are assigned to Community 3. By contrast, China and most developing countries in Latin America, Africa, and the Asia-Pacific region are assigned to Community 1, exhibiting export trade patterns distinct from those of the core olive-oil-producing region.

Second, the community structures in 2019 and 2020 are broadly similar and balanced. In 2021, however, except for Russia and several European countries assigned to Community 2, most countries are grouped into Community 1. In 2022, the export community structure of olive oil changes again. Compared with the relatively homogeneous pattern in 2021, the community partition in 2022 becomes clearer, with Community 3 consisting mainly of olive-oil-producing and exporting countries, while countries in Africa and Southeast Asia are more often assigned to Community 5. These changes might be related to the lingering effects of the COVID-19 pandemic.

Third, the export communities of some individual countries change substantially. For example, before 2022, Russia is consistently grouped into the same community as its European neighbours. In 2022, however, its community membership changes. By 2023, Community 2, to which Russia belongs, contains only three countries. This change might also be related to the Russia-Ukraine conflict.

To further illustrate the validity of the detected communities, we compare the olive-oil export patterns of two representative country pairs, Germany-France and the United States-Russia. The results are consistent with their community memberships, and the details are deferred to Appendix~\ref{sec7}.

\begin{figure*}[!h]{}
\centering
\subfigure[2019]{\includegraphics[height=2.6cm,width=4.5cm,angle=0]{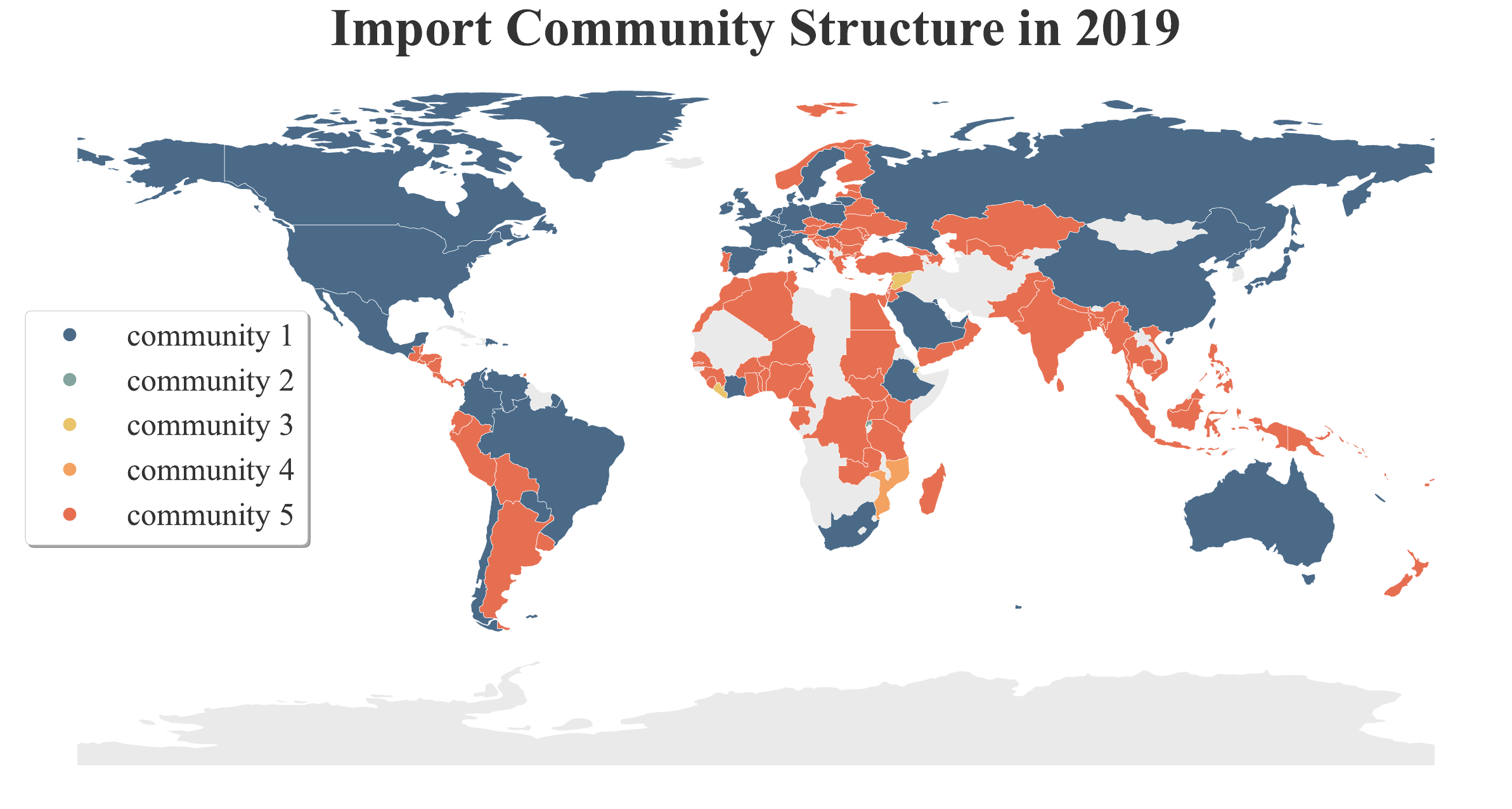}}
\subfigure[2020]{\includegraphics[height=2.6cm,width=4.5cm,angle=0]{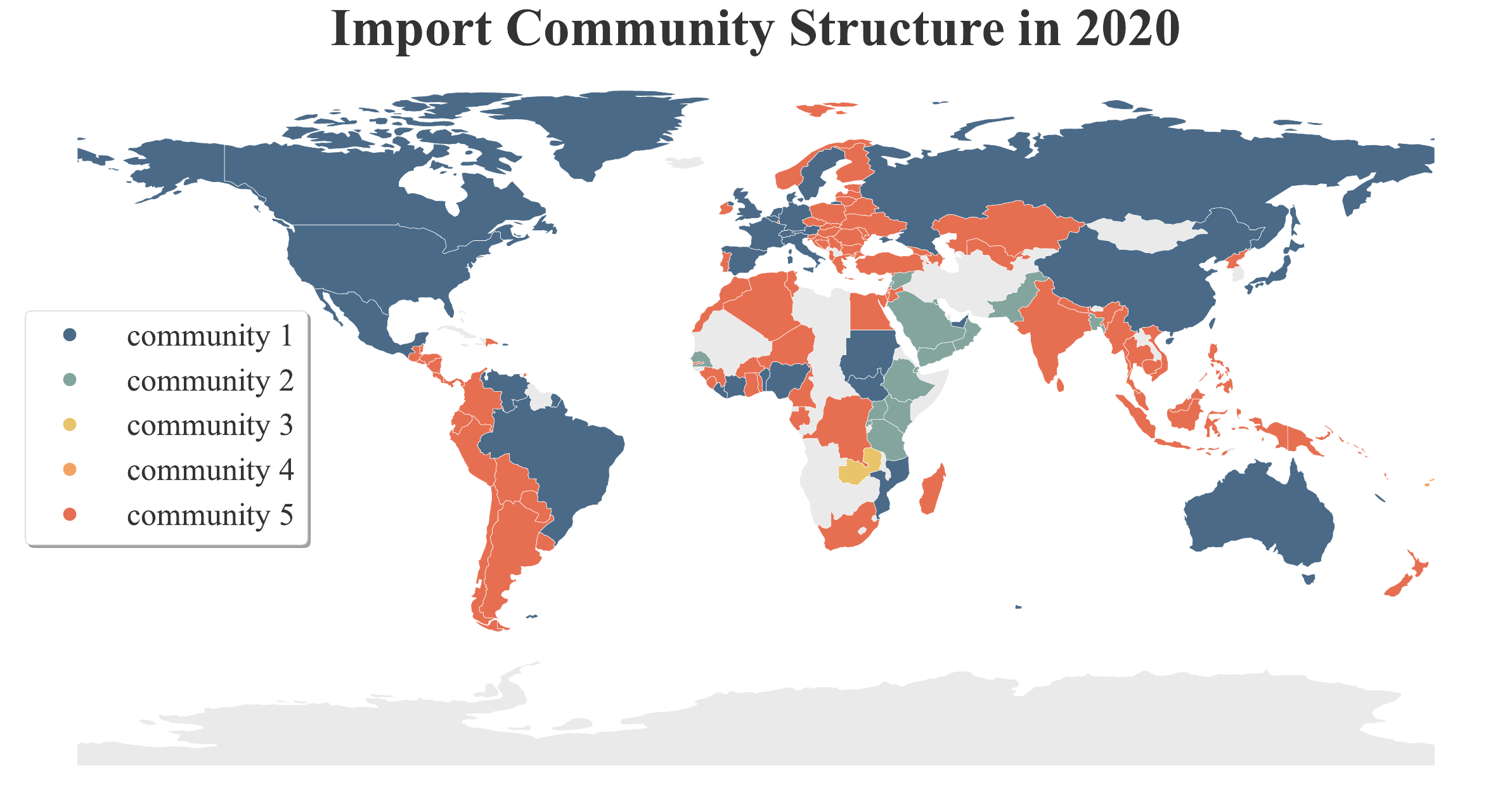}}
\subfigure[2021]{\includegraphics[height=2.6cm,width=4.5cm,angle=0]{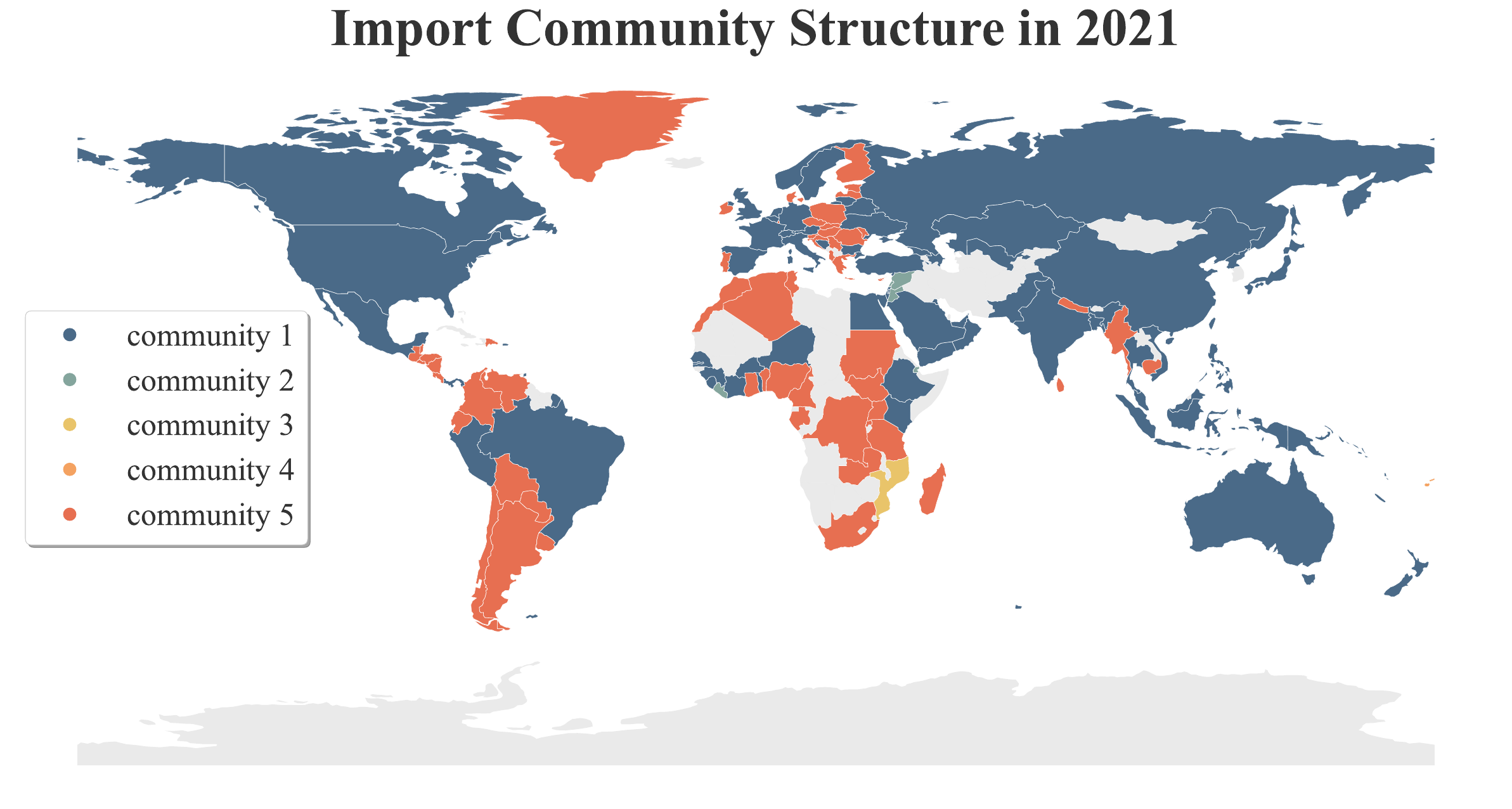}}
\subfigure[2022]{\includegraphics[height=2.6cm,width=4.5cm,angle=0]{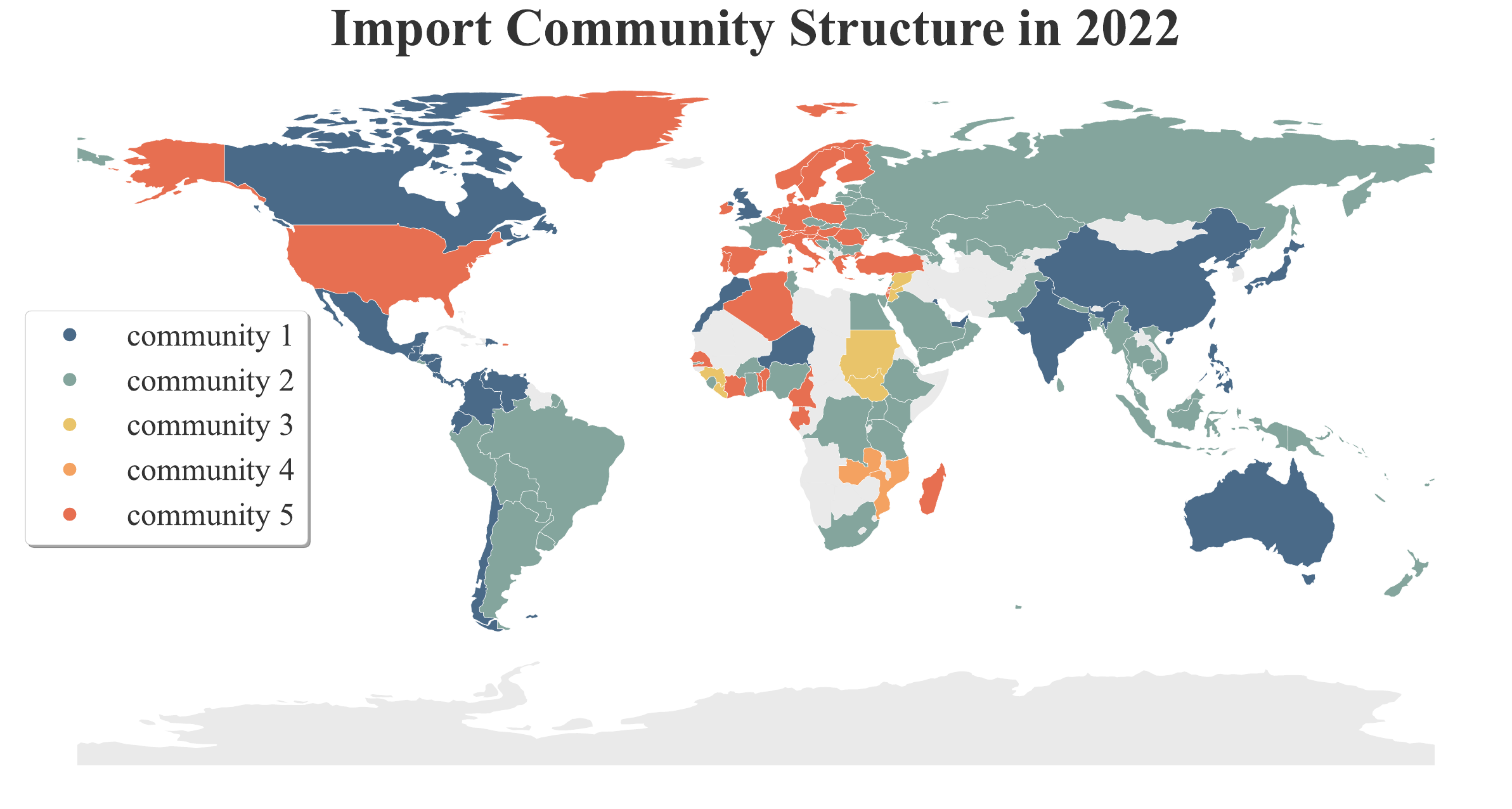}}
\subfigure[2023]{\includegraphics[height=2.6cm,width=4.5cm,angle=0]{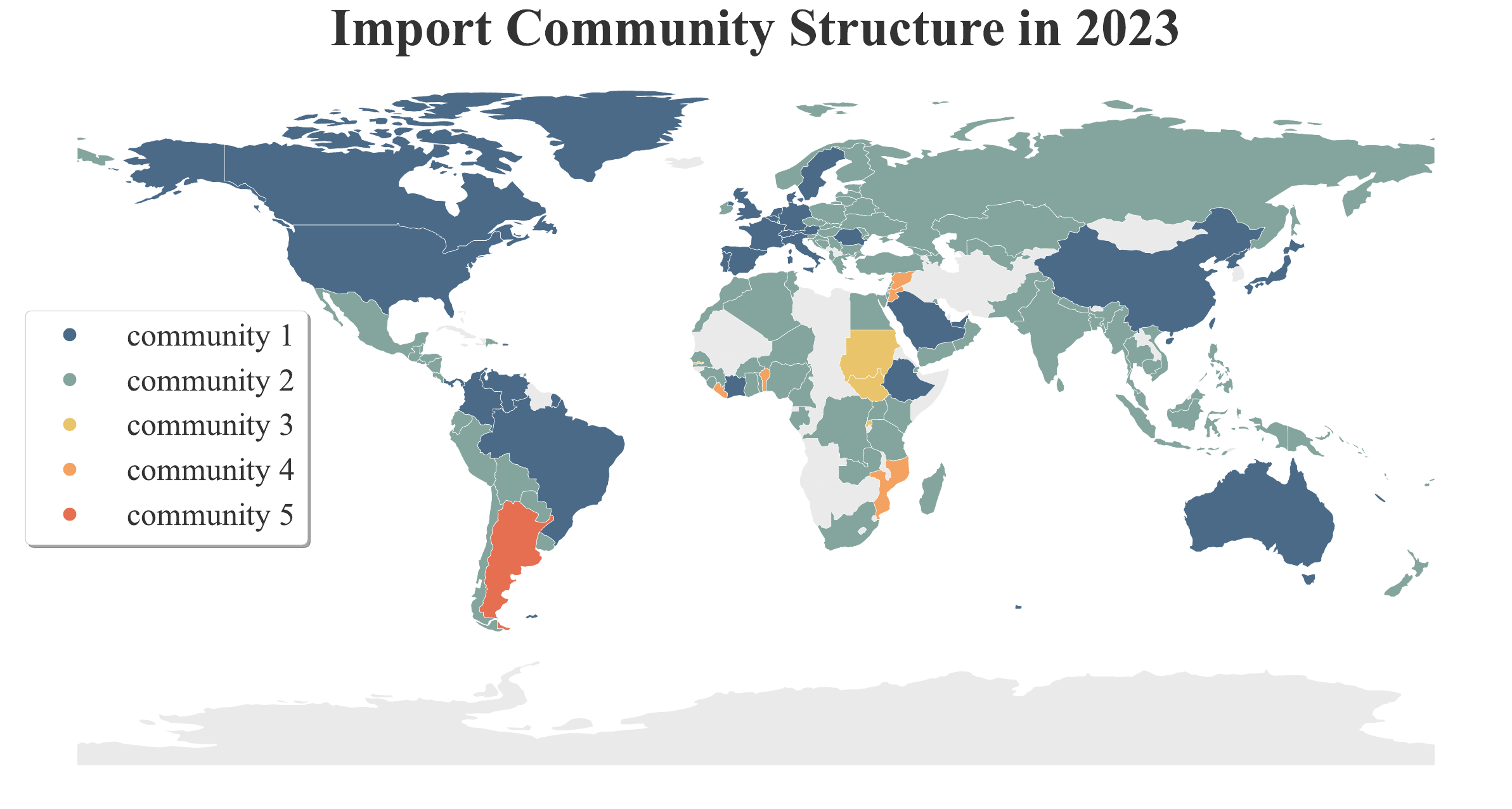}}
\caption{Import community structures of the olive-oil trade network from 2019 to 2023. The countries in the same import communities are marked with the same color.}
\label{fig8}
\end{figure*}

\textbf{Evolution of import communities.} 
The community structure of the olive-oil import trade network is shown on the map, where countries in the same community are marked with the same color; see Figure~\ref{fig8}. Compared with the export trade network, the import trade network exhibits a different community structure and pattern of community evolution. We summarize the main findings as follows.

First, from 2019 to 2021, the community structure of the global olive-oil import trade network remains relatively stable. Community 1 includes major olive-oil importing countries such as China, the United States, Germany, and Australia, and forms the largest group of countries with similar import patterns. By contrast, Community 5 mainly comprises developing countries in North Africa, Southeast Asia, Latin America, and parts of the Middle East, forming a smaller and more regional import community. China, the United States, and Russia are also consistently grouped into the same community.

Second, in 2022, the community structure of the olive-oil import trade network undergoes a noticeable adjustment. China, the United States, and Russia are assigned to three different communities: China is grouped together with Asia--Pacific countries such as Japan and Australia, as well as countries in the Americas such as Mexico; the United States is clustered with European countries such as Spain, Italy, and Germany; and Russia is grouped with Middle Eastern countries such as Egypt and Saudi Arabia, together with South American countries such as Brazil.

\subsection{Comparison with other methods}\label{subsec7}
To further evaluate the effectiveness of MuDySC, we compare its export communities for olive oil (shown in Figure~\ref{fig4}) with those obtained by the static method. The static method applies spectral clustering to each static network using only that year's olive-oil trade data, without incorporating inter-layer information from other vegetable-oil products or temporal information across adjacent years. The resulting export community structures under the static method are shown in Figure~\ref{fig10}.

As shown in Figure~\ref{fig10}, the static method clusters most countries into a single large community (Community~1), and therefore fails to reveal meaningful community structure. In contrast, the proposed MuDySC method captures dynamic changes in community structure and yields a more interpretable partition.

A direct comparison between MuDySC and PisCES, a single-layer dynamic community detection method that incorporates temporal information but not inter-layer information, is provided in Appendix~\ref{mudysc-pisces-comparison}.

\begin{figure*}[!h]{}
\centering
\subfigure[2021]{\includegraphics[height=2.6cm,width=4.5cm,angle=0]{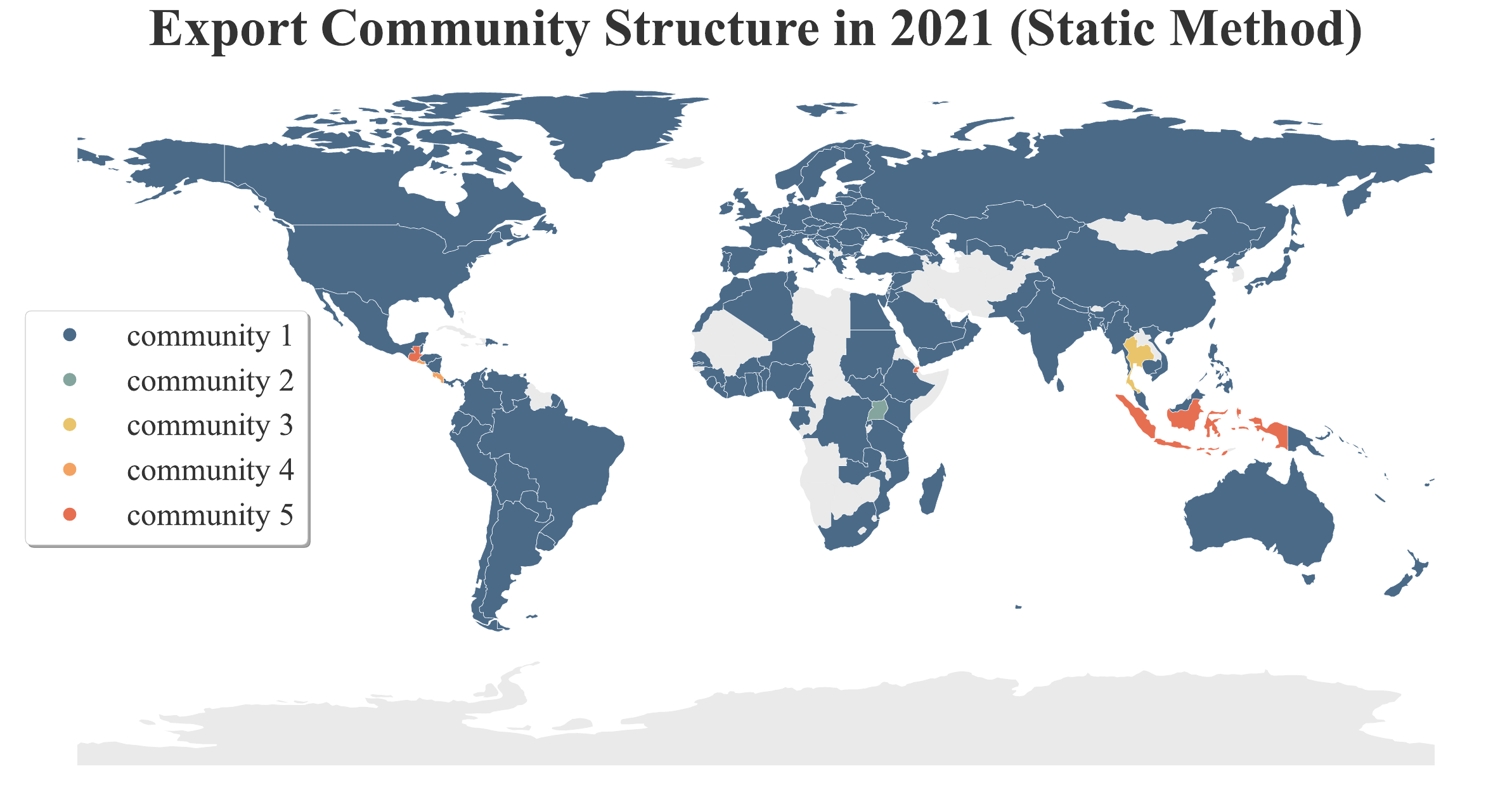}}
\subfigure[2022]{\includegraphics[height=2.6cm,width=4.5cm,angle=0]{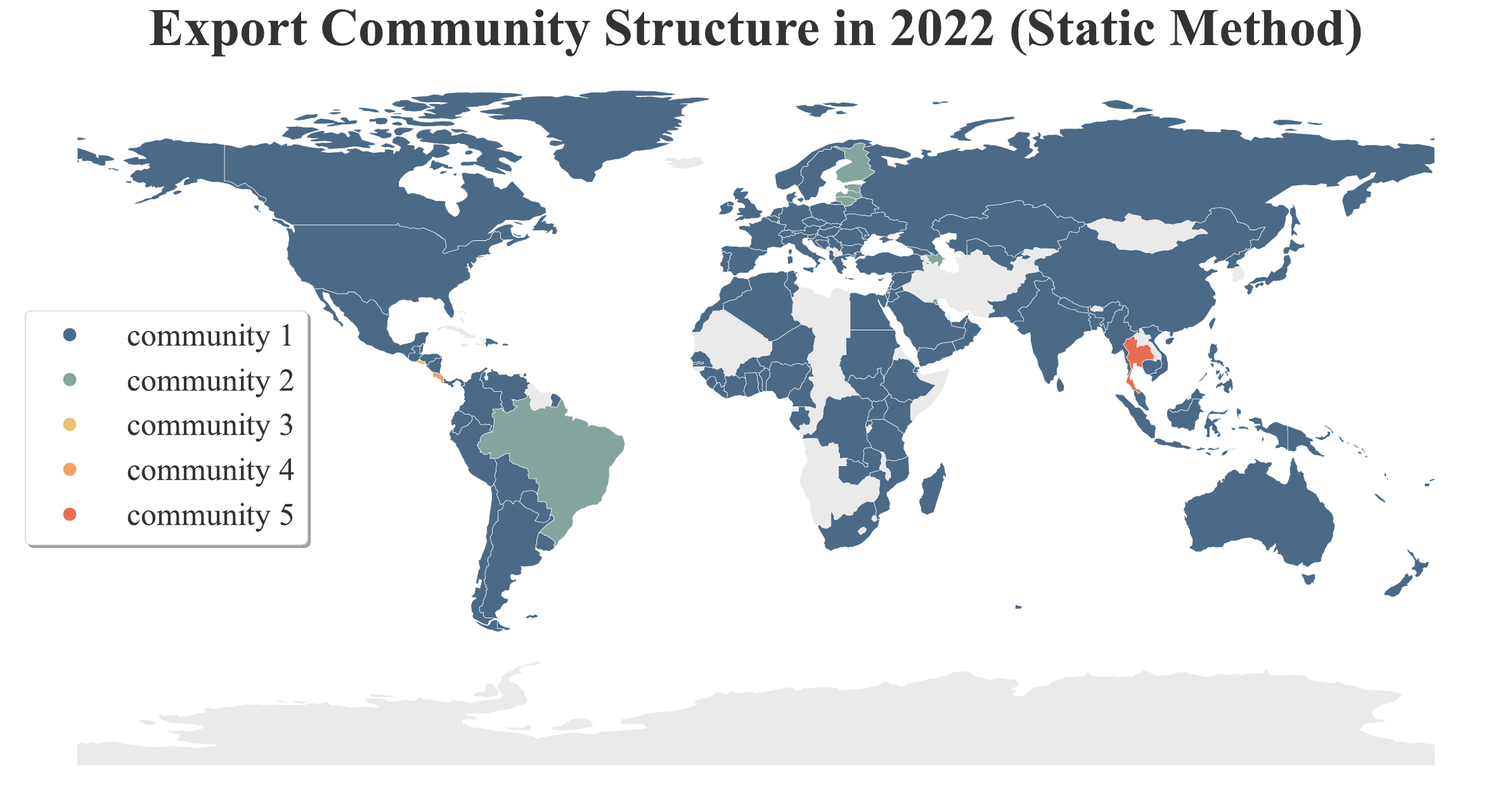}}
\subfigure[2023]{\includegraphics[height=2.6cm,width=4.5cm,angle=0]{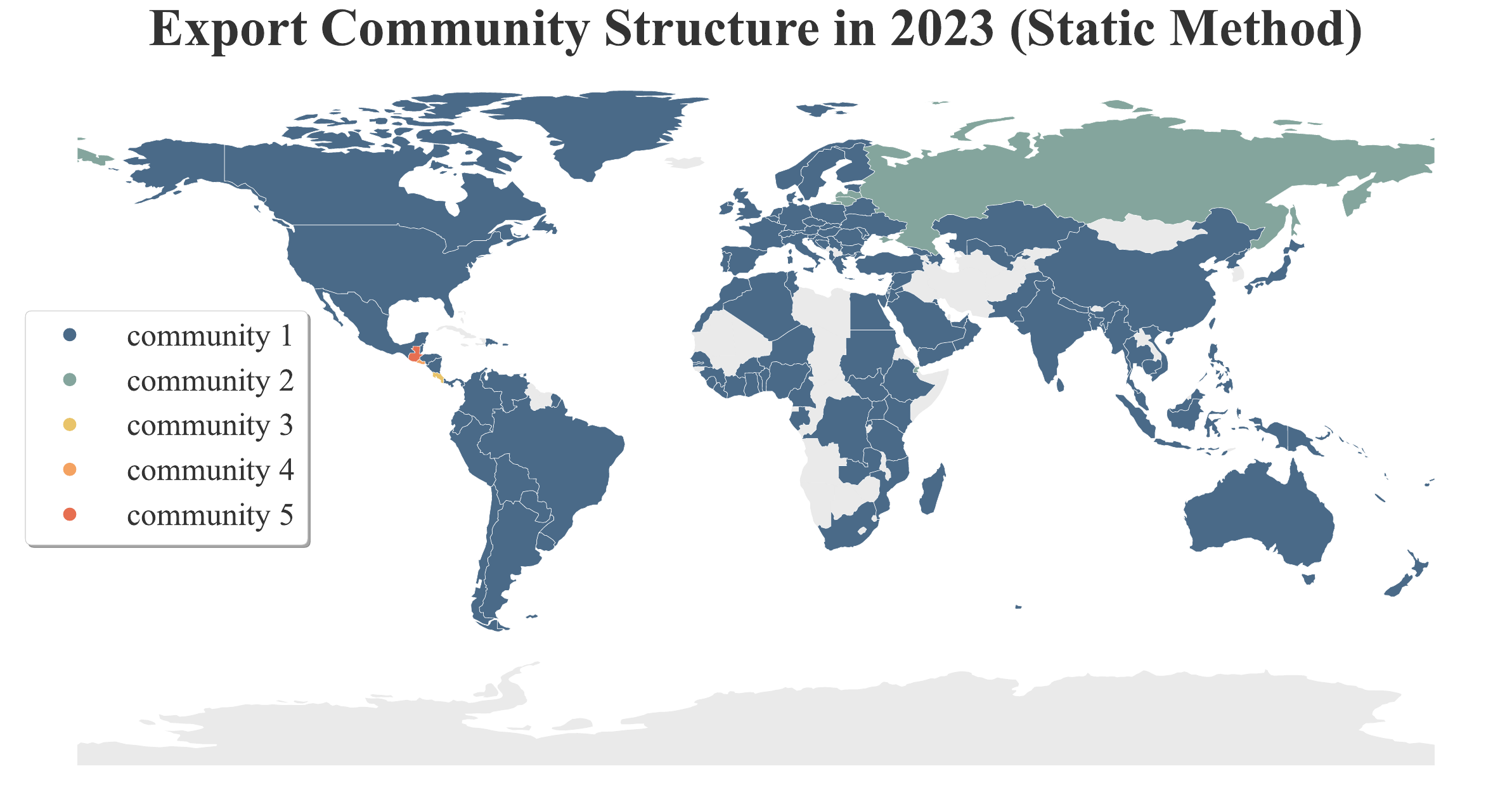}}
\caption{Community structure of olive oil export trade networks under the static method (2021-2023).}
\label{fig10}
\end{figure*}

\section{Simulations}\label{sec3}

In this section, we evaluate the effectiveness of the proposed MuDySC through numerical experiments. In particular, we study the effects of the number of time points \(T\) and the number of layers \(M\) on the clustering performance of MuDySC and the competing methods.

The multilayer-dynamic network is generated as follows, where each static network follows the well-known stochastic block model. Specifically, $[A_{m,t}]_{i,j}$, i.e., the $(i,j)$th entry of $A_{m,t}$ is generated according to
\begin{equation*}
[A_{m,t}]_{i,j} \sim \text{Bernoulli}(B_{[z_{m,t}]_i,[z_{m,t}]_j}), \quad [A_{m,t}]_{j,i}=\quad [A_{m,t}]_{i,j}; \quad j>i \quad i,j \in \{1,...,n\};  
\end{equation*}
where $B \in [0,1]^{K \times K}$ denotes the community connectivity matrix and for simplicity, we assume that \(B\) has off-diagonal entries equal to \(0.1\) and diagonal entries equal to \(0.4\); \(z_{m,t}\in\{1,\ldots,K\}^n\) is the vector of community labels for the \(n\) nodes in the network with layer \(m\) and time \(t\), and $[z_{m,t}]_i$ denotes the $i$th entry of $z_{m,t}$. We assume that the initial community label vector \(z_{m,1}\) is the same across all layers and is balanced, in the sense that each community contains \(n/K\) nodes. To incorporate community variation, we let the community labels in layer \(m\) evolve independently over time according to

\begin{equation*}
z_{m,t+1} =
\begin{cases} 
z_{m,t} & \text{with probability } 1-r, \\
\text{Multinomial}\left( \frac{1}{K}, \dots, \frac{1}{K} \right) & \text{otherwise.}
\end{cases}
\end{equation*}
where $r$ denotes the probability that a node changes its community membership at time $t+1$.

We compare the performance of the proposed MuDySC with the following three methods, which only used partial information of the multilayer-dynamic network. 

\begin{itemize}
    \item \textbf{Static \citep{rohe2011spectral}:} The method that applies the spectral clustering on each static network, respectively. 
    \item \textbf{PisCES \citep{liu2018global}:} The counterpart of MuDySC that performs eigenvector smoothing only within each dynamic network, without incorporating layer-wise information.
    \item \textbf{StaMuSC:} The counterpart of MuDySC that performs eigenvector smoothing only within the multilayer network at each time point, without incorporating temporal information.
\end{itemize}

\begin{figure*}[!h]{}
\centering
\subfigure[Case I: $n=50, K=2$]
{\includegraphics[height=4cm,width=14cm,angle=0]{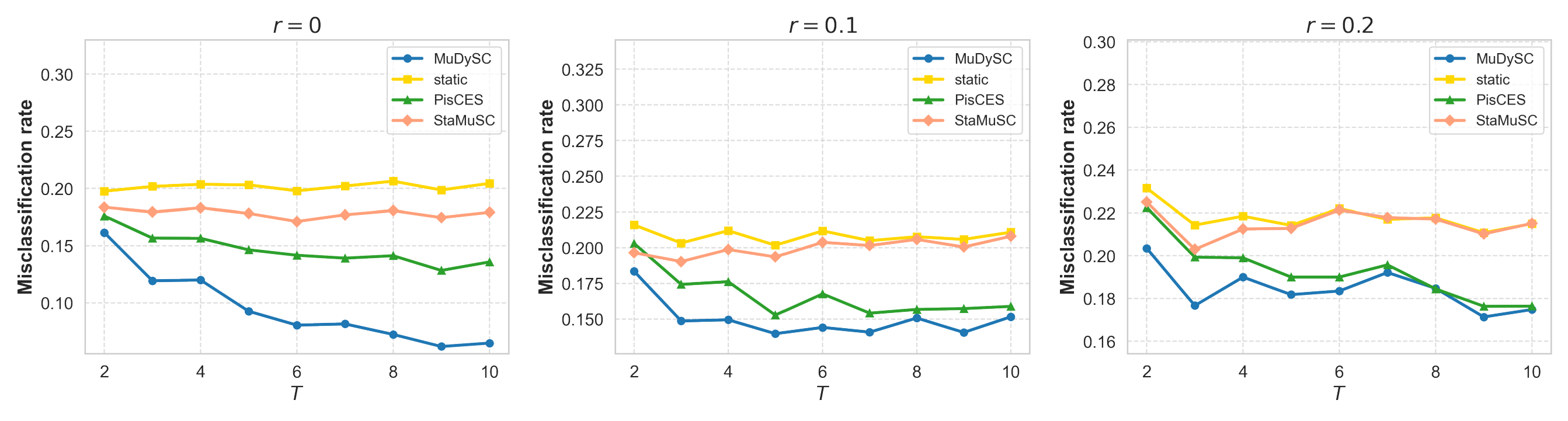}}
\subfigure[Case II: $n=100, K=3$]
{\includegraphics[height=4cm,width=14cm,angle=0]{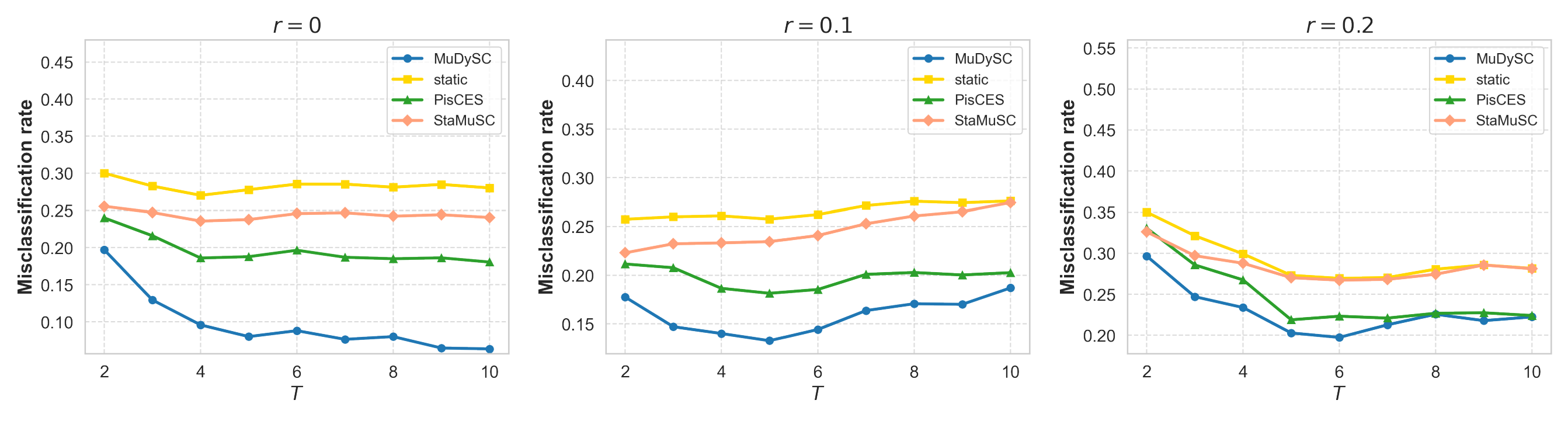}}
\caption{The average misclassification rates of four methods as $T$ increases under Cases I and II. For each case, we consider different probabilities \(r \in \{0, 0.1, 0.2\}\).}
\label{fig1}
\end{figure*}

We conduct the following Experiments 1 and 2 to assess the clustering performance of the proposed method MuDySC and three benchmark methods, one examines the effect of the number of time points \(T\), and the other investigates the effect of the number of layers \(M\). The clustering performance is measured by the average misclassification rate (i.e., the proportion of misclassified nodes) over each individual network. For each experiment, we consider two parameter settings:
\begin{itemize}
    \item Case I: The number of nodes $n=50$ and the number of communities $K=2$.
    \item Case II: The number of nodes $n=100$ and the number of communities $K=3$.
\end{itemize}
For both settings, we consider different probabilities \(r \in \{0, 0.1, 0.2\}\) that a community label of a node changes from the previous time point. 

\textbf{Experiment 1: Effect of the number of time points $T$.} In this experiment, we fix number of layers $M=2$, and study how the performance of four methods varies with $T$.  
The averaged misclassification rates over 10 replications under Cases I and II are shown in Figure \ref{fig1}. We observe that the proposed method MuDySC achieves the lowest misclassification rate in all cases, demonstrating its superior clustering performance. The methods that incorporate temporal information, MuDySC and PisCES, initially benefit from an increasing number of time points, leading to a decrease in misclassification rate; see Figure \ref{fig1}(a) and Figure \ref{fig1}(b) with $r=0$. However, the misclassification rate may increase again once \(T\) becomes moderately large; see Figure \ref{fig1}(b) with $r=0.1$ and $r=0.2$. 
When $r$ increases, the community memberships become more heterogeneous over time, which weakens the shared community information among different layers. Consequently, the benefit of layer-wise smoothing decreases. Therefore, the performance gap between MuDySC and PisCES gradually narrows, since PisCES does not incorporate layer-wise information.
In contrast with MuDySC and PisCES, StaMuSC and static  maintain nearly constant error as $T$ increases, as they do not incorporate the temporal information.  
The static method, which relies solely on individual network information, performs the worst with the highest misclassification rate. Furthermore, we observe that as $r$ increases, the performance of all methods deteriorates to varying degrees, but MuDySC consistently achieves the best results.

\begin{figure*}[!h]{}
\centering
\subfigure[Case I: $n=50, K=2$]
{\includegraphics[height=4cm,width=14cm,angle=0]{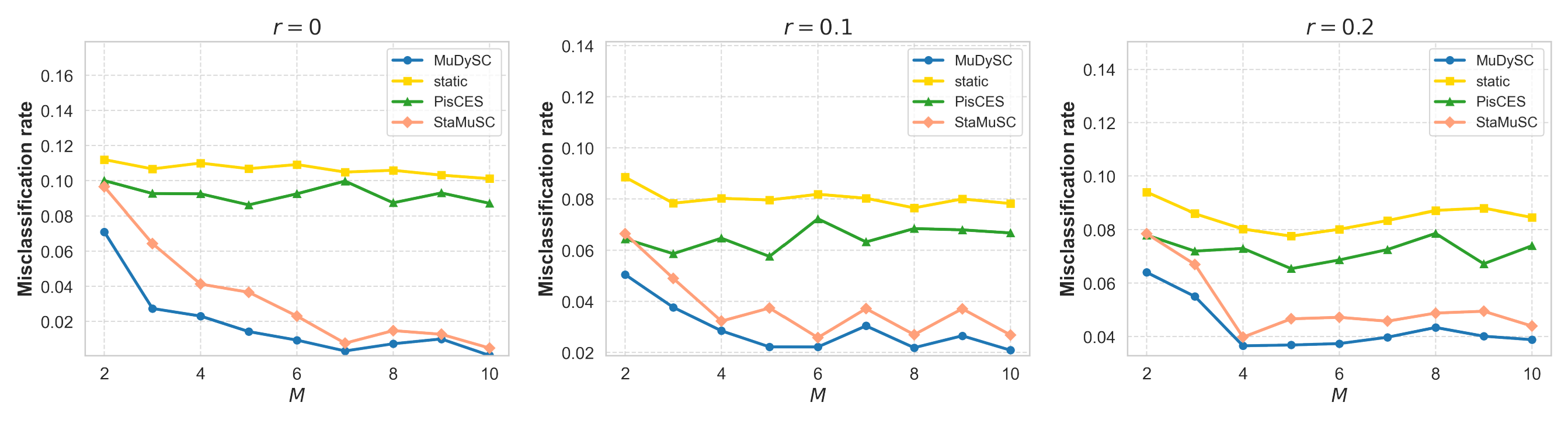}}
\subfigure[Case II: $n=100, K=3$]
{\includegraphics[height=4cm,width=14cm,angle=0]{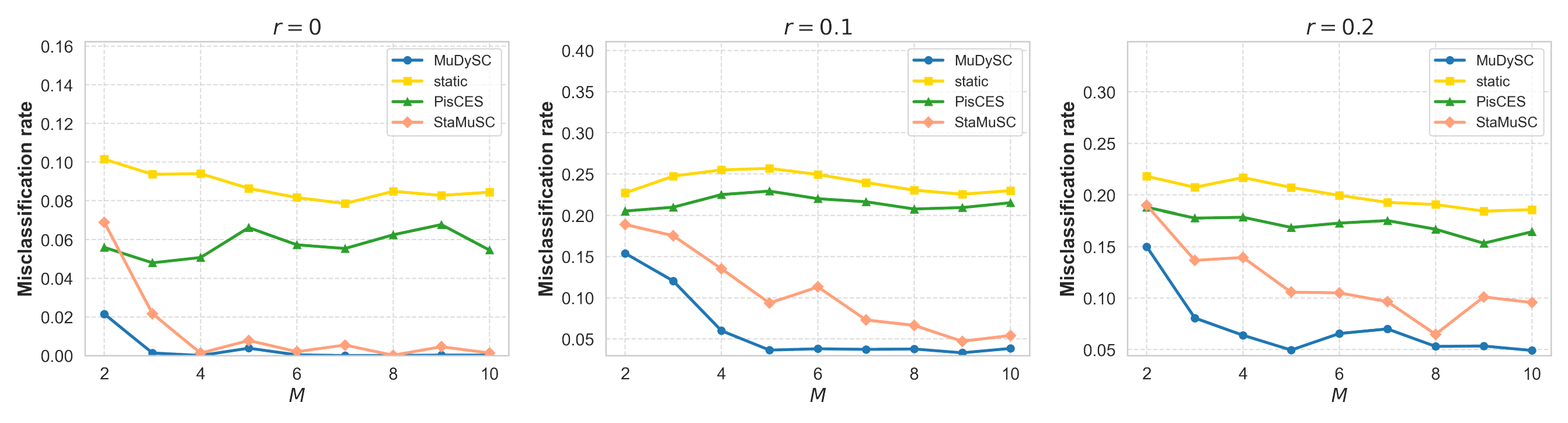}}
\caption{The average misclassification rates of four methods as $M$ increases under Cases I and II. For each case, we consider different probabilities \(r \in \{0, 0.1, 0.2\}\).}
\label{fig2}
\end{figure*}

\textbf{Experiment 2: Effect of the number of layers $M$.} In this experiment, we fix the number of time points $T=2$, and we examine how the performance of four methods varies with $M$.  
The averaged misclassification rates over 10 replications under Cases I and II are shown in Figure \ref{fig2}. 
We observe that, as the number of layers increases, the misclassification error of both MuDySC and StaMuSC decrease, indicating that incorporating layer-wise information improves community detection performance. 
However, when $r$ is large, the performance gap between MuDySC and StaMuSC becomes smaller, indicating that the benefit of temporal smoothing decreases when community structures vary substantially over time.
Among all methods, MuDySC achieves the lowest misclassification rate, followed by StaMuSC, while PisCES and the static show the worst performance.  
In addition, when $r$ becomes large, the performance of all methods declines, yet MuDySC remains the most robust.

To further investigate the sensitivity of MuDySC to the tuning parameters
$\alpha$ and $\beta$, we conducted a simulation experiment under different
levels of community variation and evaluated different combinations of
$\alpha$ and $\beta$ using the validation negative log-likelihood obtained
from network cross-validation. The detailed experimental settings and results
are provided in Appendix~\ref{parameter_sensitivity}.

\section{Conclusion}\label{sec5}

In this paper, we studied how community patterns in international trade evolve across products and over time through community detection in multilayer-dynamic networks. To address this problem, we proposed a novel method, MuDySC, which obtains smoothed versions of the eigenvector projection matrices for individual networks by enforcing smoothness across adjacent time points and across different layers at the same time point. We developed an efficient alternating iterative algorithm for solving the
resulting optimization problem and established both global and local convergence
results. The global convergence result guarantees convergence to the global
minimum under a sufficient condition on the tuning parameters, while the local
convergence result shows that, when the relevant eigenspaces are sufficiently
close and the algorithm is suitably initialized, convergence can be guaranteed
under substantially weaker conditions. The simulation experiments demonstrated the advantage of MuDySC over competing methods that use only partial information from the multilayer-dynamic network. In addition, we applied the proposed MuDySC to the FAO dataset and obtained interpretable and meaningful results on the evolution of communities in the olive-oil trade network, providing insights into the evolutionary patterns of international trade. To further accommodate substantial heterogeneity across network layers, we
developed TLC-MuDySC, a tensor-based layer-clustering extension that first
identifies structurally related layers and then applies MuDySC within the
estimated layer groups. Numerical experiments further illustrate the usefulness
of this extension under strong layer heterogeneity. The proposed framework can
also be extended to accommodate directed or weighted networks.

\setcounter{equation}{0}
\setcounter{figure}{0}
\renewcommand*{\thetheorem}{A\arabic{theorem}}
\renewcommand*{\thelemma}{A\arabic{lemma}}
\renewcommand*{\theproposition}{A\arabic{proposition}}
\renewcommand*{\thedefinition}{A\arabic{definition}}
\renewcommand*{\theequation}{A\arabic{equation}}
\renewcommand*{\thefigure}{A\arabic{figure}}
\renewcommand*{\thealgorithm}{A\arabic{algorithm}}
\renewcommand*{\theremark}{A\arabic{remark}}

\appendix
\section*{Appendix}
The Appendix is organized as follows.
Appendix~\ref{classification_table} provides the CPC product classification
used in the FAO data analysis.
Appendix~\ref{sec7} presents additional results for the FAO data analysis.
Appendix~\ref{parameter_sensitivity} provides a sensitivity analysis of the
tuning parameters $\alpha$ and $\beta$.
Appendices~\ref{sec8} and~\ref{sec9} provide the proofs of
Theorems~\ref{thm1} and~\ref{thm:local}, respectively.
Finally, Appendix~\ref{sec:tensor-extension} presents the proposed
TLC-MuDySC extension, including its latent-class model, tensor-based
layer-clustering algorithm, and numerical evaluation.

\section{CPC Product Classification}
\label{classification_table}

Table~\ref{tab1} presents the CPC product categories used to organize the FAOSTAT trade data, together with the number of products contained in each category.

\begin{table}[!h]
\begin{center}
\begin{minipage}{1.05\textwidth}
\caption{The CPC and the number of products in each category.}\label{tab1}%
\begin{tabular}{@{}cccc@{}}
\hline
Category & Number & Category & Number \\
\hline
Ice and Snow & 1 & Foodstuffs & 14 \\
Industrial Oils and Fats & 2 & Food Industry Waste & 75 \\
Textile Raw Materials & 6 & Animal Fodder & 29 \\
Industrial Raw Materials & 9 & Natural Rubber & 1 \\
Animal Products & 67 & Tobacco & 3 \\
Essential Oils and Resins/Balsams & 2 & Beverages, Spirits, and Vinegar & 10 \\
Meat Products and Fats & 147 & Vegetable Products & 178 \\
Dairy Products & 29 &  &  \\
\hline
\end{tabular}
\end{minipage}
\end{center}
\end{table}

\section{Additional results for the FAO data analysis}\label{sec7}

\subsection{Export Pattern Comparison}\label{sankey}

To further illustrate the validity of the detected communities, we analyze two pairs of countries in terms of their export patterns: Germany and France, and the United States and Russia. 

Note that the community structure shown in Figure~\ref{fig4} indicates that Germany and France are assigned to the same export community throughout the five-year period. Figures~\ref{fig5} and \ref{fig6} present Sankey diagrams of their export destinations and export values. The results show that, among the top 15 export destinations of Germany and France, 10 countries overlap, most of which are EU member states such as Austria and the Netherlands. This pattern may be explained by the strong economic integration within the European Union, which likely promotes agricultural trade among member states.

By contrast, Russia and the United States are never assigned to the same export community during the five-year period. The results show that the two countries differ substantially in their olive-oil export patterns. First, the export value of olive oil from the United States is more than 30 times that of Russia. Second, among their respective top 15 export destinations, only two countries overlap. Exports from the United States are concentrated mainly in nearby markets such as North America and the Caribbean, whereas Russian exports are dispersed across Central Asia, Eastern Europe, and East Asia, with a relatively small overall trade volume.

\begin{figure}[!htbp]
\centering
\includegraphics[height=6cm,width=12cm,angle=0]{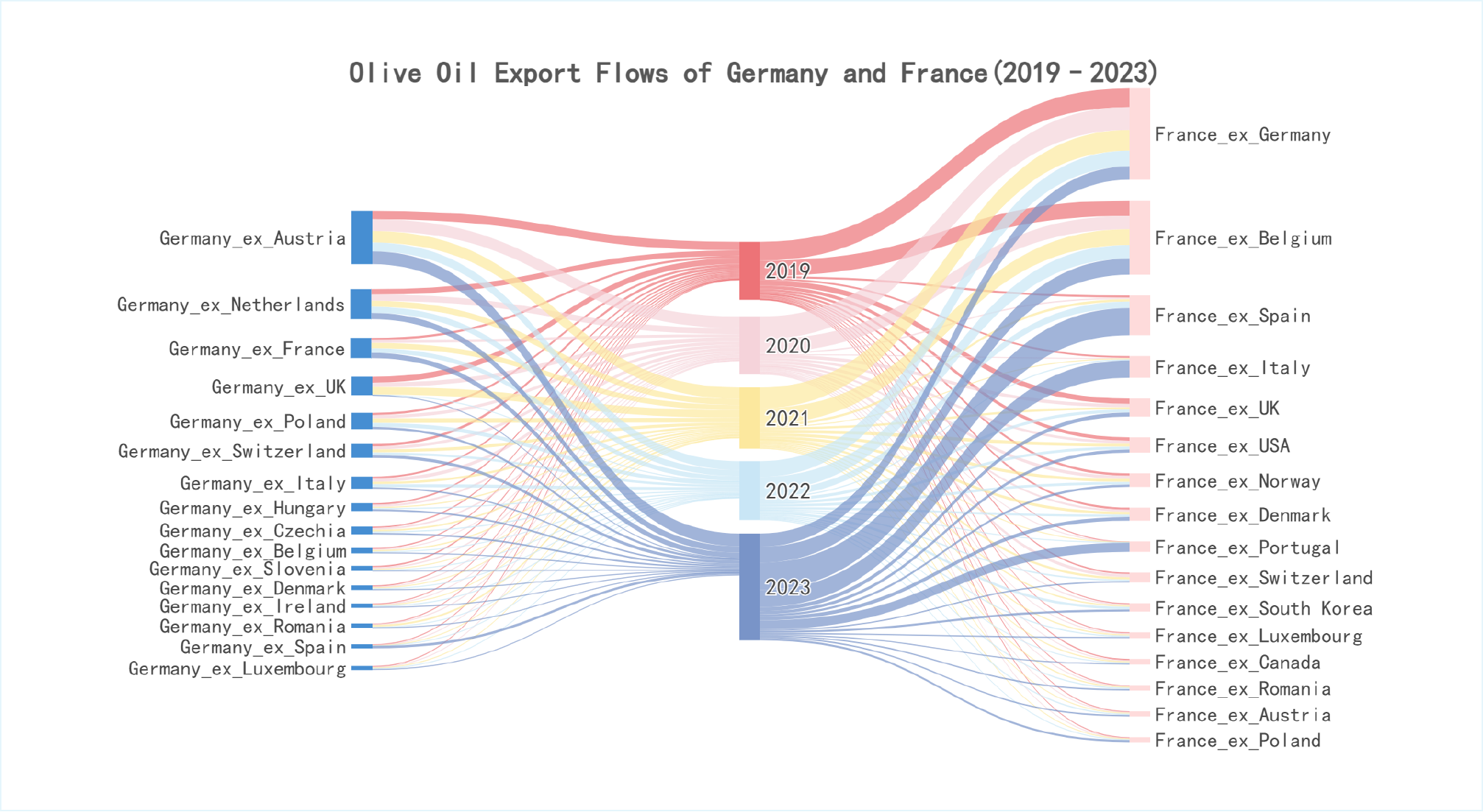}
\caption{Sankey diagram of the top 15 olive-oil export destinations of Germany and France. The labels on the left and right correspond to the top 15 export destinations of Germany and France over the five-year period, respectively, and the width of each flow is proportional to the export value.}
\label{fig5}
\end{figure}

\begin{figure}[!htbp]
\centering
\includegraphics[height=6cm,width=12cm,angle=0]{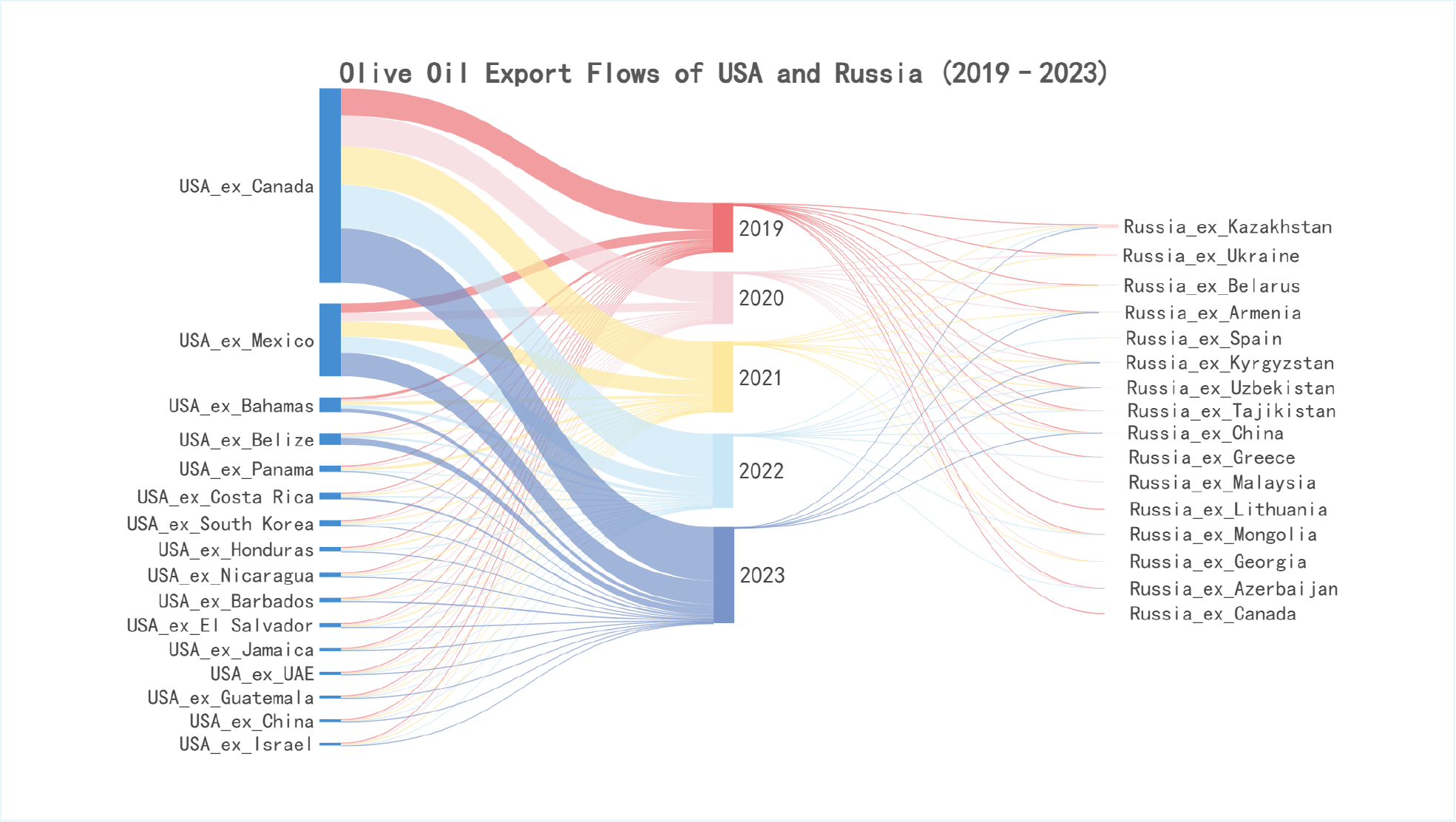}
\caption{Sankey diagram of the top 15 olive-oil export destinations of the United States and Russia. The labels on the left and right correspond to the top 15 export destinations of the United States and Russia over the five-year period, respectively, and the width of each flow is proportional to the export value.
}
\label{fig6}
\end{figure}

\subsection{Comparison between MuDySC and PisCES}\label{mudysc-pisces-comparison}

The complete community structures identified by PisCES are presented in Figs.~\ref{fig:pisces_import} and~\ref{fig:pisces_export} and are compared with the corresponding MuDySC results. Among the five years, the differences between the two methods were most evident in 2022 for the import network and in 2020 for the export network. For the 2022 import network, PisCES assigned most countries to two large communities of 76 and 46 countries, whereas MuDySC reorganized countries from these large communities and identified a separate 28-country community, including China and India, among others. For the 2020 export network, the community memberships of Belgium, the United States, and several other countries changed after cross-layer information was incorporated. Moreover, an eight-country community identified by PisCES was incorporated into a larger 22-country community by MuDySC. These differences illustrate that product-layer information affects both the identification of country groups that are less apparent in the olive-oil layer alone and the determination of community boundaries. Overall, these results show that incorporating layer-wise information can lead
to community partitions that differ meaningfully from those obtained using
temporal information alone. This advantage is further supported by our
simulation results, where MuDySC consistently achieves lower
misclassification rates than PisCES.

\begin{figure*}[!h]{}
\centering
\subfigure[2019]{\includegraphics[height=2.6cm,width=4.5cm,angle=0]{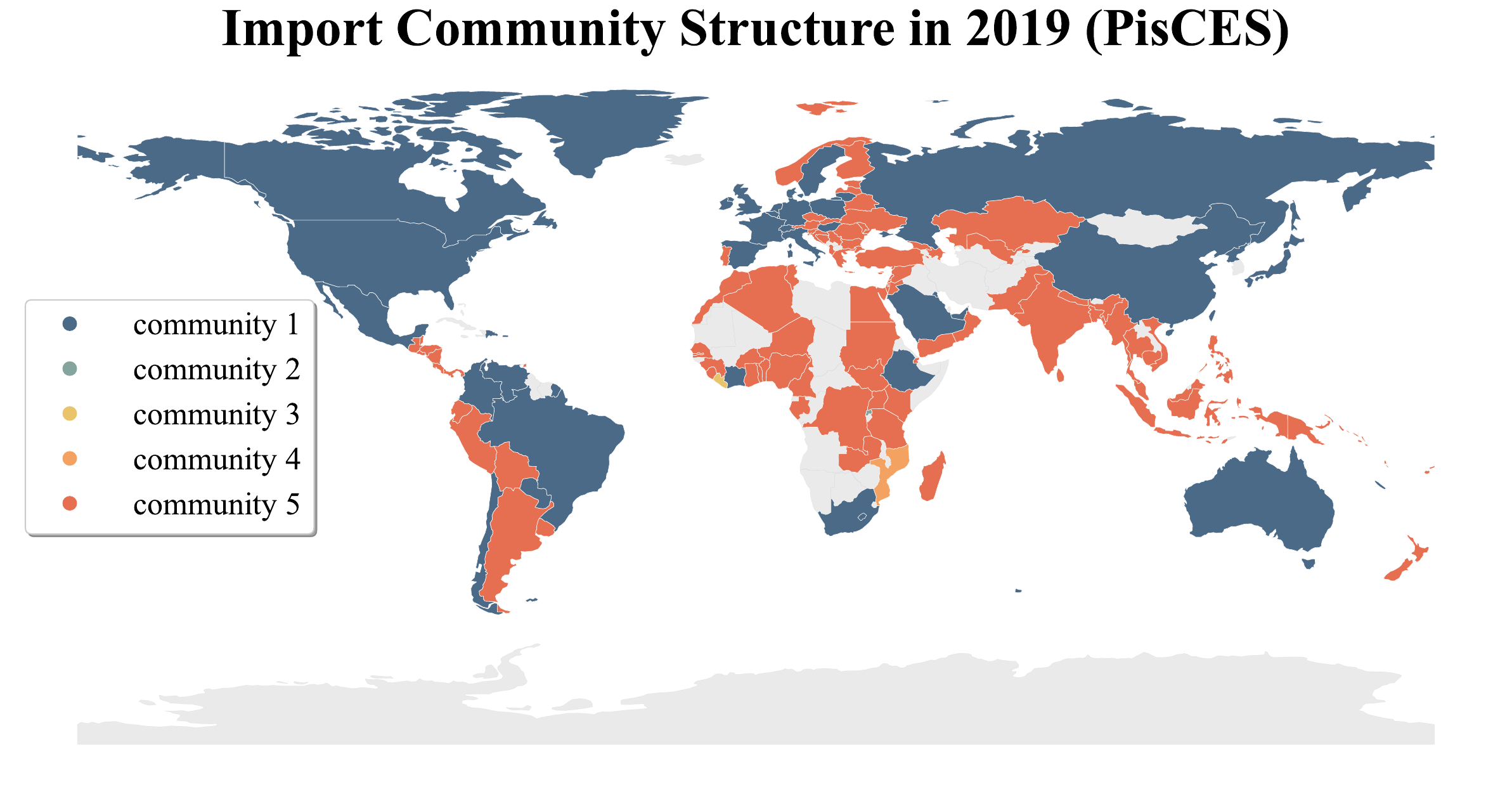}}
\subfigure[2020]{\includegraphics[height=2.6cm,width=4.5cm,angle=0]{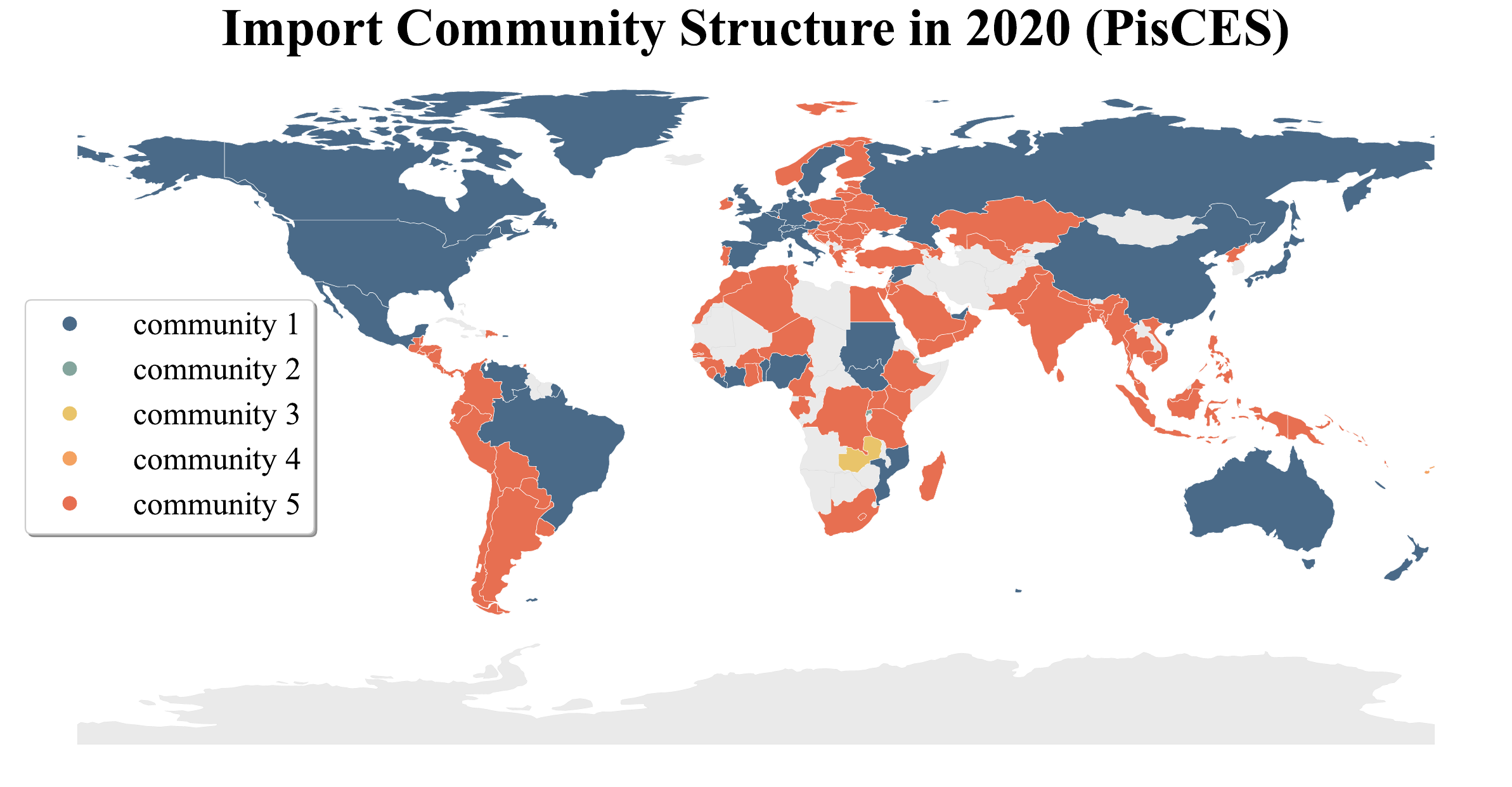}}
\subfigure[2021]{\includegraphics[height=2.6cm,width=4.5cm,angle=0]{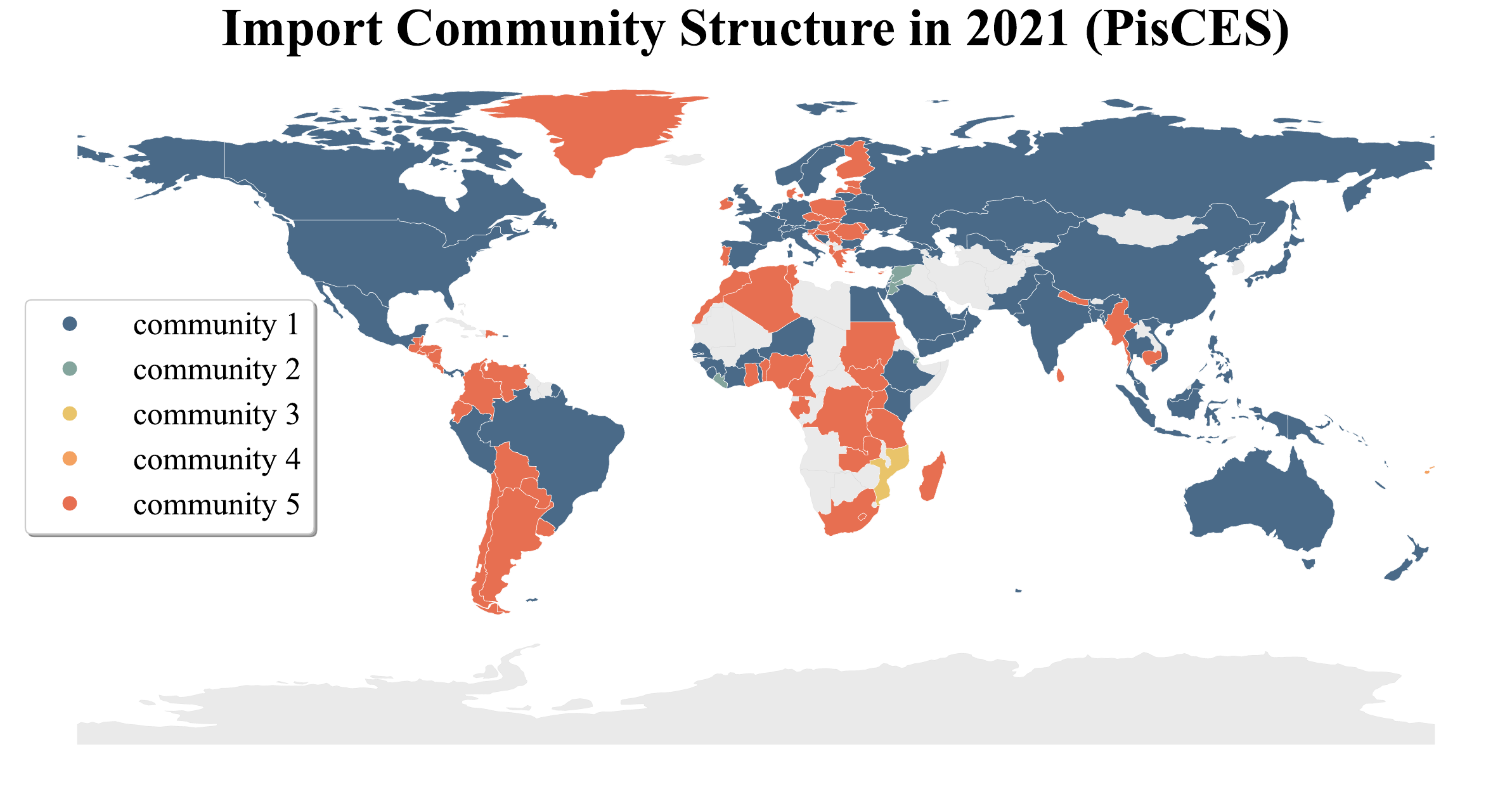}}
\subfigure[2022]{\includegraphics[height=2.6cm,width=4.5cm,angle=0]{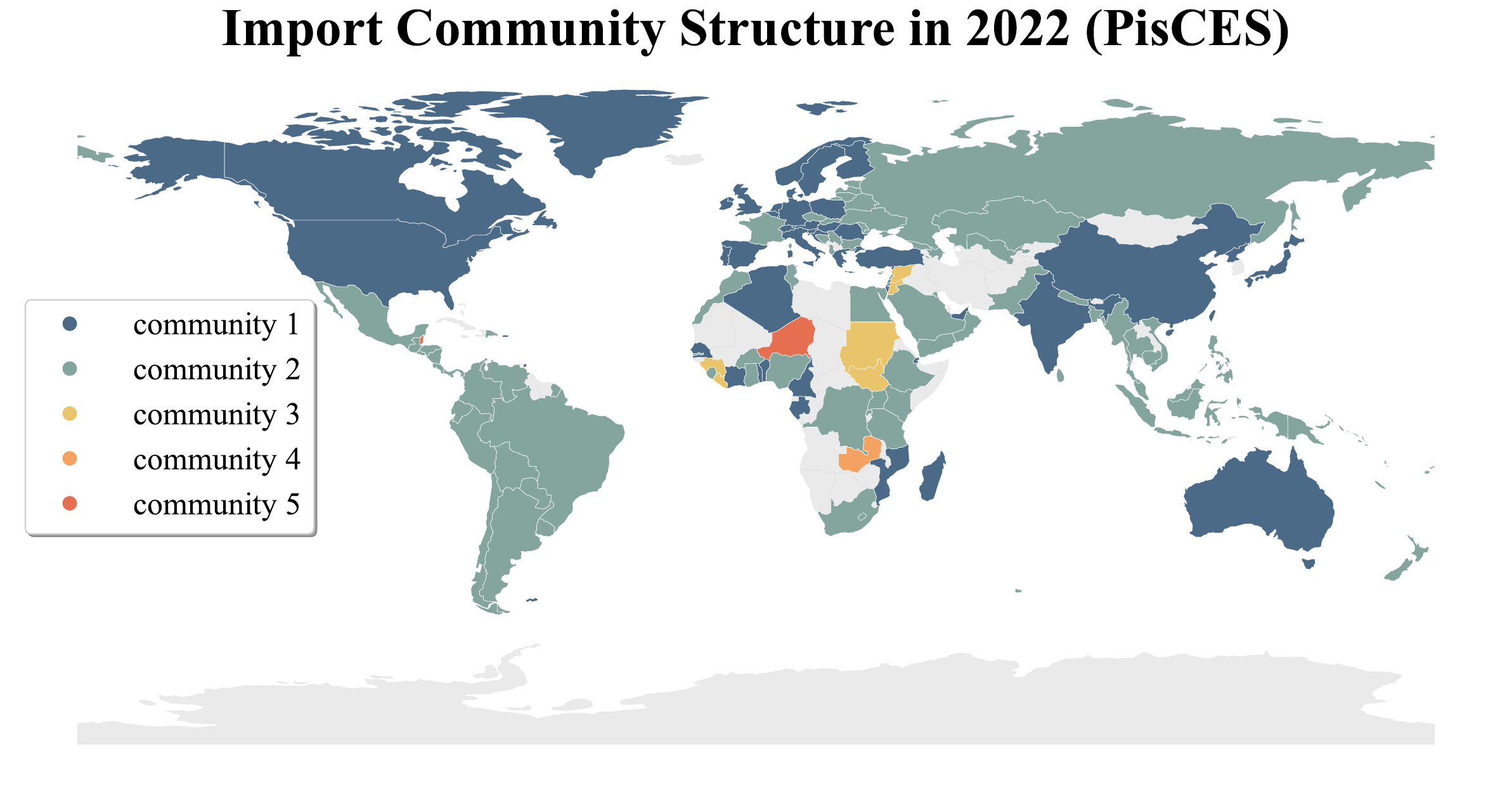}}
\subfigure[2023]{\includegraphics[height=2.6cm,width=4.5cm,angle=0]{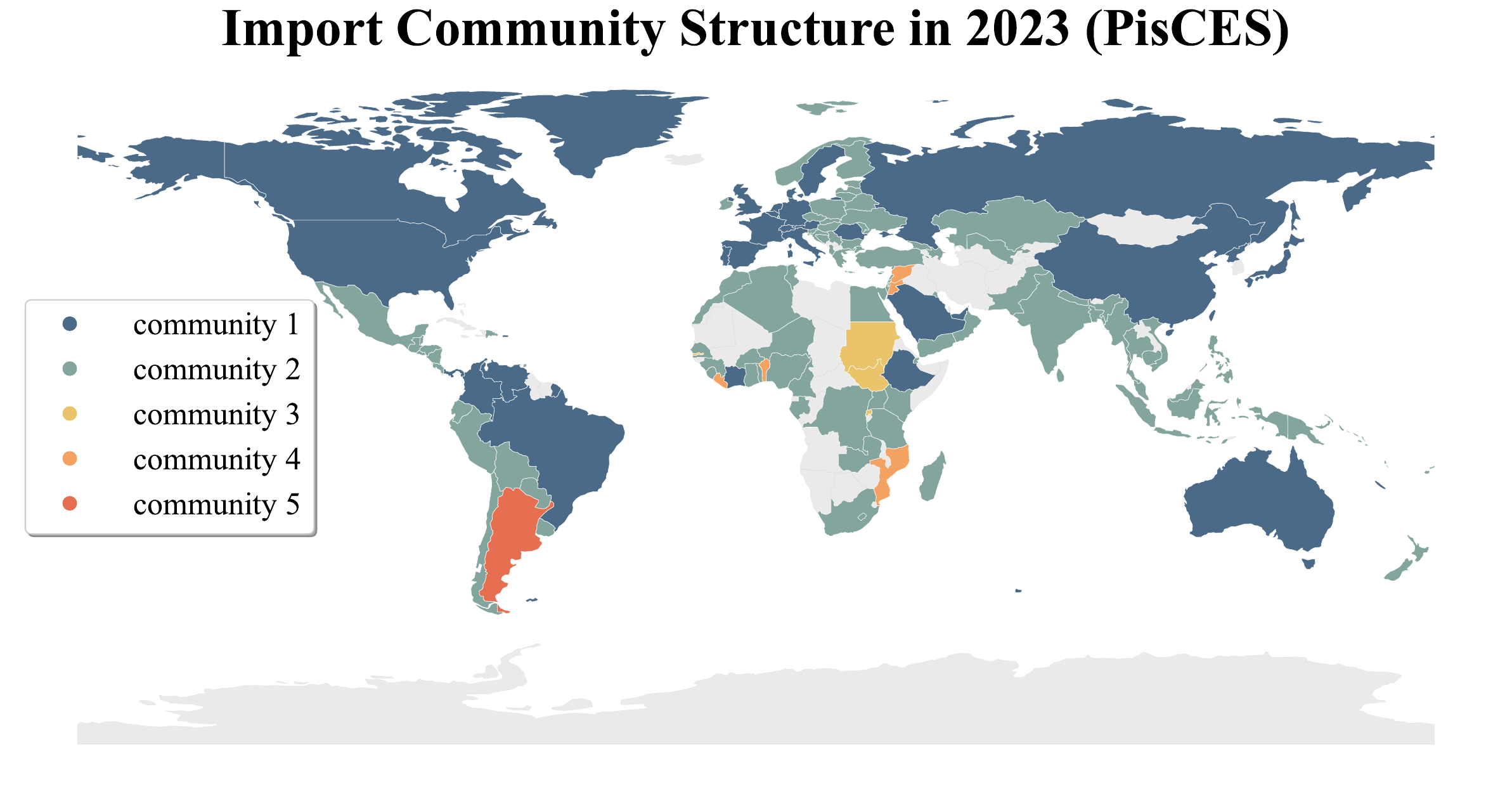}}
\caption{Import community structures of the olive-oil trade network identified by PisCES from 2019 to 2023. Countries assigned to the same import community are marked with the same color.}
\label{fig:pisces_import}
\end{figure*}

\begin{figure*}[!h]{}
\centering
\subfigure[2019]{\includegraphics[height=2.6cm,width=4.5cm,angle=0]{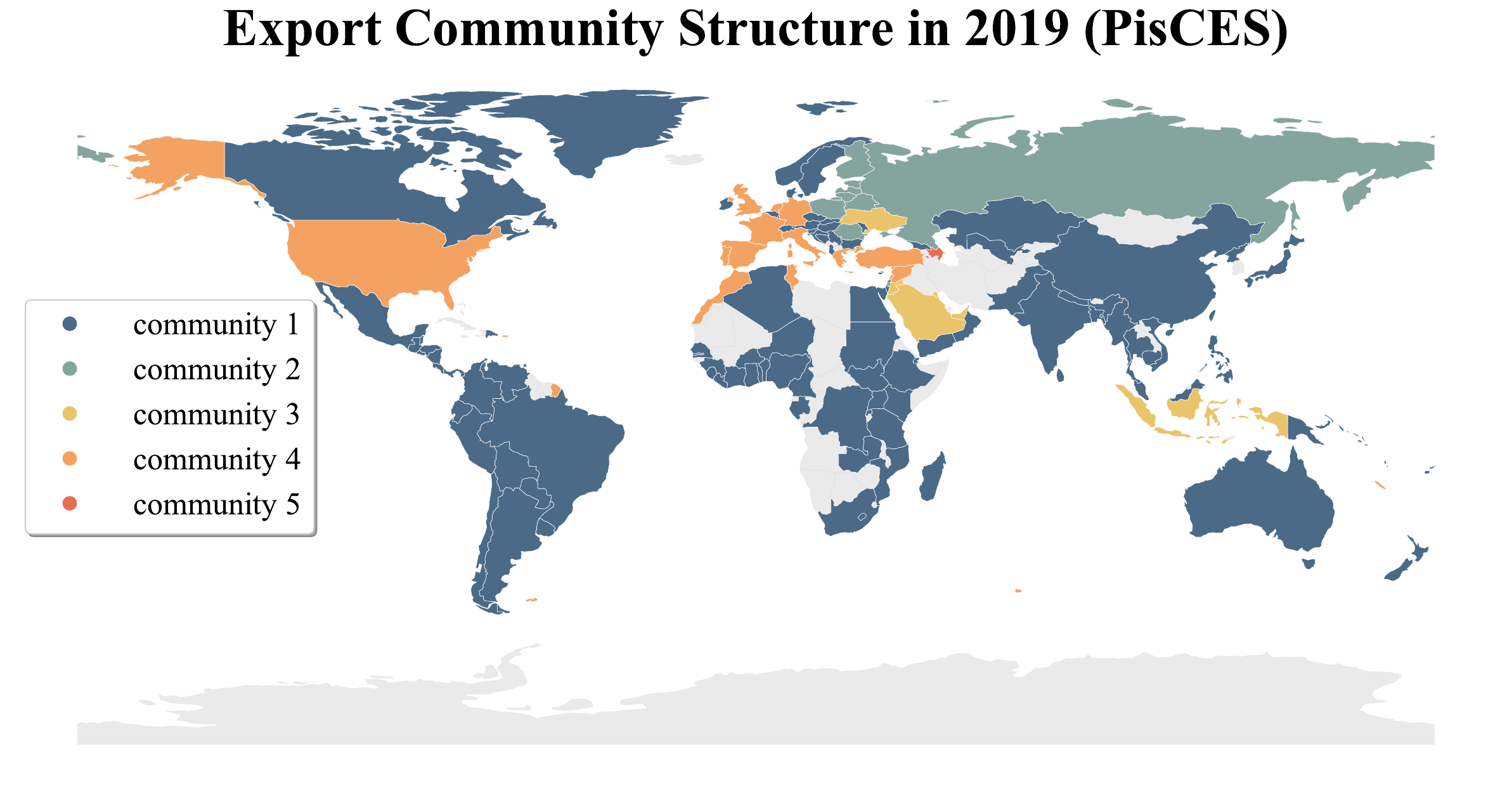}}
\subfigure[2020]{\includegraphics[height=2.6cm,width=4.5cm,angle=0]{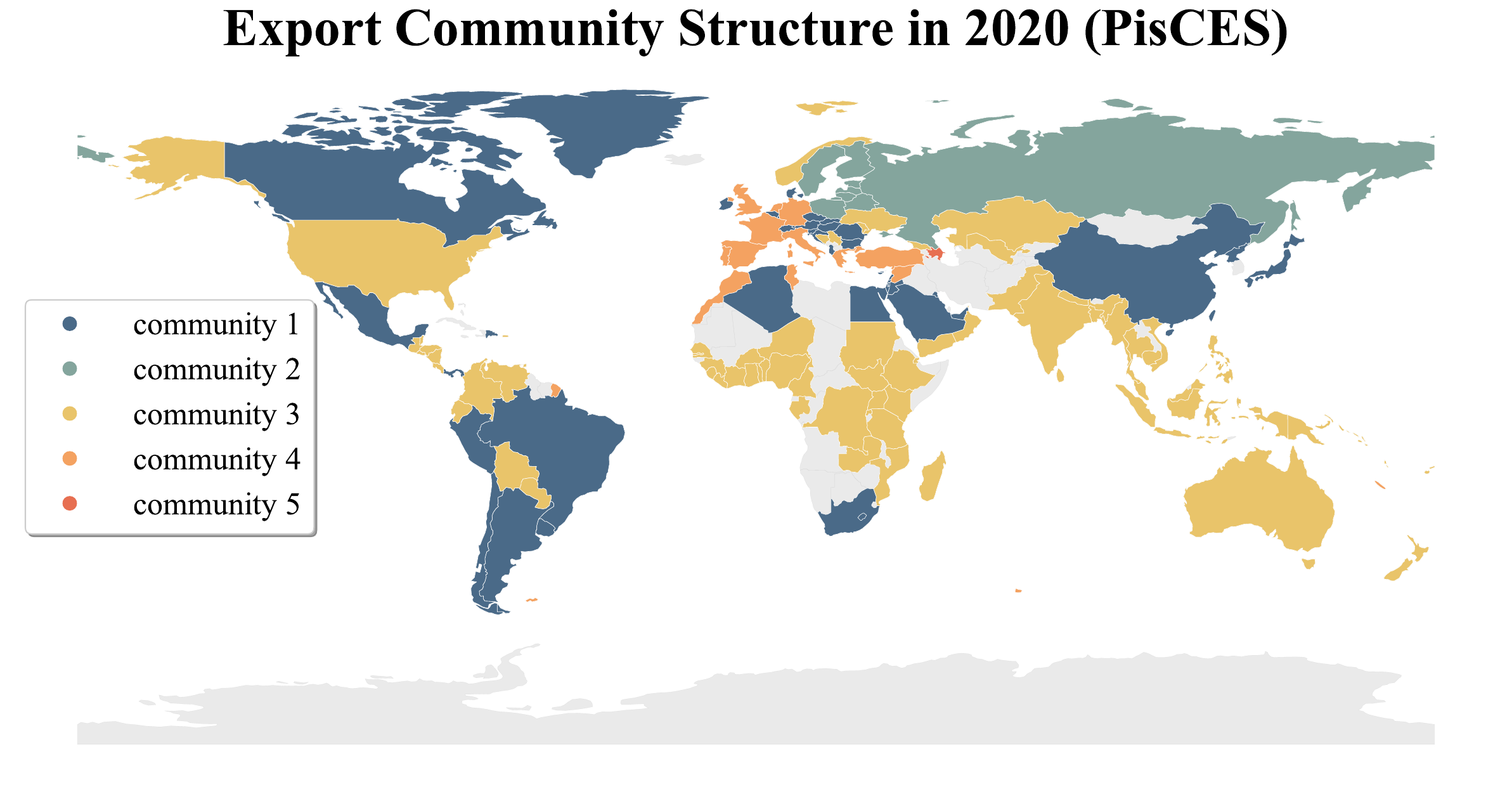}}
\subfigure[2021]{\includegraphics[height=2.6cm,width=4.5cm,angle=0]{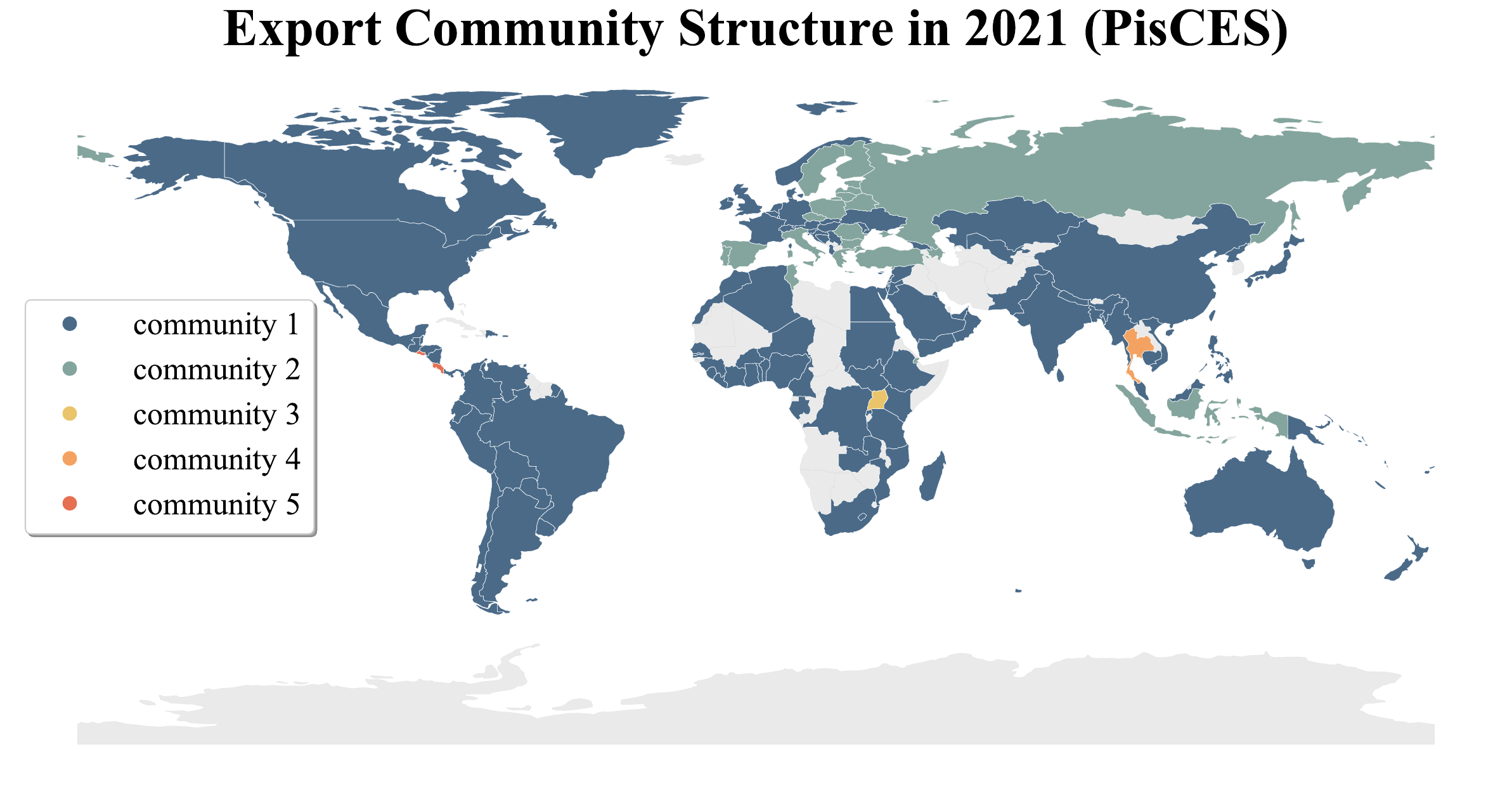}}
\subfigure[2022]{\includegraphics[height=2.6cm,width=4.5cm,angle=0]{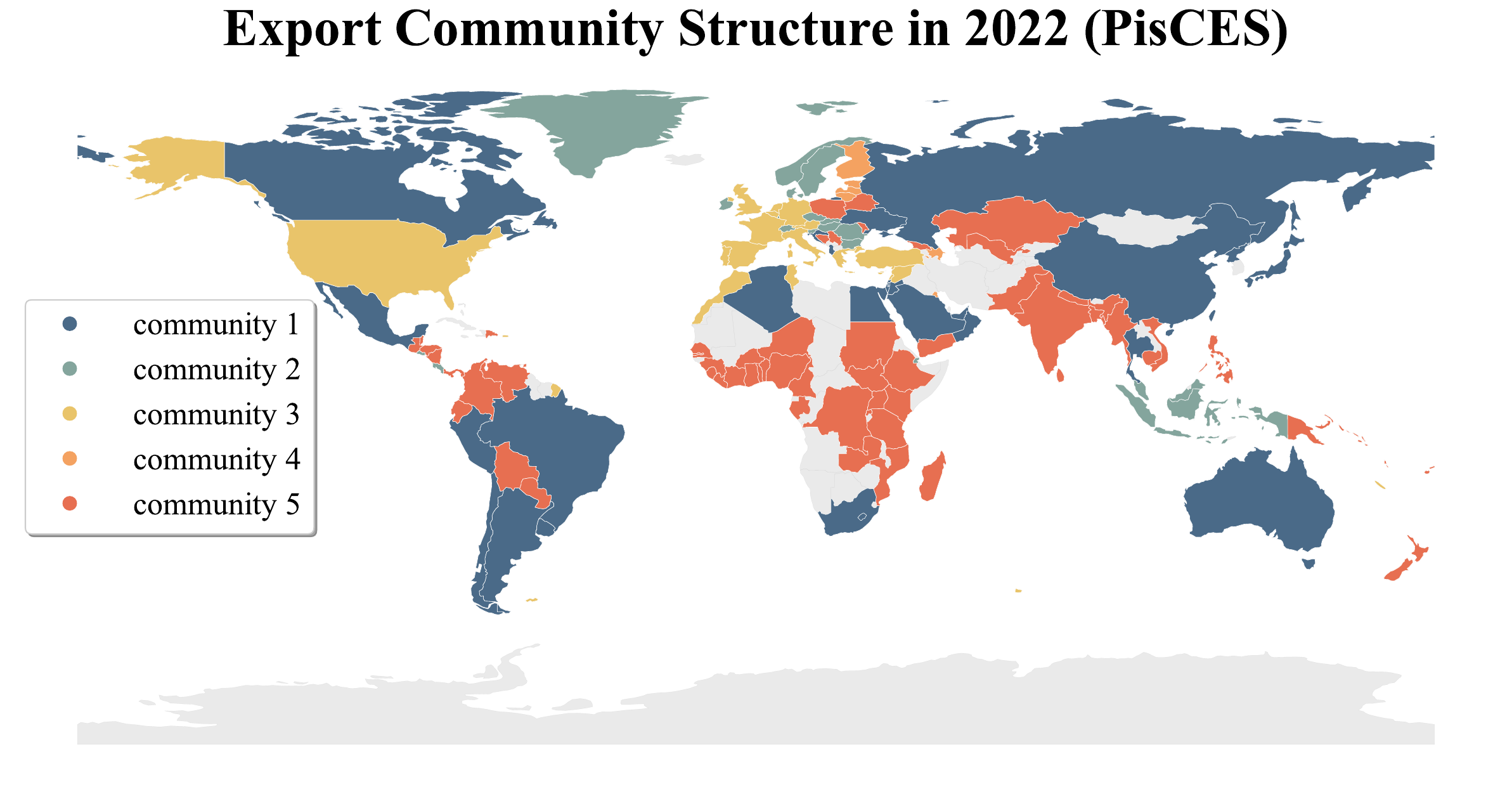}}
\subfigure[2023]{\includegraphics[height=2.6cm,width=4.5cm,angle=0]{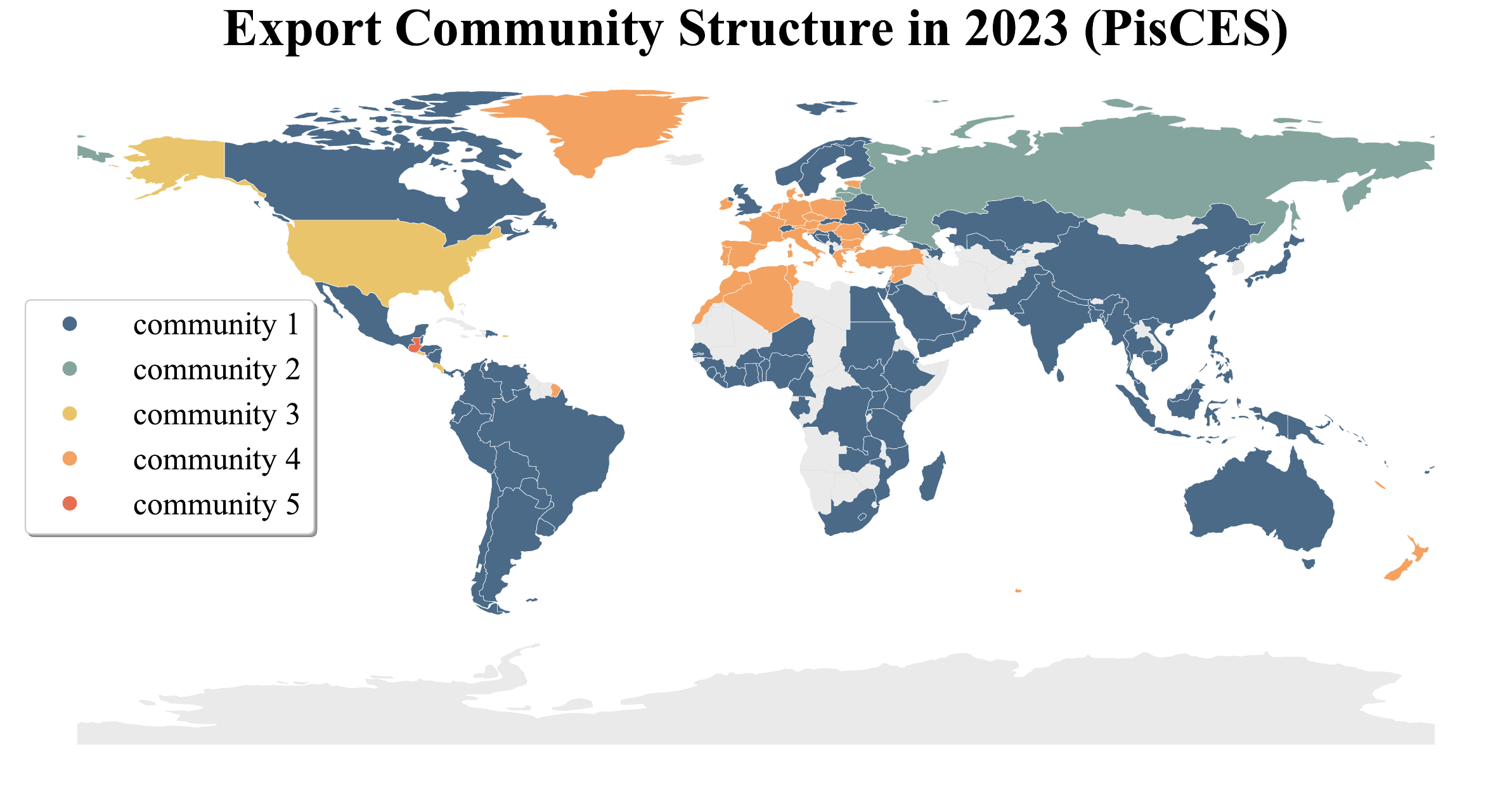}}
\caption{Export community structures of the olive-oil trade network identified by PisCES from 2019 to 2023. Countries assigned to the same export community are marked with the same color.}
\label{fig:pisces_export}
\end{figure*}

\section{Sensitivity analysis of tuning parameters}
\label{parameter_sensitivity}
We conducted an additional simulation experiment with $N=100$, $M=5$,
$T=5$, and $K=3$. The diagonal and off-diagonal entries of the block
connectivity matrix were set to $0.4$ and $0.1$, respectively. We considered
three levels of community variation, $r\in\{0,0.1,0.2\}$. Both $\alpha$ and
$\beta$ were selected from a predefined candidate set, resulting in
$12\times12$ candidate parameter combinations. For each value of $r$, we
performed five-fold network cross-validation and averaged the results over
ten independent replications.

Fig.~\ref{fig:parameter_sensitivity} presents the mean validation negative
log-likelihood under different combinations of
$\alpha$ and $\beta$. Thus, the loss surface directly reflects how the
predictive performance of MuDySC changes as the strengths of temporal and
layer-wise smoothing vary. The blue box in each panel indicates the parameter
combination attaining the smallest mean validation loss.

\begin{figure}[!htbp]
\centering
\includegraphics[height=5cm,width=13cm,angle=0]
{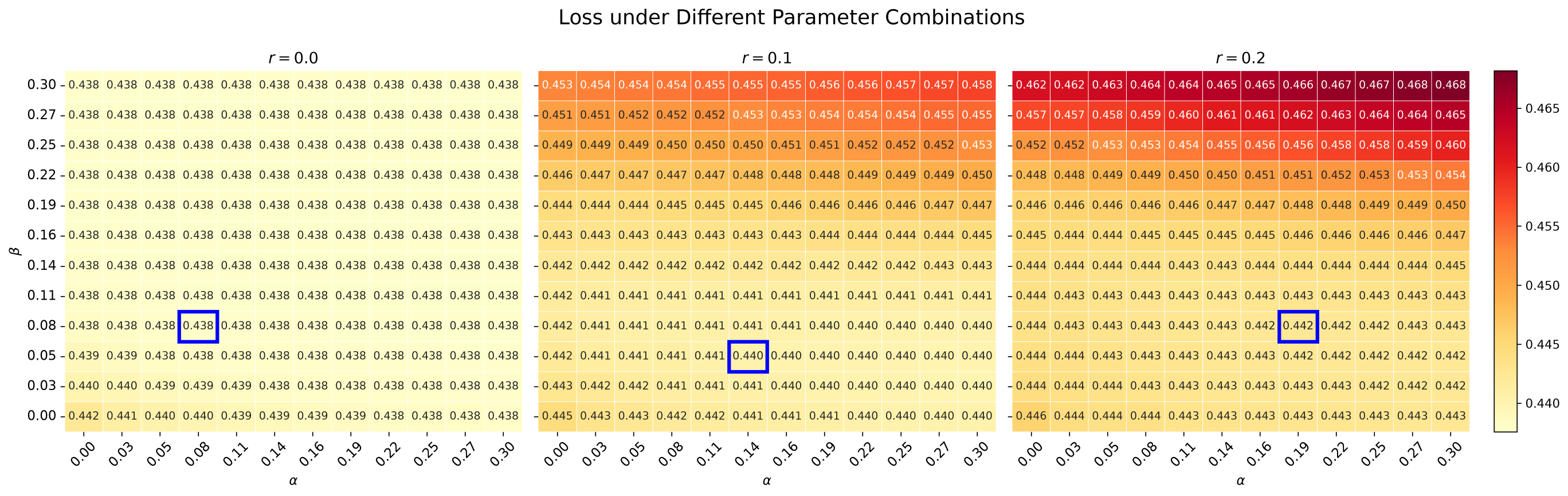}
\caption{Mean validation loss under different combinations of $\alpha$
and $\beta$. The results are averaged over ten independent replications,
and the blue boxes indicate the parameter combinations with the smallest
mean validation losses.}
\label{fig:parameter_sensitivity}
\end{figure}

The cross-validation results are summarized in
Table~\ref{tab:parameter_sensitivity}.

\begin{table}[!h]
\begin{center}
\begin{minipage}{\textwidth}
\caption{Parameter combinations selected by network cross-validation.}
\label{tab:parameter_sensitivity}%
\begin{tabular*}{\linewidth}{@{\extracolsep{\fill}}ccccc@{}}
\hline
$r$ & $\widehat{\alpha}$ & $\widehat{\beta}$
& Mean loss at selected parameters
& Range of mean validation loss \\
\hline
$0.0$ & $0.08$ & $0.08$ & $0.438$ & $[0.438,\,0.442]$ \\
$0.1$ & $0.14$ & $0.05$ & $0.440$ & $[0.440,\,0.458]$ \\
$0.2$ & $0.19$ & $0.08$ & $0.442$ & $[0.442,\,0.468]$ \\
\hline
\end{tabular*}
\end{minipage}
\end{center}
\end{table}

When $r=0$, the community structures are identical across all layers and
time points. In this case, different choices of $\alpha$ and $\beta$ lead
to similar predictive performance, resulting in a relatively flat low-loss
region. Specifically, the mean validation loss ranges from $0.438$ to
$0.442$, corresponding to a range width of
$0.442-0.438=0.004$. As $r$ increases, the heterogeneity of community
structures across layers and time points becomes more pronounced, and the
predictive performance becomes more sensitive to the choice of smoothing
parameters. For example, when $r=0.2$, the mean validation loss ranges from
$0.442$ to $0.468$, corresponding to a substantially larger range width of
$0.468-0.442=0.026$.

These results illustrate the trade-off induced by the two smoothing
parameters. Temporal and layer-wise smoothing allow MuDySC to borrow
information from adjacent time points and related network layers, which can
improve predictive performance when shared structures exist. However,
excessively strong smoothing may obscure genuine differences among community
structures and lead to oversmoothing. Therefore, the appropriate values of
$\alpha$ and $\beta$ depend on the degree of heterogeneity in the
multilayer-dynamic network, and network cross-validation provides a
data-driven strategy for balancing information sharing and structure
preservation.

\section{Proof of Theorem \ref{thm1}}\label{sec8}
We first write the alternating iterative algorithm in terms of the following operator  $G=(G_{m,t})_{M \times T}: \mathbb{R}^{M \times T \times n \times n} \to \mathbb{R}^{M \times T \times n \times n}$, where $G_{m,t}$ corresponds to the iterative update at the $m$th network layer and $t$th time step:  
\begin{equation*}
\begin{aligned}
G_{m,1}(\overline{U}) 
&=
\Pi_K\left( U_{m,1} + \alpha \overline{U}_{m,2} + \frac{\beta}{M-1} \sum_{i \neq m} \overline{U}_{i,1} \right), \quad t = 1, \\
G_{m,t}(\overline{U}) 
&=
\Pi_K\left( U_{m,t} + \alpha \overline{U}_{m,t-1} + \alpha \overline{U}_{m,t+1} + \frac{\beta}{M-1} \sum_{i \neq m} \overline{U}_{i,t} \right), \quad t = 2, \cdots, T-1, \\
G_{m,T}(\overline{U}) 
&=
\Pi_K\left( U_{m,T} + \alpha \overline{U}_{m,T-1} + \frac{\beta}{M-1} \sum_{i \neq m} \overline{U}_{i,T} \right), \quad t = T,
\end{aligned}
\end{equation*}
where $\overline{U} = (\overline{U}_{m,t})_{M \times T}$ and
\begin{equation}
\Pi_K (H) = V_K V_K^T,
\label{eq16}
\end{equation}
provided that $V_K \in \mathbb{R}^{n \times K}$ denotes the top-$K$ eigenvectors of $H$.

To prove Theorem~\ref{thm1}, we shall make use of the following lemmas and theorem.
\begin{definition}[Contraction Mapping]\label{def1}
Let $(\mathcal X, d)$ be a metric space and $G : \mathcal X \to \mathcal X$ be a mapping.  
If there exists a constant $\gamma \in [0,1)$ such that, for all $x, y \in \mathcal X$, $d(G(x), G(y)) \le \gamma d(x, y)$, then $G$ is called a contraction mapping.
\end{definition}

\begin{theorem}[Contraction Mapping Theorem \citep{goebel1990topics}]\label{thm2}
If $G : \mathcal X \to \mathcal X$ is a contraction mapping,  
then there exists a unique fixed point $x^*$ such that $G(x^*) = x^*$.  
Moreover, for any initial point $x_0 \in \mathcal X$, the sequence $\{x_n\}$ generated by $x_{n+1} = G(x_n)$ converges to $x^*$. 
\end{theorem}

\begin{lemma}\label{lemma1}
The global minimizer of the optimization problem in \eqref{eq6}, denoted by $\overline{U}^* = (\overline{U}_{m,t}^*)_{M \times T}$, is a fixed point of the mapping $G$, i.e., $\overline{U}^* = G(\overline{U}^*)$.
\end{lemma}

\begin{lemma}\label{lemma2}
If $\alpha$ and $\beta$ satisfy $2\alpha + \beta < \frac{1}{1 + 2\sqrt{2}}$, then the mapping $G$ is a contraction mapping.
\end{lemma}

The proof of Lemma \ref{lemma1} and Lemma \ref{lemma2} are provide in Appendix \ref{subsec8_1} and \ref{subsec8_2}, respectively. With these lemmas, we are ready to prove Theorem \ref{thm1}.

\emph{Proof of Theorem \ref{thm1}}: 
First, the optimization problem in \eqref{eq6} attains a global minimum, because the feasible region is bounded and the objective function is bounded below by 0.
By Lemma~\ref{lemma1}, the global minimizer of the optimization problem in~\eqref{eq6} is a fixed point of the mapping $G$. By Lemma~\ref{lemma2}, $G$ is a contraction mapping when $2\alpha + \beta < \frac{1}{1 + 2\sqrt{2}}$ satisfies the required constraint condition. Finally, by Theorem~\ref{thm2}, the mapping $G$ has a unique fixed point, and the sequence generated by the iterative algorithm converges to this fixed point. Therefore, it follows that the MuDySC method converges to the global minimum of the optimization problem in~\eqref{eq6}.
\hfill $\square$

\subsection{Proof of Lemma \ref{lemma1}}\label{subsec8_1}

The global minimizer $ \overline{U}^*$ of the optimization problem in~\eqref{eq6} satisfies:
\begin{equation}
\begin{aligned}
\overline{U}_{m,1}^* 
&=
\arg\min_ {\overline{U} \in \mathcal{U}} \left\| U_{m,1} 
- \overline{U}_{m,1} \right\|_F^2 
+ \alpha \left\| \overline{U}_{m,1} 
- \overline{U}_{m,2}^* \right\|_F^2
+ \frac{\beta}{M-1} \sum_{i \neq m} \left\| \overline{U}_{m,1} - \overline{U}_{i,1}^* \right\|_F^2,\; t=1, \\
\overline{U}_{m,t}^*
&=
\arg\min_ {\overline{U} \in \mathcal{U}} \left\| U_{m,t} 
- \overline{U}_{m,t} \right\|_F^2 
+ \alpha \left\| \overline{U}_{m,t} 
- \overline{U}_{m,t+1}^* \right\|_F^2 
+ \alpha \left\| \overline{U}_{m,t-1}^* 
- \overline{U}_{m,t} \right\|_F^2 \\
& \quad 
+ \frac{\beta}{M-1} \sum_{i \neq m} \left\| \overline{U}_{m,t} 
- \overline{U}_{i,t}^* \right\|_F^2,\; t = 2, \cdots, T-1, \\
\overline{U}_{m,T}^*
&=
\arg\min_ {\overline{U} \in \mathcal{U}} \left\| U_{m,T} 
- \overline{U}_{m,T} \right\|_F^2 
+ \alpha \left\| \overline{U}_{m,T-1}^* 
- \overline{U}_{m,T} \right\|_F^2 
+ \frac{\beta}{M-1} \sum_{i \neq m} \left\| \overline{U}_{m,T} - \overline{U}_{i,T}^* \right\|_F^2,\; t=T. 
\end{aligned}
\label{eq17}
\end{equation}

Next, we prove for $1<t<T$. Using the identities $\|M\|_F^2 = \text{Tr}(M^T M)$ and $\|U\|_F^2 = K$, \eqref{eq17} can be rewritten as
\begin{align}
\overline{U}_{m,t}^* &= \arg\max_{\overline{U} \in \mathcal{U}} \text{Tr}\left( \overline{U}_{m,t}^T U_{m,t} \right) + \alpha \text{Tr}\left( \overline{U}_{m,t}^T \overline{U}_{m,t+1}^* \right) + \alpha \text{Tr}\left( \overline{U}_{m,t}^T \overline{U}_{m,t-1}^* \right) \nonumber\\
& \quad + \frac{\beta}{M-1} \sum_{i \neq m} \text{Tr}\left( \overline{U}_{m,t}^T \overline{U}_{i,t}^* \right) \nonumber\\
&= \arg\max_{\overline{U} \in \mathcal{U}} \text{Tr}\left( \overline{U}_{m,t}^T \left( U_{m,t} + \alpha \overline{U}_{m,t+1}^* + \alpha \overline{U}_{m,t-1}^* + \frac{\beta}{M-1} \sum_{i \neq m} \overline{U}_{i,t}^* \right) \right).
\label{eq18}
\end{align}
Let $\overline{U}^* = \overline{V}^* \overline{V}^{*T}$, then~\eqref{eq18} can be rewritten as
\begin{align*}
\overline{V}^* &= \arg\max_{\overline{V} \in \mathcal{V}} \text{Tr}\left( \overline{V}^T \left( U_{m,t} + \alpha \overline{U}_{m,t+1}^* + \alpha \overline{U}_{m,t-1}^* + \frac{\beta}{M-1} \sum_{i \neq m} \overline{U}_{i,t}^* \right) \overline{V} \right) \nonumber\\
&:= \arg\max_{\overline{V} \in \mathcal{V}} \text{Tr}\left( \overline{V}^T X \overline{V} \right).
\end{align*}
Since $\overline{U}^*$ is a positive semi-definite matrix, $X$ is also positive semi-definite, and $\overline{V}^T X \overline{V}$ can be regarded as an eigen-decomposition. To maximize $\text{Tr}(\overline{V}^T X \overline{V})$, $\overline{V}$ must be the standard orthonormal eigenvectors corresponding to the first $K$ eigenvalues of $X$, namely,
\begin{equation}
\label{eq:v}
\overline{V}^* = \text{Eigvec}_K\left( U_{m,t} + \alpha \overline{U}_{m,t+1}^* + \alpha \overline{U}_{m,t-1}^* + \frac{\beta}{M-1} \sum_{i \neq m} \overline{U}_{i,t}^* \right),
\end{equation}
where $\text{Eigvec}_K(\cdot)$ denotes an $n \times K$ matrix whose columns are the first $K$ eigenvectors. \eqref{eq:v} also implies that
\begin{align*}
\overline{U}_{m,t}^* &= \Pi_K\left( U_{m,t} + \alpha \overline{U}_{m,t+1}^* + \alpha \overline{U}_{m,t-1}^* + \frac{\beta}{M-1} \sum_{i \neq m} \overline{U}_{i,t}^* \right) = G_{m,t}(\overline{U}^*).
\end{align*}
Using similar arguments, we can show that $\overline{U}^* = G(\overline{U}^*)$ also holds for $t = 1$ and $t = T$. Therefore, the global minimizer $\overline{U}^*$ is a fixed point of the mapping $G$. The proof is completed. 
\hfill $\square$

\subsection{Proof of Lemma \ref{lemma2}}\label{subsec8_2}

Let $\overline{U}, \overline{U}' \in \mathbb{R}^{M \times T \times K}$. To prove $G$ is a contraction mapping, we need to show
\begin{equation*}
\left\| G(\overline{U}) - G(\overline{U}') \right\|_F \leq \gamma \left\| \overline{U} - \overline{U}' \right\|_F, \, {\rm where}\, \gamma \in [0, 1).
\end{equation*}
For $t = 2, \cdots, T-1$, denote
\begin{align*}
\Sigma_{m,t}& = U_{m,t} + \alpha \overline{U}_{m,t+1} + \alpha \overline{U}_{m,t-1} + \frac{\beta}{M-1} \sum_{i \neq m} \overline{U}_{i,t};
\end{align*}
\begin{align*}
\Sigma_{m,t}' &= U_{m,t} + \alpha \overline{U}_{m,t+1}' + \alpha \overline{U}_{m,t-1}' + \frac{\beta}{M-1} \sum_{i \neq m} \overline{U}_{i,t}'.
\end{align*}
Then we have $ G_{m,t}(\overline{U})=\Pi_K(\Sigma_{m,t})$.

Let $V_{m,t}$ and $V_{m,t}'$ consist of the top-$K$ eigenvectors of $\Sigma_{m,t}$ and $\Sigma_{m,t}'$, respectively. Then we can write
\begin{equation*}
\Pi_K(\Sigma_{m,t}) = V_{m,t} V_{m,t}^T\quad {\rm and} \quad \Pi_K(\Sigma_{m,t}') = V_{m,t}' V_{m,t}'^T,
\end{equation*}
and thus the following arguments hold, 
\begin{align}
\label{eq:G}
\left\| G_{m,t}(\overline{U}) - G_{m,t}(\overline{U}') \right\|_F =& \left\| \Pi_K(\Sigma_{m,t}) - \Pi_K(\Sigma_{m,t}') \right\|_F \nonumber\\
=& \left\| V_{m,t} V_{m,t}^T - V_{m,t}' V_{m,t}'^T \right\|_F \nonumber\\
=& \left\| V_{m,t} R R^T V_{m,t}^T - V_{m,t}' V_{m,t}'^T \right\|_F \nonumber\\
=& \left\| (V_{m,t} R - V_{m,t}') R^T V_{m,t}^T + V_{m,t}' (V_{m,t} R - V_{m,t}')^T \right\|_F \nonumber\\
\leq & \left\| (V_{m,t} R - V_{m,t}') R^T V_{m,t}^T \right\|_F + \left\| V_{m,t}' (V_{m,t} R - V_{m,t}')^T \right\|_F \nonumber\\
=&2 \left\| V_{m,t} R - V_{m,t}' \right\|_F\nonumber\\
\leq & \frac{2\sqrt{2}}{\delta} \left\| \Sigma_{m,t} - \Sigma_{m,t}' \right\|_F,
\end{align}
where $R\in \mathbb R^{K\times K}$ is an orthogonal matrix satisfying $RR^T=I_K$, and $R$ comes from the Davis-Kahan theorem (see Lemma \ref{lemma3}), 
the last inequality follows from the Davis-Kahan theorem and $\delta=\lambda_K(\Sigma_{m,t})-\lambda_{K+1}(\Sigma_{m,t}')$ with $\lambda_i(A)$ denoting the $i$th eigenvalue of the matrix $A$.

Next, we show that $\delta \geq 1-2\alpha-\beta$. Let us denote 
$\Sigma_{m,t}:= U_1 + \alpha (U_2 + U_3) + \frac{\beta}{M-1} U_4$, 
where $U_1:= U_{m,t}$, $U_2:= \overline{U}_{m,t+1}$, 
$U_3:=\overline{U}_{m,t-1}$, and 
$U_4:= \sum_{i \neq m} \overline{U}_{i,t}$. We then have 
\begin{align*} 
\lambda_K\left(\Sigma_{m,t} \right) \nonumber 
&\stackrel{(a)}{\geq} \lambda_K(U_1) \stackrel{(b)}{=} 1, 
\end{align*} 
where (a) follows from Lemma~\ref{lemma5}, noting that 
$\alpha(U_2+U_3)+\frac{\beta}{M-1}U_4$ is positive semidefinite, 
and (b) follows from the fact that $U_1$ is a projection matrix whose first $K$ eigenvalues are 1. 
Similarly, let
\[
\Sigma'_{m,t}:= U_1 + \alpha (U'_2 + U'_3) 
+ \frac{\beta}{M-1} U'_4,
\]
where $U'_2:= \overline{U}'_{m,t+1}$, 
$U'_3:=\overline{U}'_{m,t-1}$, and 
$U'_4:= \sum_{i \neq m} \overline{U}'_{i,t}$. We have
\begin{align*} 
 \lambda_{K+1}\left(\Sigma'_{m,t}\right) \nonumber 
&\stackrel{(a)}{\leq} \lambda_{K+1}(U_1) 
+ \lambda_1\left(\alpha (U'_2 + U'_3) 
+ \frac{\beta}{M-1} U'_4 \right) \nonumber\\ 
&\stackrel{(b)}{\leq} \lambda_{K+1}(U_1) 
+ \alpha \lambda_1(U'_2) 
+ \alpha \lambda_1(U'_3) 
+ \frac{\beta}{M-1} \lambda_1(U'_4) \nonumber\\ 
&\stackrel{(c)}{\leq} 0 + \alpha + \alpha + \beta 
= 2\alpha + \beta, 
\end{align*} 
where (a) and (b) follow from Weyl's inequality in Lemma~\ref{lemma4}, 
and (c) follows from the fact that the $U$'s are projection matrices with their first $K$ eigenvalues being 1 and the remaining eigenvalues being 0. Therefore, we have proved that 
\begin{equation} 
\label{eq:delta} 
\delta 
= \lambda_K(\Sigma_{m,t})-\lambda_{K+1}(\Sigma'_{m,t})
\geq 1 - 2\alpha - \beta. 
\end{equation}

Let $\Delta_{m,t} = \overline{U}_{m,t} - \overline{U}'_{m,t}$ for $m\in\{1,...,M\}$ and $t\in\{1,...,T\}$, then combining \eqref{eq:G} with \eqref{eq:delta}, we have for $t\in\{2,...,T-1\}$ that
\begin{equation}
\left\| G_{m,t}(\overline{U}) - G_{m,t}(\overline{U}') \right\|_F \leq \frac{2\sqrt{2}}{1 - 2\alpha - \beta} \left\| \alpha \Delta_{m,t-1} + \alpha \Delta_{m,t+1} + \frac{\beta}{M-1} \sum_{i \neq m} \Delta_{i,t} \right\|_F.
\label{eq34}
\end{equation}
Similarly, when $t = 1$ and $t = T$, we can obtain
\begin{equation*}
\left\| G_{m,1}(\overline{U}) - G_{m,1}(\overline{U}') \right\|_F \leq \frac{2\sqrt{2}}{1 - 2\alpha - \beta} \left\| \alpha \Delta_{m,2} + \frac{\beta}{M-1} \sum_{i \neq m} \Delta_{i,1} \right\|_F,
\end{equation*}
and
\begin{equation}
\left\| G_{m,T}(\overline{U}) - G_{m,T}(\overline{U}') \right\|_F \leq \frac{2\sqrt{2}}{1 - 2\alpha - \beta} \left\| \alpha \Delta_{m,T-1} + \frac{\beta}{M-1} \sum_{i \neq m} \Delta_{i,T} \right\|_F \label{eq36}.
\end{equation}
Summing \eqref{eq34}-\eqref{eq36} across all $t$'s and $m$'s, we obtain
\begin{align*}
&\sum_{m=1}^M \sum_{t=1}^T \left\| G_{m,t}(\overline{U}) - G_{m,t}(\overline{U}') \right\|_F \nonumber\\
&\leq \frac{2\sqrt{2}}{1 - 2\alpha - \beta} \sum_{m=1}^M \sum_{t=1}^T \left\| \alpha (\Delta_{m,t-1} + \Delta_{m,t+1}) + \frac{\beta}{M-1} \sum_{i \neq m} \Delta_{i,t} \right\|_F \nonumber\\
&\leq \frac{2\sqrt{2}}{1 - 2\alpha - \beta} \left( \alpha \sum_{m=1}^M \sum_{t=1}^T \left\| \Delta_{m,t-1} + \Delta_{m,t+1} \right\|_F + \frac{\beta}{M-1} \sum_{m=1}^M \sum_{t=1}^T \left\| \sum_{i \neq m} \Delta_{i,t} \right\|_F \right) \nonumber\\
&\leq \frac{4\sqrt{2}\alpha + 2\sqrt{2}\beta}{1 - 2\alpha - \beta} \left( \sum_{m=1}^M \sum_{t=1}^T \left\| \overline{U}_{m,t} - \overline{U}_{m,t}' \right\|_F \right).
\end{align*}
Finally, by the definition of a contraction mapping, when $\frac{4\sqrt{2}\alpha + 2\sqrt{2}\beta}{1 - 2\alpha - \beta} < 1 $, i.e., $2\alpha + \beta < \frac{1}{1 + 2\sqrt{2}}$, the mapping $G$ is a contraction mapping under the metric $d(\overline{U}_{m,t}, \overline{U}_{m,t}') = \sum_{m=1}^M \sum_{t=1}^T \left\| \overline{U}_{m,t} - \overline{U}_{m,t}' \right\|_F$. The proof is completed.
\hfill $\square$

\subsection{Auxiliary lemmas}\label{subsec8_3}


\begin{lemma}[Davis--Kahan Theorem \citep{davis1970rotation}]
\label{lemma3}
Let $\Sigma,\Sigma' \in \mathbb{R}^{n\times n}$ be symmetric matrices, and let
$V,V' \in \mathbb{R}^{n\times K}$ contain the eigenvectors corresponding to the
top $K$ eigenvalues of $\Sigma$ and $\Sigma'$, respectively. Suppose that
\[
\delta := \lambda_K(\Sigma)-\lambda_{K+1}(\Sigma')>0.
\]
Then, there exists an orthogonal matrix
$R\in\mathbb{R}^{K\times K}$ such that
\[
\|VR-V'\|_F
\leq
\frac{\sqrt{2}}{\delta}\|\Sigma-\Sigma'\|_F.
\]
\end{lemma}

\begin{lemma}[Weyl's Inequality \citep{horn2012matrix}]\label{lemma4}
For two real symmetric matrices $A, B \in \mathbb{R}^{n \times n}$, let $\lambda_i(A)$ denote the $i$th largest eigenvalue of $A$, and define $\lambda_i(B)$ similarly.
the following inequalities hold for any integers $i,j \in \{1,2,\cdots,n\}$,
\begin{equation*}
\lambda_{i+j-1}(A+B) \leq \lambda_i(A) + \lambda_j(B)\quad {\rm for}\quad 1 \leq i+j-1 \leq n;
\end{equation*}
\begin{equation*}
\lambda_i(A) + \lambda_j(B) \leq \lambda_{i+j-n}(A+B) \quad {\rm for}\quad i+j-n \geq 1.
\end{equation*}
\end{lemma}
\begin{lemma}[Monotonicity of Eigenvalues \citep{horn2012matrix}]\label{lemma5}
Let $A, B \in \mathbb{R}^{n \times n}$ be symmetric matrices and suppose $B$ is positive semidefinite. Then
\begin{equation*}
\lambda_i(A) \leq \lambda_i(A + B), \qquad i = 1, \ldots, n.
\end{equation*}
\end{lemma}


\section{Proof of Theorem~\ref{thm:local}}
\label{sec9}
We first prove part~\textnormal{(a)}. Recall the iterative mapping
$G=(G_{m,t})_{M\times T}$ defined in Appendix~\ref{sec8}.

Since $\overline U^\dagger$ is a local minimizer of \eqref{eq6},
fixing all the other projection matrices at their values in
$\overline U^\dagger$, the optimization with respect to
$\overline U_{m,t}$ is equivalent to maximizing
\[
\operatorname{Tr}
\left(
\overline U_{m,t}\Sigma_{m,t}^\dagger
\right)
\]
over rank-$K$ projection matrices. Hence,
$\overline U_{m,t}^\dagger$ is a local maximizer of this blockwise
problem. This trace maximization problem has no non-global local
maximizer. Indeed, if a selected eigendirection corresponds to a
smaller eigenvalue than an unselected eigendirection, a small rotation
toward the latter strictly increases the objective value. Therefore,
any local maximizer must be the projection onto the top-$K$ eigenspace
of $\Sigma_{m,t}^\dagger$. Moreover,
condition~\eqref{eq:local-gap-condition} implies
$\Delta^\dagger>0$, so this top-$K$ projection is uniquely defined.
Thus,
\[
\overline U_{m,t}^\dagger
=
\Pi_K(\Sigma_{m,t}^\dagger)
=
G_{m,t}(\overline U^\dagger),
\]
for all $m$ and $t$. Hence,
\[
G(\overline U^\dagger)=\overline U^\dagger.
\]
Now fix any
\[
0<r<
\frac{\Delta^\dagger-2\sqrt{2}(2\alpha+\beta)}
     {2(2\alpha+\beta)},
\]
and consider any
$\overline U,\overline U'\in\mathcal B_r(\overline U^\dagger)$,
where
\[
\mathcal B_r(\overline U^\dagger)
=
\left\{
\overline U:
d(\overline U,\overline U^\dagger)\leq r
\right\}.
\]
For each $(m,t)$, let $\Sigma_{m,t}$ and $\Sigma'_{m,t}$ denote the
matrices appearing in the updates $G_{m,t}(\overline U)$ and
$G_{m,t}(\overline U')$, respectively. For example, when
$2\leq t\leq T-1$,
\[
\Sigma_{m,t}
=
U_{m,t}
+\alpha\overline U_{m,t-1}
+\alpha\overline U_{m,t+1}
+\frac{\beta}{M-1}\sum_{i\neq m}\overline U_{i,t},
\]
and $\Sigma'_{m,t}$ is defined analogously with $\overline U$ replaced
by $\overline U'$.

Since
$d(\overline U,\overline U^\dagger)\leq r$, we have
\[
\|\Sigma_{m,t}-\Sigma_{m,t}^\dagger\|_2
\leq
\|\Sigma_{m,t}-\Sigma_{m,t}^\dagger\|_F
\leq
(2\alpha+\beta)r.
\]
Similarly,
\[
\|\Sigma'_{m,t}-\Sigma_{m,t}^\dagger\|_2
\leq
(2\alpha+\beta)r.
\]
The same bounds hold for $t=1$ and $t=T$.

It then follows from Weyl's inequality in Lemma~\ref{lemma4} that
\begin{align*}
\lambda_K(\Sigma_{m,t})
-\lambda_{K+1}(\Sigma'_{m,t})
&\geq
\lambda_K(\Sigma_{m,t}^\dagger)
-\lambda_{K+1}(\Sigma_{m,t}^\dagger)
-2(2\alpha+\beta)r\\
&\geq
\Delta^\dagger-2(2\alpha+\beta)r.
\end{align*}
Applying the Davis--Kahan theorem in Lemma~\ref{lemma3} and following
the same argument as in \eqref{eq:G}, we obtain
\[
\|G_{m,t}(\overline U)-G_{m,t}(\overline U')\|_F
\leq
\frac{2\sqrt{2}}
{\Delta^\dagger-2(2\alpha+\beta)r}
\|\Sigma_{m,t}-\Sigma'_{m,t}\|_F.
\]
Summing over all $m$ and $t$ and proceeding as in the proof of
Lemma~\ref{lemma2} gives
\[
d\bigl(G(\overline U),G(\overline U')\bigr)
\leq
\rho_r\,d(\overline U,\overline U'),
\]
where
\[
\rho_r
=
\frac{2\sqrt{2}(2\alpha+\beta)}
{\Delta^\dagger-2(2\alpha+\beta)r}.
\]
By the choice of $r$,
\[
\rho_r<1.
\]
Therefore, $G$ is contractive on
$\mathcal B_r(\overline U^\dagger)$.

Moreover, since $G(\overline U^\dagger)=\overline U^\dagger$, for any
$\overline U\in\mathcal B_r(\overline U^\dagger)$,
\[
d\bigl(G(\overline U),\overline U^\dagger\bigr)
=
d\bigl(G(\overline U),G(\overline U^\dagger)\bigr)
\leq
\rho_r d(\overline U,\overline U^\dagger)
\leq \rho_r r<r.
\]
Thus, $G$ maps $\mathcal B_r(\overline U^\dagger)$ into itself.
The contraction mapping theorem therefore implies that, if the
initialization belongs to $\mathcal B_r(\overline U^\dagger)$, the
sequence generated by \eqref{eq7} converges to
$\overline U^\dagger$.

We next prove part~\textnormal{(b)}. For each $(m,t)$, by the
definition of $\Sigma_{m,t}^\dagger$, we can write
\begin{align*}
\Sigma_{m,t}^\dagger
={}&
(1+q_t\alpha+\beta)\overline U_{m,t}^\dagger
+\left(U_{m,t}-\overline U_{m,t}^\dagger\right)\\
&+
\alpha\sum_{s\in\mathcal N_t}
\left(
\overline U_{m,s}^\dagger-\overline U_{m,t}^\dagger
\right)
+
\frac{\beta}{M-1}\sum_{i\neq m}
\left(
\overline U_{i,t}^\dagger-\overline U_{m,t}^\dagger
\right).
\end{align*}
Under condition~\eqref{eq:subspace-alignment}, the spectral norm of
the last three terms is bounded by
\[
\eta+q_t\alpha\eta+\beta\eta
=
(1+q_t\alpha+\beta)\eta.
\]
Since $\overline U_{m,t}^\dagger$ is a rank-$K$ projection matrix,
the $K$th and $(K+1)$th eigenvalues of
$(1+q_t\alpha+\beta)\overline U_{m,t}^\dagger$ are
$1+q_t\alpha+\beta$ and $0$, respectively. Hence, by Weyl's
inequality,
\[
\lambda_K(\Sigma_{m,t}^\dagger)
\geq
(1+q_t\alpha+\beta)(1-\eta)
\]
and
\[
\lambda_{K+1}(\Sigma_{m,t}^\dagger)
\leq
(1+q_t\alpha+\beta)\eta.
\]
Therefore,
\[
\lambda_K(\Sigma_{m,t}^\dagger)
-\lambda_{K+1}(\Sigma_{m,t}^\dagger)
\geq
(1+q_t\alpha+\beta)(1-2\eta).
\]
Taking the minimum over $m$ and $t$ yields
\[
\Delta^\dagger
\geq
\min_t
\left\{
(1+q_t\alpha+\beta)(1-2\eta)
\right\},
\]
which proves \eqref{eq:local-gap-bound}. Finally,
\eqref{eq:local-sufficient-condition} implies
\eqref{eq:local-gap-condition}, and hence the local contractivity and
convergence results in part~\textnormal{(a)} follow.
\hfill $\square$

\section{TLC-MuDySC: A tensor-based layer-clustering extension}
\label{sec:tensor-extension}

The original MuDySC method borrows information across all network layers at
each time point. While this strategy is effective when the layers share
similar community structures, smoothing across highly heterogeneous layers
may mix information from dissimilar networks and reduce the benefit of
layer-wise information sharing.

To motivate a layer-clustering extension, we consider a structured form of
layer heterogeneity in which the network layers arise from a small number of
latent classes. This assumption is not required by the original MuDySC method,
but serves as a working model for identifying layers that are suitable for
information sharing. Specifically, at each time point, layers belonging to the
same latent class are assumed to share the same node membership structure and
block connectivity pattern. These class-specific structures may evolve over
time, while different latent classes may exhibit distinct community structures
and evolutionary patterns.

Motivated by this structure, we propose Tensor-based Layer-Clustering MuDySC
(TLC-MuDySC). It first estimates the latent layer classes from the layer-mode
spectral structure of a fourth-order adjacency tensor and then applies the
original MuDySC method separately within each estimated class. In this way,
information is mainly borrowed among structurally related layers, reducing
interference from heterogeneous layers. In what follows, we first introduce a
dynamic mixture multilayer model to formalize the latent layer structure,
then present the TLC-MuDySC algorithm, and finally evaluate its performance
through numerical experiments.

\subsection{A dynamic mixture multilayer model and its layer-mode structure}
\label{sec:tensor-model}
To avoid ambiguity in the matrix and tensor notation used in this section,
we adopt the following conventions. For a matrix $X$, $X(i,:)$ and $X(:,j)$
denote its $i$th row and $j$th column, respectively, and $X(i,j)$ denotes
its $(i,j)$th entry. For a fourth-order tensor
$\boldsymbol A\in\mathbb R^{n\times n\times M\times T}$,
$\boldsymbol A(:,:,m,t)$ denotes the $n\times n$ matrix slice obtained by
fixing the third and fourth indices at $m$ and $t$. We use $\operatorname{vec}(X)$ to denote the column-stacking vectorization
of a matrix $X$, and $\operatorname{Mat}_k(\boldsymbol A)$ to denote the
mode-$k$ unfolding of a tensor. In particular, the mode-3 unfolding is
defined such that
\[
\operatorname{Mat}_3(\boldsymbol A)(m,:)
=
\bigl(
\operatorname{vec}\{\boldsymbol A(:,:,m,1)\}^{T},
\ldots,
\operatorname{vec}\{\boldsymbol A(:,:,m,T)\}^{T}
\bigr),
\]
and hence
\[
\operatorname{Mat}_3(\boldsymbol A)
\in\mathbb R^{M\times n^2T}.
\]
The notation $\times_k$ denotes the mode-$k$ tensor--matrix product.

Following the notation in the main paper, consider a multilayer-dynamic
network with $M$ layers, $T$ time points, and $n$ aligned nodes,
\[
\{A_{m,t}:m=1,\ldots,M,\ t=1,\ldots,T\},
\qquad
A_{m,t}\in\{0,1\}^{n\times n}.
\]
Suppose that the $M$ network layers arise from $G$ latent classes. Let
$c_m\in[G]:=\{1,\ldots,G\}$ denote the latent class label of layer $m$.
Define
\[
\mathcal M_g=\{m:c_m=g\},
\qquad
M_g=|\mathcal M_g|,
\qquad g=1,\ldots,G.
\]

For latent class $g$ at time $t$, let
$Z_{g,t}\in\{0,1\}^{n\times K}$ denote the node membership matrix, with
exactly one nonzero entry in each row, and let
$B_{g,t}\in[0,1]^{K\times K}$ denote the corresponding symmetric block connectivity
matrix. Conditional on $c_m=g$, assume that
\begin{equation}
A_{m,t}(i,j)
\mathrel{\mathop{\sim}^{\mathrm{ind}}}
\operatorname{Bernoulli}
\left\{
Z_{g,t}(i,:)B_{g,t}Z_{g,t}(j,:)^{T}
\right\},
\qquad
A_{m,t}(j,i)=A_{m,t}(i,j),
\qquad
1\leq i<j\leq n.
\label{eq:tensor-layer-model}
\end{equation}
Thus, layers belonging to the same latent class share the same population
network structure at each time point, while this structure is allowed to
evolve over time. In particular, the layer labels $c_m$ remain fixed over
the observation period, whereas $Z_{g,t}$ and $B_{g,t}$ may vary with $t$.
Different latent classes may therefore follow different dynamic community
structures. Model~\eqref{eq:tensor-layer-model} can be viewed as a dynamic
extension of the mixture multilayer stochastic block model.

We next show why this model naturally induces a layer-clustering structure.
Define the class- and time-specific population adjacency matrix by
\[
P_{g,t}
=
Z_{g,t}B_{g,t}Z_{g,t}^{T},
\]
and collect the population adjacency matrices into the fourth-order tensor
\[
\boldsymbol P\in\mathbb R^{n\times n\times M\times T},
\qquad
\boldsymbol P(:,:,m,t)=P_{c_m,t}.
\]
Let $Y\in\{0,1\}^{M\times G}$ denote the layer-membership matrix,
\[
Y(m,g)=\mathbf 1\{c_m=g\},
\]
and define $\Theta\in\mathbb R^{G\times n^2T}$ by
\[
\Theta(g,:)
=
\left(
\operatorname{vec}(P_{g,1})^{T},
\ldots,
\operatorname{vec}(P_{g,T})^{T}
\right).
\]
The mode-3 unfolding of the population tensor satisfies
\begin{equation}
\operatorname{Mat}_3(\boldsymbol P)=Y\Theta.
\label{eq:mode3-factorization}
\end{equation}
This factorization makes the connection between the tensor decomposition and
the latent layer classes explicit. Let
\[
D_M
=
Y^TY
=
\operatorname{diag}(M_1,\ldots,M_G),
\qquad
W=YD_M^{-1/2}.
\]
Then $W^TW=I_G$, and~\eqref{eq:mode3-factorization} can be rewritten as
\begin{equation}
\operatorname{Mat}_3(\boldsymbol P)
=
W D_M^{1/2}\Theta.
\label{eq:mode3-normalized}
\end{equation}
Suppose that $\operatorname{rank}(\Theta)=G$, and let
\[
D_M^{1/2}\Theta
=
R\Sigma V^T
\]
be its compact singular value decomposition, where
$R\in\mathbb R^{G\times G}$ is orthogonal. It follows from
\eqref{eq:mode3-normalized} that
\[
\operatorname{Mat}_3(\boldsymbol P)
=
(WR)\Sigma V^T.
\]
Therefore, the matrix of the leading $G$ left singular vectors of
$\operatorname{Mat}_3(\boldsymbol P)$ is
\[
U_3
=
WR
=
YD_M^{-1/2}R.
\]

For any layer $m$ belonging to latent class $g$, the $m$th row of $U_3$ is
\[
U_3(m,:)
=
M_g^{-1/2}e_g^TR,
\qquad m\in\mathcal M_g,
\]
where $e_g$ denotes the $g$th canonical basis vector in $\mathbb R^G$.
Hence, all layers in the same latent class have exactly the same population
spectral representation, whereas layers from different classes correspond to
distinct rows. In particular, the $M$ rows of $U_3$ form exactly $G$
distinct population centers.

This shows that the latent layer classes are directly encoded in the
mode-3 left singular subspace of the population tensor. Consequently,
estimating this subspace from the observed adjacency tensor and applying
$K$-means to its rows provides a natural spectral procedure for recovering
the latent layer labels.

\subsection{The TLC-MuDySC algorithm}
\label{sec:tlc-algorithm}

The TLC-MuDySC procedure consists of two stages. In Stage I, we estimate the
latent layer classes using a regularized tensor decomposition. In Stage II,
we apply the original MuDySC method separately within the estimated layer
classes. The detailed procedure is described below.

To exploit the low-rank structure along the two node modes, define
\[
\mathcal Z
=
\bigl(
Z_{1,1},\ldots,Z_{G,1},
Z_{1,2},\ldots,Z_{G,2},
\ldots,
Z_{1,T},\ldots,Z_{G,T}
\bigr),
\]
let $d=\operatorname{rank}(\mathcal Z)$, and let
$H\in\mathbb R^{n\times d}$ be an orthonormal basis for the column space of
$\mathcal Z$. Then the population tensor introduced in
Section~\ref{sec:tensor-model} admits the representation
\begin{equation}
\boldsymbol P
=
\mathcal C
\times_1 H
\times_2 H
\times_3 W,
\label{eq:four-way-representation}
\end{equation}
for some core tensor
$\mathcal C\in\mathbb R^{d\times d\times G\times T}$.
Here, the time mode is retained without dimension reduction, allowing the
class-specific community structures to evolve over time.

\paragraph{Stage I: Tensor-based layer clustering.}

Collect the observed adjacency matrices into the fourth-order tensor
\[
\boldsymbol A\in\mathbb R^{n\times n\times M\times T},
\qquad
\boldsymbol A(:,:,m,t)=A_{m,t},
\]
where $\boldsymbol A(:,:,m,t)$ denotes the matrix slice obtained by fixing
the third and fourth indices at $m$ and $t$, respectively. Let
$\operatorname{Mat}_k(\cdot)$ denote the mode-$k$ unfolding of a tensor and
let $\times_k$ denote the mode-$k$ tensor--matrix product. In particular,
\[
\operatorname{Mat}_3(\boldsymbol A)
\in\mathbb R^{M\times(n^2T)},
\]
so that each row corresponds to one network layer across all node pairs and
time points. Let $\operatorname{SVD}_q(X)$ denote the matrix whose columns
are the leading $q$ left singular vectors of $X$.

\medskip
\noindent\textit{Step 1: Initialization.}
We first aggregate the observed networks over layers and time:
\[
S=\sum_{m=1}^{M}\sum_{t=1}^{T}A_{m,t}.
\]
The node factor is initialized by
\[
\widehat H^{(0)}
=
\operatorname{SVD}_d(S).
\]
Given $\widehat H^{(0)}$, we initialize the layer factor as
\[
\widehat W^{(0)}
=
\operatorname{SVD}_G
\left[
\operatorname{Mat}_3
\left(
\boldsymbol A
\times_1\widehat H^{(0)T}
\times_2\widehat H^{(0)T}
\right)
\right].
\]

\medskip
\noindent\textit{Step 2: Regularized tensor iterations.}
To stabilize the estimation in sparse networks, we regularize the rows of
the estimated factors before each iteration. For
$X\in\mathbb R^{p\times q}$, define
\[
X^{\star}(i,:)
=
X(i,:)
\min\left\{
1,\frac{\delta}{\|X(i,:)\|_2}
\right\},
\qquad i=1,\ldots,p,
\]
and
\[
\mathcal R_{\delta}(X)
=
\operatorname{SVD}_q(X^{\star}).
\]
At iteration $s$, set
\[
\widetilde H^{(s)}
=
\mathcal R_{\delta_H}\bigl(\widehat H^{(s)}\bigr),
\qquad
\widetilde W^{(s)}
=
\mathcal R_{\delta_W}\bigl(\widehat W^{(s)}\bigr),
\]
where $\delta_H>0$ and $\delta_W>0$ are row-regularization thresholds for
the node and layer factors, respectively.
We then update the node and layer factors alternately:
\begin{align}
\widehat H^{(s+1)}
&=
\operatorname{SVD}_d
\left[
\operatorname{Mat}_1
\left(
\boldsymbol A
\times_2\widetilde H^{(s)T}
\times_3\widetilde W^{(s)T}
\right)
\right],
\label{eq:tensor-update-H}\\
\widehat W^{(s+1)}
&=
\operatorname{SVD}_G
\left[
\operatorname{Mat}_3
\left(
\boldsymbol A
\times_1\widetilde H^{(s)T}
\times_2\widetilde H^{(s)T}
\right)
\right].
\label{eq:tensor-update-W}
\end{align}
The iterations continue until the estimated factor subspaces stabilize or a
prespecified maximum number of iterations is reached.

\medskip
\noindent\textit{Step 3: Layer clustering.}
After convergence, we apply $K$-means with $G$ clusters to the $M$ rows of
$\widehat W$. Denote the resulting layer labels by
$\widehat c_1,\ldots,\widehat c_M$ and define
\[
\widehat{\mathcal M}_g
=
\{m:\widehat c_m=g\},
\qquad
\widehat M_g
=
|\widehat{\mathcal M}_g|,
\qquad
g=1,\ldots,G.
\]

\paragraph{Stage II: Class-wise MuDySC.}

For each estimated layer class, define
\[
\mathcal D_g
=
\left\{
A_{m,t}:
m\in\widehat{\mathcal M}_g,\;
t=1,\ldots,T
\right\}.
\]
We then apply the original MuDySC method separately to each
$\mathcal D_g$:
\begin{equation}
\left\{
\widehat z_{m,t}:
m\in\widehat{\mathcal M}_g,\;
t=1,\ldots,T
\right\}
=
\operatorname{MuDySC}
\left(
\mathcal D_g;K,\alpha_g,\beta_g
\right),
\label{eq:classwise-mudysc}
\end{equation}
where $\alpha_g$ and $\beta_g$ can be selected separately for each estimated
layer class using the network cross-validation procedure in Section 2.2.

Although layers within the same latent class share the same population SBM
at a given time point under model~\eqref{eq:tensor-layer-model}, their
observed adjacency matrices are independent realizations. Thus, their sample
projection matrices need not be identical, and MuDySC continues to borrow
information across both layers and time points to reduce estimation noise.

If $\widehat M_g=1$, the layer-wise smoothing term is omitted and only
temporal smoothing is used. When $G=1$, all layers form a single class and
TLC-MuDySC reduces to the original MuDySC procedure.

\subsection{Numerical evaluation}
\label{sec:tlc-simulation}

We conduct an additional simulation study to evaluate TLC-MuDySC when substantial heterogeneity exists across network layers. We set
\[
n=100,\qquad K=3,\qquad G=2,\qquad T=3,
\]
and vary the number of network layers $M$ from $2$ to $12$. The $M$ layers are generated from two latent layer classes according to model~\eqref{eq:tensor-layer-model}. 

We consider two levels of network density. In the first setting, the within-community and between-community connection probabilities are
\[
p_{\mathrm{in}}=0.30,
\qquad
p_{\mathrm{out}}=0.10,
\]
respectively. In the second setting, we use
\[
p_{\mathrm{in}}=0.25,
\qquad
p_{\mathrm{out}}=0.10.
\]
We compare TLC-MuDySC with MuDySC and PisCES, the latter of which incorporates temporal smoothing but does not borrow information across network layers. Performance is evaluated by the node misclassification rate.

\begin{figure*}[!h]{}
\centering
\subfigure[$p_{\mathrm{in}}=0.3,\ p_{\mathrm{out}}=0.10$]
{\includegraphics[height=5cm,width=6.5cm,angle=0]
{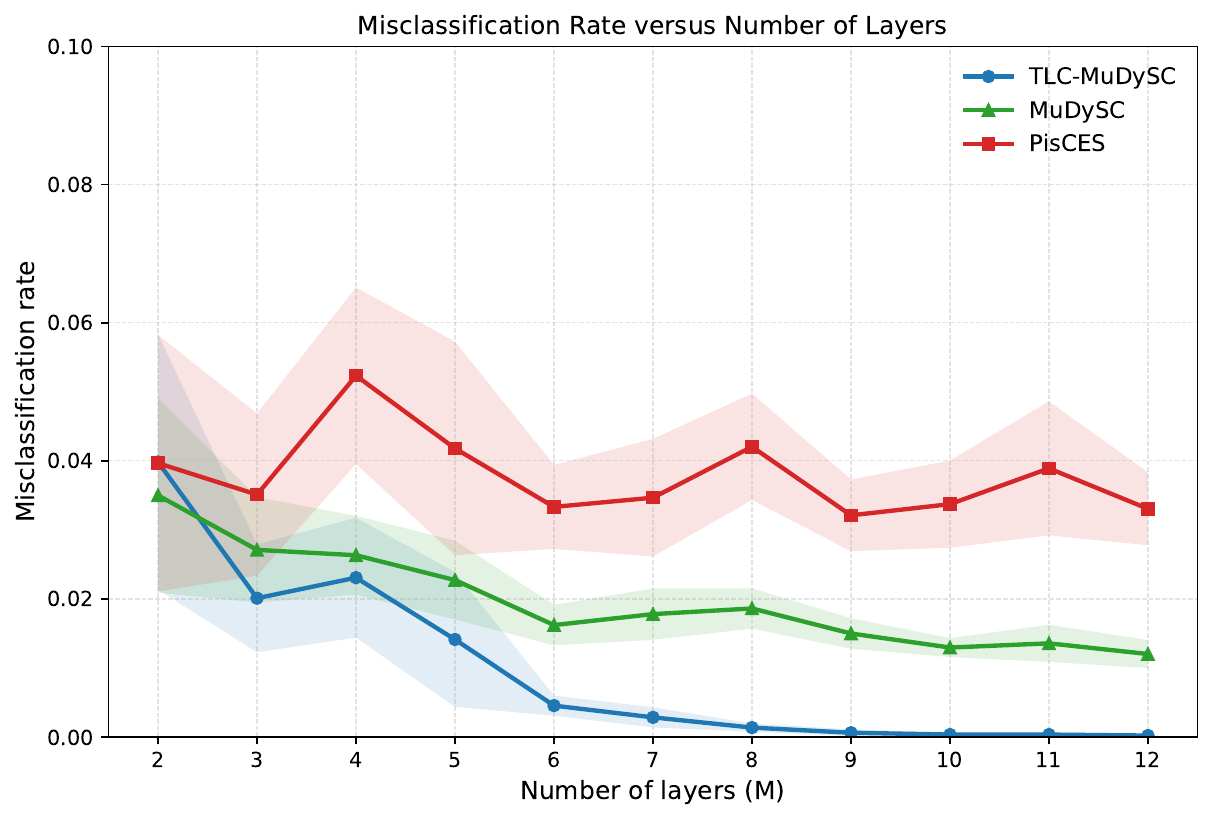}}
\subfigure[$p_{\mathrm{in}}=0.25,\ p_{\mathrm{out}}=0.10$]
{\includegraphics[height=5cm,width=6.5cm,angle=0]
{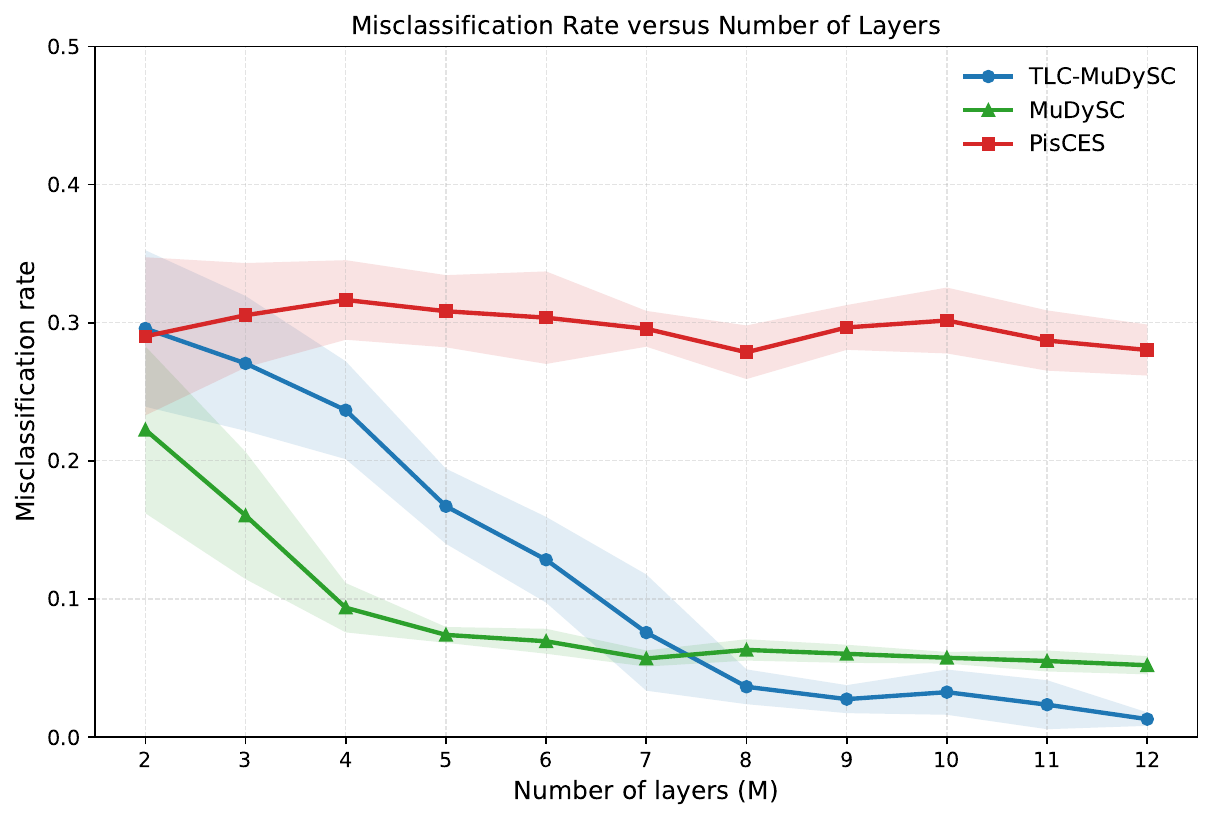}}
\caption{Misclassification rates of TLC-MuDySC, MuDySC, and PisCES as the number of layers $M$ increases under two network settings.}
\label{fig:tlc_mudysc_simulation}
\end{figure*}

The results are shown in Fig.~\ref{fig:tlc_mudysc_simulation}.
Under both network settings, the misclassification rates of TLC-MuDySC and MuDySC show an overall downward trend as $M$ increases, whereas the rate of PisCES remains approximately stable. At larger
values of $M$, TLC-MuDySC achieves the lowest misclassification rate among the three methods. This suggests that, even when substantial heterogeneity exists across the full collection of network layers,
additional layers can provide useful information if sharing is restricted to structurally related layer classes.

TLC-MuDySC first identifies these classes and then applies MuDySC within each class. It can thus use the growing amount of shared layer-wise information while reducing interference from structurally different networks.

\bibliographystyle{plainnat}
\vspace{1cm}
\spacingset{0.8}
\bibliography{reference}
\end{document}